\documentclass[11pt]{article}
\usepackage[letterpaper, margin=1in]{geometry} 

\usepackage{amsmath, amssymb, amsthm}
\usepackage{graphicx}
\usepackage{url}
\usepackage{tikz}
\usepackage{hyperref}
\usepackage{bbm}
\usepackage{braket}
\usepackage{algorithm}
\usepackage{algpseudocode}
\usepackage{fullpage}
\usepackage{natbib}
\usepackage{booktabs}
\usepackage{subcaption}

\newtheorem{thm}{Theorem}[section]
\newtheorem{prop}{Proposition}[section]
\newtheorem{lem}{Lemma}[section]
\newtheorem{cor}{Corollary}[section]
\newtheorem{defi}{Definition}[section]
\newtheorem{rem}{Remark}[section]

\tikzstyle{startstop} = [rectangle, rounded corners, minimum width=3cm, minimum height=1cm, text centered, draw=black, fill=red!30]
\tikzstyle{process} = [rectangle, minimum width=3cm, minimum height=1cm, text centered, draw=black, fill=blue!15]
\tikzstyle{decision} = [diamond, minimum width=2.5cm, minimum height=1cm, text centered, draw=black, fill=green!30]
\tikzstyle{arrow} = [thick,->,>=stealth]
\tikzstyle{line} = [thick,-,>=stealth]

\hypersetup{
    colorlinks=true,
    linkcolor=blue,
    citecolor=blue,
    urlcolor=blue
}
\usepackage{comment}
\usepackage{tikz}
\usepackage{amsmath}
\usepackage{amssymb}  
\usepackage{amsthm}

\usepackage{amsfonts}
\usepackage{graphicx}
\usepackage{bbm}
\usepackage[T1]{fontenc}
\usepackage{url}

\usepackage{braket}
\usepackage{ upgreek }
\usepackage{dsfont}
\usepackage{makecell}
\usepackage{natbib}
\usepackage{fancyref} 
\usepackage{algorithm}
\usepackage{algpseudocode}
\usepackage{fullpage}

\usepackage{float}
\newfloat{algorithm}{t}{lop}

\usepackage{tikz}
\usetikzlibrary{shapes.geometric, arrows, positioning}

\tikzstyle{startstop} = [rectangle, rounded corners, minimum width=3cm, minimum height=1cm, text centered, draw=black, fill=red!30]
\tikzstyle{process} = [rectangle, minimum width=3cm, minimum height=1cm, text centered, draw=black, fill=blue!15]
\tikzstyle{decision} = [diamond, minimum width=2.5cm, minimum height=1cm, text centered, draw=black, fill=green!30]
\tikzstyle{arrow} = [thick,->,>=stealth]
\tikzstyle{line} = [thick,-,>=stealth]

\usepackage{hyperref}

\newcommand{\eee}{\end{aligned}\end{equation*}}
\theoremstyle{plain}

\newcommand*{\fancyrefthmlabelprefix}{thm}
\newcommand*{\fancyreflemlabelprefix}{lem}
\newcommand*{\fancyrefcorlabelprefix}{cor}
\newcommand*{\fancyrefdefilabelprefix}{defi}
\frefformat{plain}{\fancyreflemlabelprefix}{lemma\fancyrefdefaultspacing#1}
\Frefformat{plain}{\fancyreflemlabelprefix}{Lemma\fancyrefdefaultspacing#1}
\frefformat{plain}{\fancyrefthmlabelprefix}{theorem\fancyrefdefaultspacing#1}
\Frefformat{plain}{\fancyrefthmlabelprefix}{Theorem\fancyrefdefaultspacing#1}
\frefformat{plain}{\fancyrefcorlabelprefix}{corollary\fancyrefdefaultspacing#1}
\Frefformat{plain}{\fancyrefcorlabelprefix}{Corollary\fancyrefdefaultspacing#1}
\frefformat{plain}{\fancyrefdefilabelprefix}{definition\fancyrefdefaultspacing#1}
\Frefformat{plain}{\fancyrefdefilabelprefix}{Definition\fancyrefdefaultspacing#1}
\newcommand*{\fancyrefalglabelprefix}{alg}
\newcommand*{\frefalgname}{algorithm}
\newcommand*{\Frefalgname}{Algorithm}
\frefformat{plain}{\fancyrefalglabelprefix}{%
  \frefalgname\fancyrefdefaultspacing#1%
}%
\Frefformat{plain}{\fancyrefalglabelprefix}{%
  \Frefalgname\fancyrefdefaultspacing#1%
}%

\newcommand*{\fancyrefapplabelprefix}{app}
\newcommand*{\frefappname}{appendix}
\newcommand*{\Frefappname}{Appendix}
\frefformat{plain}{\fancyrefapplabelprefix}{%
  \frefappname\fancyrefdefaultspacing#1%
}%
\Frefformat{plain}{\fancyrefapplabelprefix}{%
  \Frefappname\fancyrefdefaultspacing#1%
}%

\definecolor{Green}{HTML}{00AD69}  

\def\beq{\begin{equation}}
\def\eeq{\end{equation}}
\def\bq{\begin{quote}}
\def\eq{\end{quote}}
\def\ben{\begin{enumerate}}
\def\een{\end{enumerate}}
\def\bit{\begin{itemize}}
\def\eit{\end{itemize}}

\def\l|{\left|}
\def\r|{\right|}

\newcommand\C{\mathbbm{C}}

\newcommand\R{\mathbbm{R}}

\newcommand\W{\mathcal{W}}

\newcommand\cB{\mathcal{B}}

\newcommand{\ketbra}[1]{|#1\rangle\langle#1|}
\newcommand{\tr}[1]{\operatorname{tr}\left[#1\right]}

\newcommand{\cH}{\mathcal{H}}

\newcommand{\eps}{\epsilon}

\newcommand{\Tr}{\operatorname{tr}}
\definecolor{daniel}{rgb}{.8,.5,.3}

\definecolor{marco}{rgb}{.4,.2,.6}

\date{\today}

\begin{document}

\title{Efficient Hamiltonian, structure and trace distance\\ learning of Gaussian states}

\author{Marco Fanizza\thanks{Inria, Télécom Paris - LTCI, Institut Polytechnique de Paris, 91120 Palaiseau, France. \texttt{marco.fanizza@inria.fr}} \thanks{Department of Mathematical Sciences, University of Copenhagen, Universitetsparken 5, 2100 Denmark} \thanks{F\'{i}sica Te\`{o}rica: Informaci\'{o} i Fen\`{o}mens Qu\`{a}ntics, Departament de F\'{i}sica, Universitat Aut\`{o}noma de Barcelona, ES-08193 Bellaterra (Barcelona), Spain}
\and Cambyse Rouz\'{e}\thanks{Inria, Télécom Paris - LTCI, Institut Polytechnique de Paris, 91120 Palaiseau, France. \texttt{rouzecambyse@gmail.com}}
\and Daniel Stilck Fran\c{c}a \thanks{
Department of Mathematical Sciences, University of Copenhagen, Universitetsparken 5, 2100 Denmark. \texttt{dsfranca@math.ku.dk}}
\thanks{Univ Lyon, ENS Lyon, UCBL, CNRS, Inria, LIP, F-69342, Lyon Cedex 07, France}}

\maketitle

\begin{abstract}

    In this work, we initiate the study of Hamiltonian learning for positive temperature bosonic Gaussian states, the quantum generalization of the widely studied problem of learning Gaussian graphical models. We obtain efficient protocols, both in sample and computational complexity, for the task of inferring the parameters of their underlying quadratic Hamiltonian under the assumption of bounded temperature, squeezing, displacement and maximal degree of the interaction graph.
    Our protocol only requires heterodyne measurements, which are often experimentally feasible, and has a sample complexity that scales logarithmically with the number of modes. Furthermore, we show that it is possible to learn the underlying interaction graph in a similar setting and sample complexity. 
    In addition, we use our techniques to obtain the first results on learning Gaussian states in trace distance with a quadratic scaling in precision and polynomial in the number of modes, albeit imposing certain restrictions on the Gaussian states. Our main technical innovations are several continuity bounds for the covariance and Hamiltonian matrix of a Gaussian state, which are of independent interest, combined with what we call the local inversion technique. In essence, the local inversion technique allows us to reliably infer the Hamiltonian of a Gaussian state by only estimating in parallel submatrices of the covariance matrix whose size scales with the desired precision, but not the number of modes. This way we bypass the need to obtain precise global estimates of the covariance matrix, controlling the sample complexity.
\end{abstract}

\newpage


\section{Introduction}
Inferring parameters of quantum systems from measurement data is a fundamental task in quantum information science that has seen significant advancements over the recent years~\cite{haah2022optimal,2004.07266,bakshi2024learning,anshu2024survey,rouze2024learning,li2307heisenberg,rouze2024efficient,stilck2024efficient,mobus2023dissipation,huang2023learning,bakshi2024structure,narayanan2024improved,fanizza2023learning,huang2022provably,lewis2024improved}. Although it has been known for roughly a decade that complete tomography of general quantum states inevitably requires a sample complexity scaling exponentially in the system's size~\cite{flammia2012quantum,Wright2016,Haah2017}, recent works have shown that it is possible to obtain a lot of information of practical interest much more efficiently. This is exemplified by the framework of classical shadows~\cite{huang2020predicting}.

Another fruitful approach to obtain efficient algorithms is to consider the learning of more restricted and physically motivated quantum states and evolutions, such as those generated by local Hamiltonians~\cite{li2307heisenberg,stilck2024efficient,mobus2023dissipation,huang2023learning,bakshi2024structure} or their Gibbs states~\cite{bakshi2024learning,haah2022optimal,2004.07266,rouze2024learning,narayanan2024improved,rouze2024efficient}, where the problem of inferring the parameters of the underlying Hamiltonian is broadly referred to as the Hamiltonian learning problem. This line of work culminated in efficient algorithms (both in terms of samples and computationally) under various assumptions. The first sample-efficient algorithm for an arbitrary constant temperature local quantum Hamiltonian was given by~\cite{2004.07266}, however the sample complexity was still polynomial in system size and the postprocessing was only efficient under further assumptions on the model, as it required solving a maximum entropy optimization problem. This was followed by work of~\cite{haah2022optimal}, which focused on the high temperature regime. Resorting to cluster expansion techniques, they were able to obtain optimal sample complexity bounds and efficient post-processing. More recently~\cite{bakshi2024learning} gave Hamiltonian learning algorithms for general Hamiltonians with efficient sample and computational complexity based on sums-of-squares relaxations. However, the scaling of both the sample complexity and computational complexity is a high-degree polynomial and likely far from optimal.

This program can be understood as a quantum version of the widely studied classical problem of learning a discrete graphical or Ising model \cite{santhanam2012information,Ravikumar2010,bresler2015efficiently,vuffray2016interaction,klivans2017learning,hamilton2017information,wu2019sparse}, but a lot of progress is still required to put the quantum protocols on an equal footing with their classical counterparts. The currently best known general purpose computationally efficient quantum protocols for Hamiltonian learning from a Gibbs state of spin systems require previous knowledge of the graph w.r.t. which the Hamiltonian is local and have a sample complexity that is polynomial in the system's size~\cite{bakshi2024learning,narayanan2024improved}\footnote{After the first version of this work was made public, logarithmic sample complexity was proven by~\cite{chen2025learning}.}. In contrast, for classical discrete graphical models it is known how to learn the graph with a number of samples that scales logarithmically in the system's size~\cite{bresler2015efficiently,vuffray2016interaction,klivans2017learning}, and known algorithms are essentially optimal~\cite{santhanam2012information}. Furthermore, the only known computationally efficient algorithms have a highly suboptimal scaling in precision~\cite{bakshi2024learning,narayanan2024improved} and, in a nutshell, the main technical barrier to generalize such results to quantum states is the lack of (exact) conditional independence structures for the underlying states which is central to most classical results.

In the classical case there is also an extensive literature dedicated to the learning of Gaussian graphical models~\cite{dempster1972covariance,besag_spatial_1974, friedman2008sparse}, which can be understood as continuous variable versions of discrete graphical models. In the last years, many fundamental questions related to this problem were settled, such as~\cite{misra20a}, that gave sample optimal~\cite{Wang2010} algorithms to learn these models, although the community still focuses on learning under more constrained settings, say in a differentially private way~\cite{10.1145/3564246.3585194} and even much more complex models such as mixtures of Gaussians~\cite{10.1145/3519935.3519953}.
The natural quantum variation of learning a Gaussian model is that of learning bosonic Gaussian states~\cite{Holevo2012}, which correspond to Gibbs and ground states of quadratic Hamiltonians. Somewhat surprisingly given the prevalence and importance of such states in quantum optics and continuous variable approaches to quantum computing~\cite{serafini2017quantum}, this  problem was not considered prior to this work. Indeed, to the best of our knowledge, the only works that considered bosonic Hamiltonian learning problems focus on learning parameters from dynamics~\cite{li2307heisenberg,mobus2023dissipation}, not states.

In this work we initiate the study of Hamiltonian learning in bosonic Gaussian states, { which can be completely characterized in terms of their first and second moments, or in terms of their first moments and their Hamiltonian matrix.} 
We show that it is possible to learn all the parameters of a quadratic Hamiltonian on $m$ modes on a known interaction graph $G$ of maximal degree $\Delta-1$ up to precision $\epsilon>0$ with a number of samples scaling as $O(\epsilon^{-{2+{\gamma}}})$ for any $\gamma>0$, with the precision and logarithmically with the number of modes, the natural measure of size of the system. Furthermore, the classical postprocessing of the data is efficient and  we resort only to heterodyne measurements, which are simple to implement in the lab. In addition, we show that it is possible to learn the graph on which the Hamiltonian is defined under mild assumptions with a similar number of samples in the precision and number of modes as for the Hamiltonian learning. Thus, our method only requires experimentally feasible measurements and is highly scalable. 
{ In Fig.~\ref{fig:graph}, we show an example of a graph of interactions, and its corresponding structure of the Hamiltonian matrix.}

In addition, we also obtain the first protocol to learn a positive temperature Gaussian state in trace distance with an inverse-quadratic scaling in precision. Even though previous recent work~\cite{mele2024learning,holevo2024estimates} showed how to solve this problem with a polynomial scaling in the number of modes, the scaling in precision was quartic, albeit it required fewer assumptions on the state compared to our result. A concurrent work also shows a related result with a different method, that also applies to pure states, and a different concrete trace distance bound~\cite{bittel2024optimalestimatestracedistance}.

From a technical perspective, our main innovations are various perturbation bounds for covariance and Hamiltonian matrices of Gaussian states that are of independent interest, combined with what we call the local inversion subroutine, a key part of our algorithms for graphical models. This allows us to estimate the Hamiltonian of Gaussian states only accessing small blocks of the covariance matrix at a time, with the precision of the estimates only scaling with the size of the blocks and not the whole system. { The workflow of the Hamiltonian learning algorithm is represented in Fig.~\ref{fig:workflow}.}

By addressing many previously open and unexplored questions in the learning of Gaussian states, our work significantly advances the field of quantum Hamiltonian learning and opens new avenues for research in the learning of continuous variable quantum systems.

\begin{figure}[h!]
    \centering
\includegraphics[scale=0.3]{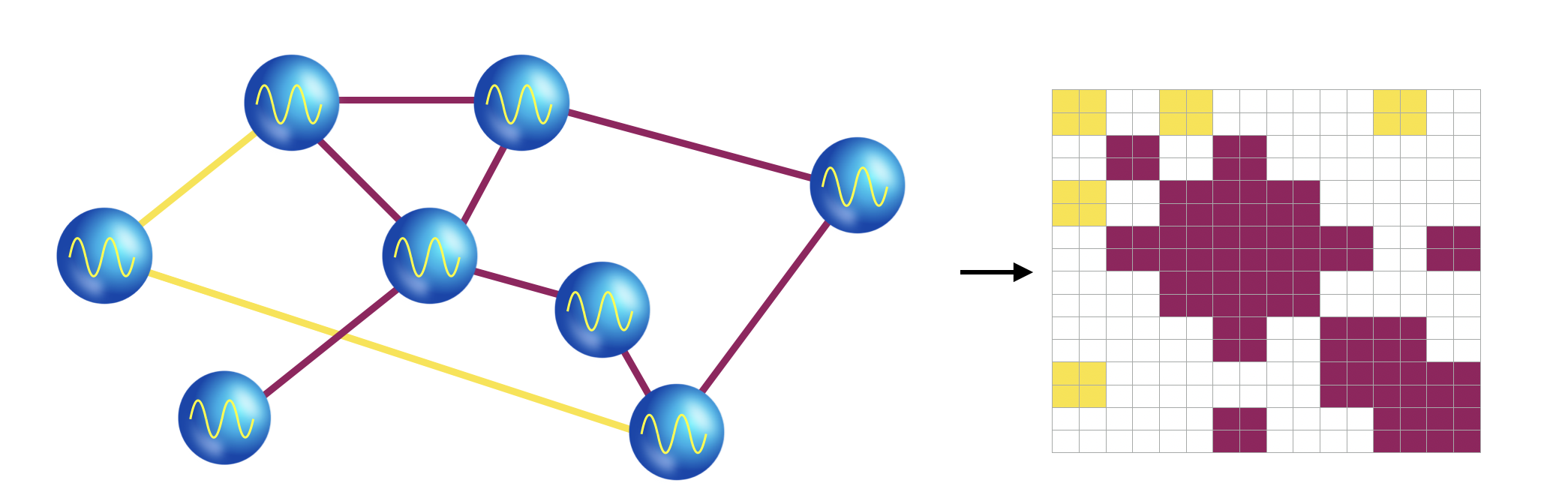}
\caption{Graph of interactions for 8 modes and corresponding structure of the Hamiltonian matrix. Vertices are numbered from left to right. The submatrices corresponding to the neighborhood of vertex $1$ at distance $1$, i.e. $\{1,3,7\}$, are highligthed in yellow. The degree of this graph is 4.}\label{fig:graph}
\end{figure}

\begin{figure}[h!]
    \centering
\includegraphics[scale=0.98]{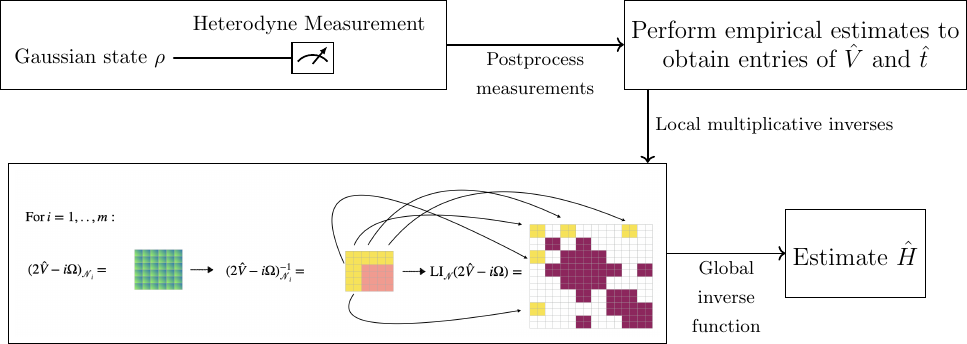}
\caption{Schematic workflow for our Hamiltonian learning protocol given we know the interaction graph. We first perform heterodyne measurements on copies of a Gaussian state $\rho$. We then use the outcomes to obtain empirical estimates of the entries of the covariance matrix and the mean. Importantly, for most of our results we will only require that the precision with which we estimate each entry does not scale with the number of modes in the systems. The reason for that is the next step, our local inversion procedure. Here, we are given a neighborhood structure $\{\mathcal{N}_i\}_{i\in V}$. Roughly speaking, for each $i$, $\mathcal{N}_i$ contains all the vertices that are at most a distance $l$ from $i$, where $l$ depends on the target precision. We then only invert the matrix $(2\hat{V}-i\Omega)$ on the submatrix corresponding to each neighborhood $\mathcal{N}_i$ (in the figure, we illustrate this step for the neighborhood of vertex $1$ of Figure~\ref{fig:graph}, for $l=1$) and "stitch together" the global matrix combining all estimates from the various neighborhoods. We then apply the function that computes $H$ from $(2V-i\Omega)^{-1}$ to the resulting matrix to form our guess of the Hamiltonian.}
\label{fig:workflow}
\end{figure}
\section{Notations and Background}
Given a finite set $\Lambda$, we denote by $\cH_\Lambda=\bigotimes_{v\in \Lambda}\cH_v$ the Hilbert space of $m=|\Lambda|$ particles, and by $\cB_\Lambda$ the algebra of bounded linear operators on $\cH_\Lambda$. The trace on $\cB_\Lambda$ is denoted by $\tr{\cdot}$. 
We use the notations $[m]$ when refering to the set $\{1,\dots ,m\}$. 
{We say that a square matrix $A\in M_{m,m}(\mathbb{C})$ is positive-semidefinite (psd) it is hermitian with nonnegative eigenvalues, and we write $A\geq B$ to say that $A-B$ is psd. A psd matrix is positive-definite if the eigenvalues are strictly positive.}
We denote by $O^\dagger$ the adjoint of an matrix $O\in M_{m_1,m_2}(\mathbb{\mathbb{C}})$, and by $A^\intercal$ the transpose.
We also denote by $\|A\|_p=\tr{(A^\dagger{A})^{p/2}}^{\frac{1}{p}}$ the Schatten $p$-norm of a matrix $A$, and by $M^{>0}_{2m}(\mathbb{R})$ the set of positive-definite symmetric matrices of size $2m\times 2m$. 
The identity map is denoted by $I$. 
The set of real symplectic matrices of size $2m\times 2m$ is denoted by $\mathrm{Sp}(2m,\mathbb{R})$.

Let us now review basic concepts of continuous variable systems and explain the connection between Gaussian quantum states and classical Gaussian graphical models. We refer to~\cite{Holevo2012} for a more complete treatment and proofs of the staments here. A continuous variable (CV) system with $m$ modes is defined on the Hilbert space $\cH_m:= L^2(\R^m)$, and
we denote the canonical operators on each mode as $X_j,P_j$ ($j=1,\ldots,m$). Define the formal vector
\begin{align*}
R:= \begin{pmatrix} X_1,P_1,\dots, X_m,P_m \end{pmatrix}^\intercal 
\end{align*}
and the symplectic form
\begin{align*}
\Omega   := \bigoplus_{i=1}^{m}\begin{pmatrix} 0 & 1 \\ -1 & 0 \end{pmatrix} 
\label{Omega}
\end{align*}
(this a block diagonal matrix, where all diagonal blocks are $2\times 2$ matrices), in terms of which the canonical commutation relations read (at least when evaluated on Schwartz functions)
\begin{equation}\begin{aligned}\hspace{0pt}
[ R_j, R_k ] = i\Omega_{jk}\, .
\label{CCR}
\end{aligned}\end{equation}

With a slight abuse of notation, we will often denote by the same symbol $\Omega$ the symplectic form for different sets of modes. The quantum covariance matrix $V[\rho]$ and mean vector $t[\rho]$ associated with a generic state $\rho$ are defined by
\begin{equation}\begin{aligned}\hspace{0pt} \label{firstsecondmoments}
t[\rho]_j := \Tr \left[\rho R_j\right]\, ,\quad V[\rho]_{jk} :=\frac{1}{2} \Tr \left[\rho \left\{R_j-t_j,\, R_k-t_k\right\}\right]\, ,
\end{aligned}\end{equation}
provided that these expressions are well defined. By the canonical commutation relations, we have that 
\begin{align*}
V[\rho]+\frac{i}{2}\Omega=\tr{(R-t[\rho])\rho (R-t[\rho])}\,.
\end{align*}
This relation as well as its transpose imply the following uncertainty relation:
\begin{align*}
V[\rho]\ge \pm\frac{i}{2}\Omega\,.
\end{align*}
The condition above is in fact not only necessary but also sufficient for a matrix $V$ to be the covariance matrix of a quantum state. For an arbitrary $x\in \R^{2m}$, we define the associated Weyl operator by
\begin{equation}\begin{aligned}\hspace{0pt}
\W(x) := e^{ i x^\intercal R }\, .
\label{D}
\end{aligned}\end{equation}
Note that $\W(x)^\dag = \W(x)^{-1} = \W(-x)$. Using displacement operators, we can re-write~\eqref{CCR} in Weyl form as
\begin{equation}\begin{aligned}\hspace{0pt}
\W(x+y) = e^{\frac{i}{2} x^\intercal\Omega y}\, \W(x) \W(y)\, .
\label{Weyl}
\end{aligned}\end{equation}
This implies that 
\begin{align*}
\W(-\Omega^{-1}x)^\dagger RW(-\Omega^{-1}x)=R+xI \,.
\end{align*}

 For an arbitrary trace-class operator $Z$, we construct its characteristic function $\chi_Z:\R^{2m}\to \C$ by
\begin{equation}\begin{aligned}\hspace{0pt}
\chi_Z(x) := \Tr \big[Z\, \W(x)\big]\, .
\label{chi}
\end{aligned}\end{equation}
Characteristic functions are always bounded and furthermore continuous, because of the strong operator continuity of the mapping $x\mapsto \W(x)$. A Gaussian state $\rho$ is a state whose characteristic function has the Gaussian form
\begin{align*}
\chi_\rho(x)=\operatorname{exp}\Big(it[\rho]^\intercal x-\frac{1}{2}x^\intercal V[\rho]x\Big)\,.
\end{align*}
It is easy to verify in that case that $t[\rho]$ and $V[\rho]$  correspond to the mean vector and covariance matrix of $\rho$. 

A Gaussian state is nondegenerate if and only if (see \cite[Lemma 12.22]{Holevo2012})
\begin{align*}
\operatorname{det}\left(V[\rho]+\frac{i}{2}\Omega\right)\ne 0\,,
\end{align*}
that is, if the matrix $V[\rho]+\frac{i}{2}\Omega$ is nondegenerate. In that case, the state $\rho$ takes the form \cite[Theorem 12.23]{Holevo2012}
\begin{align}\label{equ:structure_Gaussian}
\rho(t,H)=\W\big(-\Omega^{-1}t[\rho]\big)\,\frac{e^{-R^\intercal H[\rho] R}}{\sqrt{\operatorname{det}\left(V[\rho]+\frac{i}{2}\Omega\right)}}\,\W\big(-\Omega^{-1}t[\rho]\big)^\dagger\,,
\end{align}
where the matrix $H[\rho]$ is referred to as the (Gaussian) Hamiltonian of the system and satisfies 

\begin{align}\label{eq:HtoV}
2i \Omega V[\rho]=\operatorname{coth}(i H[\rho]\Omega)\,.
\end{align}
Inverting it, we have 
\begin{equation}\label{eq:VtoH}
H[\rho]=\frac{1}{2}\ln \left(\frac{2i\Omega V[\rho]+I}{2i\Omega V[\rho]-I}\right)i\Omega=\frac{1}{2}\ln \left(I+\frac{2}{2i\Omega V[\rho]-I}\right)i\Omega.
\end{equation}

A fundamental property of $H$ that we will repeatedly use is the following.  For any (symmetric) $2m\times 2m$ Hamiltonian $H>0$, let a symplectic diagonalization be
\begin{equation}\label{eqsymplect2}
H=S^{-\intercal}D{S^{-1}}\,,
\end{equation}
with $D=\oplus_{i=1}^m d_i I_{M_2}$, with $d_i>0$, and $S\in\mathrm{Sp}(2m,\mathbb{R})$. {Operationally, the symplectic diagonalization (or Williamson's normal form) corresponds to a preparation procedure for a Gaussian state: 1) prepare $m$ thermal states of harmonic oscillators with inverse temperatures $2d_i$. 2) Act on them via the phase space transfomation corresponding to the symplectic matrix $S$, which can be realized via a passive linear transformation, followed by single-mode squeezing, and by another passive linear transformation~\cite{serafini2017quantum}. The operator norm $\|S\|_{\infty}$ encodes the maximum amount of single-mode squeezing of this decomposition of $S$. Moreover, Hamiltonians that commute with the total energy $E_m=\frac{1}{2}\sum_{i=1}^{2m} R_iR_i$ have the form $\sum_{i,j}a^{\dagger}_{i}K_{ij}a_j$, where we expressed position and momentum operators in terms of creation and annihilation operators. Here $K$ must be a complex psd matrix and it is diagonalized by a unitary matrix. As a symplectic transformation of the phase space, this is an orthogonal symplectic transformation and it has $\|S\|_{\infty}=1$. }
 Let us also denote as $d_{\min}$ and $d_{\max}$ the smallest and the largest symplectic eigenvalues of $H$, respectively, and let $\lambda_{\min}$ and $\lambda_{\max}$ be the smallest and the largest eigenvalues of $H$, respectively.
By Theorem 11 in~\cite{Bhatia2015}, $\lambda_{\min}\leq d_{\min}\leq d_{\max}\leq\lambda_{\max}$, and since $\lambda_{\max}=\|H\|_{\infty}\geq \|S\|_{\infty}^2 d_{\min}\geq  \|S\|_{\infty}^2 \lambda_{\min} $, we have $1\leq \|S\|_{\infty}^2\leq \frac{\lambda_{\max}}{\lambda_{\min}}$. In the following, $\|S\|, d_{\min}, d_{\max}$ will be used as parameters, as they have a natural physical interpretation, but these inequalities clarify how they relate to more usual parameterization of Hamiltonian matrices, in terms of their norm and condition number.
We will say that 
a Hamiltonian matrix $H$ has graph $G=([m],\mathsf{E})$ made of $m$ vertices, where two vertices $i,j\in [m]$ are connected by an edge whenever the $2\times 2$ sub-matrix
$\left\{H_{2i-\delta,2j-\delta'}\right\}_{\delta,\delta'\in\{0,1\}}$ is non-zero. 

{ We can also bound $\|S\|_{\infty}$, $d_{\min}$ and $\|t\|_{2}$ as a function of the energy of the state, i.e. the expectation value of operator $E_m$, see Appendix~\ref{sec:energybounds} for a proof:
\begin{align}
\|t\|_{2}^2&\leq 2E, \label{ineq:energy1}\\
\|S\|_{\infty}^2&\leq 4E, \label{ineq:energy2}\\
(1-e^{-2d_{\min}})^{-1}&\leq 8E^2.\label{ineq:energy3}
\end{align}

However, $d_{\max}$ is essentially independent of the energy, and expresses how close the state is to a pure state. In fact, for a sequence of states converging to the vacuum, with $d_{\max}=d_{\min}=d$, $S=I$, $d$ going to infinity, the energy goes to a constant.

}

{
\subsection{Bosonic graphical models in many-body physics}

Bosonic Hamiltonians with finite number of modes are used to describe a multitude of interacting many-body quantum systems, such as cold atoms in optical lattices, magnetization fields, optomechanical systems, trapped ions, and  electromagnetic fields in cavities, interacting with matter~\cite{serafini2017quantum}. Among those, sparse Hamiltonian have been extensively considered for bosonic systems, as most materials are naturally described by Hamiltonians with a lattice structure, where the interactions strength decays with the distance. The sparse Hamiltonians considered in this work are models where the interactions are assumed to be zero if vertices are not connected, for example because the corresponding sites are further than some fixed distance on a lattice. This includes the approximation of considering non-zero interactions only between first neighbors. A key example of this is a grid of harmonic oscillators, an idealization which is a cornerstone of theoretical physics, since any potential can be approximated around its minimum by a quadratic one. Moreover, in systems with higher order-interactions, the Bogoliubov-De Gennes approach is used to study excitations above the ground state, via an effective quadratic Hamiltonian. This method captures the physics of weakly-interacting dilute Bose gases~\cite{lewenstein_ultracold_2012}.  Moreover, a typical quadratic Hamiltonian is also the hopping (or tight-binding) Hamiltonian on a graph, which is the basis of all lattice models.
The sharp cutoff on local interactions that is also quite convenient for learning questions, as it reduces the number of parameters. 

In fact, experimentally realizing local (and thus sparse) Hamiltonians, including bosonic ones, is the goal of analog quantum simulation~\cite{lewenstein_ultracold_2012}, both for reproducing models of condensed matter physics and for creating and studying meta-materials that cannot be found in nature. Starting from proposals for quantum simulation of the Bose-Hubbard model in cold atoms~\cite{jaksch_cold_1998}, the field is also exploring also Hamiltonian with non-Euclidean interaction graphs, which exhibiting non-trivial band structure due to their topology~\cite{peano_topological_2015,carusotto_photonic_2020}.

With current technology, it is in fact possible to realize simulators that implement Hamiltonians with a desired graph structure, in a variety of platforms and in a programmable way. 
Perhaps the closest work to our idealized setting is \cite{periwal_programmable_2021}, which reports the engineering of various interaction graphs, including rings, chains, ladders, non-Archimedean geometries, in an array of atomic ensembles with 18 sites and $10^{4}$ Rubidium atoms per site, interacting via spin-exchange. The collective magnetizations at each site can be treated as bosonic variables in the relevant regime, with a quadratic Hamiltonian. The paper also shows procedures to verify the interactions from two-point correlation functions of the magnetization, using a combination of symmetries of their model, an ansatz on the covariance matrix, and a high inverse temperature approximation,  estimating the interaction matrix as the inverse of the covariance matrix -- i.e. the classical method, expected to be valid at high temperature. Our methods give a way to estimate the interactions and the graph in these types of systems that is applicable at lower temperatures and without model-specific assumptions. 
We also mention \cite{senanian_programmable_2023}, where quadratic bosonic Hamiltonian on programmable graphs have also been simulated via photonic synthetic lattices with up to $10^{5}$ sites, where frequency modes are mapped to lattice sites, and \cite{youssefi_topological_2022}, where  lattice bosonic Hamiltonians, including honeycomb lattice structure, have been enginereed on $10$ sites of a superconducting optomechanical systems, and the quadratic Hamiltonian is reconstructed with modeshape measurements. 
Finally, \cite{katz_programmable_2023} is a theoretical proposal for bosonic Hamiltonian simulation using phonons in trapped-ion crystals, interacting via the excitations of the trapped ion spins. In all the above, quadratic hopping Hamiltonians already explore interesting physics and testing their correct implementation is paramount to go forward to the simulation of interacting models.

The problem of Gaussian Hamiltonian learning can also be seen as a non-asymptotic instance of Gaussian multiparameter metrology, for which studies on the asymptotic achievable precision exist (via the Fisher information matrix, e.g.~\cite{nichols_multiparameter_2018}). Indeed, our algorithms also suggest many-body metrology applications~\cite{montenegro_review_2025}: without trying to achieve the optimal information-theoretic limits which are highly sensitive on precise values of the parameter themselves, assumed to be roughly known, our methods ensure that the estimate of any parameter encoded linearly in the Hamiltonian matrix would be robust with respect to uncertainties in the other parameters, since the full Hamiltonian can be estimated at once, with non-asymptotic guarantees. 
As an example, bosonic Gaussian systems are experimentally used as thermometric probes~\cite{Purdy2015}. Existing works on this (see e.g.~\cite{Cenni2022}) optimize over feasible measurements
(including those we use in out protocol) to infer the temperature of a Gaussian state departing from a \emph{known} value.

It is also worth to discuss our assumption that the Gibbs states of quadratic Hamiltonians is a natural resource to consider for a learning problem. In the theory of open quantum systems, the fact that Gibbs states are thermal equilibrium states can be justified through the microscopic derivation of the Lindblad master equation from an interaction with a large bath, using the Born-Markov and secular approximations~\cite{davies_markovian_1974}. The result of this procedure is a time evolution that has the Gibbs state of the system Hamiltonian as fixed point and, if a certain ergodicity condition is satisfied, converges to the Gibbs state.  While this approximation had tremendous success in quantum optics, it is not well justified in other settings, especially when the system of interest is large, see e.g.~\cite{scandi_thermalization_2025} where other approximation methods are proposed. Nevertheless, a recent paper showed a model that gives relaxation to Gibbs states for generic quadratic fermionic or bosonic models coupled to generic Gaussian baths satisfying detailed balance, where uniform interactions avoid the need of the secular approximation~\cite{shiraishi_davies_2025}. We would also stress that for reconstructing the Hamiltonian, it would be enough that first and second moments thermalize to those of the Gibbs state, and that the covariance estimator stays sub-exponential. The former assumption is at the core of less strict definitions of thermalization~\cite{gluza_equilibration_2019}, while the latter is a very reasonable assumption for imperfect implementations of heterodyne measurements. 
}

\subsection{Relation to Gaussian graphical models}

Let us briefly recall the extensively studied problem of learning Gaussian graphical models, initiated by~\cite{dempster1972covariance,besag_spatial_1974}, and make the connection to the quantum problems we will study. In the problem of estimating a Gaussian graphical model, one is given $n$ i.i.d. $X_1,\ldots,X_n\in\mathbb{R}^m$ samples from a multivariate normal distribution $\mathcal{N}(t,\Sigma)$ with mean vector $t\in\mathbb{R}^m$ and covariance matrix $\Sigma\in\mathbb{R}^{m\times m}$. It is then known that the inverse of the covariance matrix $\Theta=\Sigma^{-1}$, often also called the precision matrix, encodes the conditional independence structure of the random variables. Furthermore, the pdf of the random variables is given by:
\begin{align}
f_{t,\Sigma}(x)=\frac{1}{\sqrt{\det(\Sigma)}(2\pi)^{m/2}}\textrm{exp}\left(-\frac{1}{2}(x-t)^t\Sigma^{-1}(x-t)\right)
\end{align}

This is in direct analogy with Equation~\eqref{equ:structure_Gaussian}, where we see that the quantum state can be expressed in a similar matter, albeit with a different functional dependence on the covariance matrix for the matrix in the exponent. But the interpretation in the quantum case of the nonzero entries of $H$ is similar to the classical case, as these represent the systems that "interact" with each other.

In the literature on learning of Gaussian graphical models, one is often interested in learning the conditional independence structure of the random variables, which corresponds to estimating the sparsity structure of $\Theta$. From that, one is also usually interested in estimating the parameters of $\Theta$, while having a scaling of the sample complexity that only scales logarithmically with the dimension of the random variables.
Thus, we see that the problem of learning the sparsity structure of $H$ is in direct correspondence with that of learning the conditional independence structure graph of a Gaussian graphical model. And that of estimating the entries of $H$ in turn corresponds to that of estimating the parameters of the Gaussian graphical model.

{

Furthermore, the two problems also share some technical aspects; for instance, in constrast to the discrete case, the underlying objects we sample from typically take unbounded values and a finer understanding of concentration is required. But the quantum problem poses unique challenges that are not present in the classical case. One difference is that the uncertainty principle obstructs from sampling from a multivariate normal with covariance matrix $V$ measuring one copy of the state: we resort to heterodyne measurement which gives as output a multivariate normal with covariance matrix $V+I/2$. More importantly,  the more complicated functional dependence between the covariance matrix and $H$ makes the analysis more challenging, since entries of $V$ corresponding to the neighborhood of a vertex $i$ are not sufficient by themselves to reconstruct exactly the interactions of $i$ with its neighbors (matrix elements $H_{ij}$). This fact can be understood as related to the lack of conditional independence in the quantum case: in general, one cannot reconstruct the state of the system by applying a channel on the complement of vertex $i$ that acts nontrivially only on the vertices connected to $i$.

The problem of estimating Gaussian graphical models with sparsity is a central theme of high-dimensional statistics, and we refer to books such as~\cite{hastie_elements_2009,hastie_statistical_2015,wainwright_high-dimensional_2019} for a more in-depth history of the problem. When the graph is known, a simple procedure using local inversions for each neighborhood suffices to find the coefficients very efficiently, as explained in Chapter 17 of~\cite{hastie_elements_2009}. In the zero-mean case, $O(\frac{d}{\epsilon^2})$ are sufficient to learn a $d\times d$ covariance matrix in relative error, i.e. $\Sigma(1-\epsilon)\leq\hat{\Sigma}\leq \Sigma(1+\epsilon)$~\cite{wainwright_high-dimensional_2019}. Applying local inversion to each neighborhood, symmetrizing, and via union bound, this means that $O(\frac{\Delta \|\Theta\|_{\infty}}{\epsilon^2}\log m)$ samples guarantee an entrywise error for the entries of $\Theta=\Sigma^{-1}$ of order $O(\epsilon)$, thus with no condition number dependence\footnote{It seems to us that a slight dependence on the diagonal elements of $V$ is instead needed in the non-zero mean case, which is not always treated in the literature. To learn the precision matrix alone, one can simply subtract copies of samples to get a zero-mean gaussian with twice the covariance matrix.}. An essentially matching lower bound to this simple algorithm is given by Theorem 2 of~\cite{Wang2010}, which states that $\Omega(\frac{\Delta^2}{\epsilon^2}\log m)$ samples are necessary. 
Let us thus stress that local inversions on neighborhoods at distance $1$ are already sufficient in the classical case because conditional independence ensures that the conditional distribution of a random variable $X_i$, when the other variables at distance one in the graph are known, contains the full information about the interactions of $i$ with its first neighbors. Our local inversion method, instead, goes further in the graph to account for the lack of conditional independence.

In fact, most of the literature has directly dealt with the harder problem of estimating the precision matrix when the graph is not known, as well as estimating the graph itself. To do that, the main strategy is to build an estimator ot the precision matrix as the solution of a regularized and/or constrained optimization problem, to enforce sparsity of the solution. A popular approach is the graphical LASSO, 
first studied in~\cite{friedman2008sparse, yuan_model_2007}, that consists in maximizing the $l_1$-regularized log-likelihood. The analysis of the sample complexity of the graphical LASSO~\cite{ravikumar_high-dimensional_2011} gives an upper bound $O(\frac{\Delta^2}{\alpha^2\epsilon^2}\log m)$ sample complexity, where $\alpha$ is a parameter encoding a certain incoherence condition of the precision matrix. Other approaches are based on estimating the precision matrix one row at a time, via LASSO~\cite{meinshausen_high-dimensional_2006}, Danzig~\cite{yuan_high_2010}, or $l_1$-constrained optimization (CLIME)~\cite{cai_constrained_2011}.
Simpler loss functions~\cite{zhang2014sparse} and non-convex methods have also been considered~\cite{wang2016precision}.
These algorithms 
that are all efficient to implement, scaling polynomially with the number of variables. The sample complexity is $O(\log m)$, but only if appropriate assumptions, that can be broadly linked to the condition number of $\Theta$ being bounded, hold~\cite{misra20a}. For the problem of graph selection, the paper~\cite{misra20a} showed a sample complexity $O(\frac{\Delta\log m}{\kappa^2})$, with $\kappa$ being a lower bound on the relative strenghts of the interactions, and no condition number dependence, and matching the $\kappa$ and $m$ dependence in the lower bounds in~\cite{Wang2010}. The tradeoff is a worse scaling in computational complexity, still polynomial in $m$ but with a degree depending on the degree of the graph, since for each node one needs to test all possible neighborhoods (improvements are possible for certain models in~\cite{kelner2020learning}).
Our graph learning algorithm is inspired by the approach in~\cite{misra20a} (especially their DICE algorithm). However, we were not able to remove the condition number dependence, and leave it to future work.

}
\section{Results}
In this paper, we make substantial progress or deliver the first results on three fundamental problems related to learning Gaussian quantum states: learning in trace distance, Hamiltonian learning and graph learning. 
We will state all of our results in terms of parameters that control the complexity of the Gaussian states. First, the amount of squeezing in the Gaussian state (namely the operator norm of the symplectic diagonalization matrix $S$), the ``temperature'' of the state (as measured by its largest and smallest symplectic eigenvalues $d_{\min},d_{\max}$), the maximal displacement $t_{\max}$, the number of modes $m$ and the maximal degree of the graph of the Hamiltonian, equal to $\Delta-1$. For all our results we will assume we have constant-size bounds on the relevant parameters.
{
See Table~\ref{tab:parameters} for a recap of notation, interpretations and relations between the parameters

\begin{table}[t]
\centering
\renewcommand{\arraystretch}{1.4}
\begin{tabular}{l p{7.5cm} p{6cm}}
\toprule
\textbf{Symbol} & \textbf{Description} & \textbf{Useful Relations} \\ 
\midrule
$m$ & Number of modes & -\\
$H$ & Hamiltonian matrix & - \\
$\Delta$  & Degree of the interaction graph +1 & $\|H\|_{\infty}\leq 2\Delta\max_{i,j}|H_{i,j}|$ \\
$\mathsf{E}$ & Edges of the interaction graph & $0<\max_{\delta_1,\delta_2\in\{0,1\}}|H_{2i-\delta_1,2j-\delta_2}|$ iff $(i,j)\in \mathsf{E}$ \\
$\kappa$  & Lower bound on absolute values of\newline Hamiltonian interaction matrix elements & $\kappa<\max_{\delta_1,\delta_2\in\{0,1\}}|H_{2i-\delta_1,2j-\delta_2}|$ for any  $(i,j)\in\mathsf{E}$. \\
$D,S$ & Williamson's normal form of $H$ & 
\(\begin{aligned} 
H &= S^{-\intercal}DS^{-1} 
\end{aligned}\) \\
$V$   & Covariance matrix & 
\(\begin{aligned} 
V &= \frac{1}{2} S \coth(D) S^{\intercal} 
\end{aligned}\) \\
$t$   & Vector of first moments & - \\
$E$   & Expectation value of the total energy & 
\(\begin{aligned}
\Tr[V] + \|t\|_2^2 &\le 2E
\end{aligned}\) \\
$d_{\max}$ & Maximum symplectic eigenvalue of $H$, maximum inverse temperature of a normal mode, purity parameter &
\vspace{-0.3cm}
$\|D\|_{\infty}= d_{\max} \le \|H\|_{\infty}$\\
$d_{\min}$ & Minimum symplectic eigenvalue of $H$, minimum inverse temperature of a normal mode, spectral gap of $RHR^{\intercal}$ & 
\vspace{-0.3cm}
\(\begin{aligned}
\|D^{-1}\|_{\infty} = d_{\min}^{-1}& \le \|H^{-1}\|_{\infty} \\
(1-e^{-2d_{\min}})^{-1} &\le 8E^2
\end{aligned}\)\\
$\|S\|_{\infty}$ & Operator norm of $S$, maximum amount of squeezing & 
\(\begin{aligned}
1 \le \|S\|_{\infty} &\le \frac{\|H\|_{\infty}}{d_{\min}} \\
\|S\|_{\infty} &= \|S^{-1}\|_{\infty} \\
\|S\|_{\infty}^2 &\le 4E
\end{aligned}\) \\
\bottomrule
\end{tabular}
\caption{Notation and relations between the several parameters used in the analysis.}
\label{tab:parameters}
\end{table}
}

Moreover,
\textit{all of our results only require heterodyne measurements, which are easy to implement experimentally, and have polynomial postprocessing in the number of modes}. 
As is well-known and explained in detail in Sec.~\ref{sec:learning_covariance}, given a Gaussian state $\rho$ with covariance matrix $V$, performing a heterodyne measurement corresponds to sampling from a Gaussian random vector with covariance matrix given by $V+I/2$. Thus, by resorting to standard concentration inequalities of Gaussian random variables, we show in Sec.~\ref{sec:learning_covariance} how to estimate the entries of the covariance matrix by simply taking the empirical average, subtracting the identity and projecting onto the set of valid covariance matrices. This procedure is the bread and butter of all our algorithms, as we always depart from some estimate of the covariance matrix. In particular, we show that it suffices to take
\begin{align}
N= \mathcal{O}\left(\max_i\frac{\left(2V_{i,i}+1+|t_i|^2\right)^2(1+|t_i|)^2}{\epsilon^2}\ln\left(\frac{m}{\delta}\right)\right)
\end{align}
samples to obtain a covariance matrix estimate such that all entries are $\epsilon$ close to the true with probability of success at least $1-\delta$.

The second ingredient that is present in the majority of our proofs are continuity bounds for the Hamiltonian in terms of the covariance matrix and vice-versa and are discussed in Sec.~\ref{sec:conntinuity_bounds}. These are akin to the bounds derived in~\cite{2004.07266} for spin system, although our proofs are based on techniques from matrix analysis and theirs more on techniques from many-body systems. We believe that these continuity bounds are of independent interest.

We start with our result on learning in trace distance. Roughly speaking (see Def.~\ref{def:trace_dist_learning} for a formal definition), the trace distance learning problem requires us to output a covariance matrix or Hamiltonian and mean vectors such that the corresponding Gaussian state is $\epsilon$ close in trace distance to the true state with probability at least $1-\delta$ given access to measurement data and corresponds to the classical problem of estimating a multivariate Gaussian in total variation distance. We then show in Sec.~\ref{sec:tracelearn}:
\begin{thm}[Learning Gaussian states in trace distance (informal)]
Let $\rho(t,H)$ be a Gaussian state on $m$ modes.
Then, for $1>\epsilon,\delta>0$, it suffices to take
\begin{align}
    N={\mathcal{O}}\big(\epsilon^{-2}m^3\ln(m\delta^{-1})\mathrm{poly}(\|S\|_{\infty},(e^{2d_{\max}}-1)(1-e^{-2d_{\min}})^{-1},t_{\max})\big)
\end{align}
copies of $\rho$ to obtain an estimate of $\rho$ up to trace distance $\epsilon$ with success probability at least $1-\delta$.
\end{thm}
Note that this is the first result for this problem with an $\epsilon^{-2}$ scaling~\cite{mele2024learning}, together with~\cite{bittel2024optimalestimatestracedistance}. 
{ In fact,~\cite{mele2024learning} used a different parametrization than ours and expressed the sample complexity upper bounds in terms of an upper bound on the expected energy of the state. Theorem 4 in~\cite{mele2024learning} shows an upper bound $\mathcal{O}(\frac{m^7 E^3}{\epsilon^4})$, while Theorem 5 in the concurrent paper~\cite{bittel2024optimalestimatestracedistance} improved this to $\mathcal{O}\left(\frac{E^4}{\epsilon^2}\left(m+\log \frac{2}{\delta}\right)\right)$. 
We can express our bound in terms of the expectation value of a bound on the energy of the state to learn, and of a bound $\beta_{\max}\geq d_{\max}$ as $N=\mathcal{O}\left(\frac{m^3E^4(e^{\beta_{\max}}-1)^2}{\epsilon^2}\log\left(\frac{m}{\delta}\right)\right)$, see Remark~\ref{remarkenergy}. We see that we have to assume not only bounded energy but also bounded value of $\beta_{\max}$: highest values mean that the state is closer to being non-full rank, and in that case the Hamiltonian is not well-defined and our approach cannot work. This means in particular our approach does not cover pure state learning. We also mention that alternative trace distance bounds with $\mathcal{O}(\sqrt{\eps})$ scaling were obtained by~\cite{holevo2024estimates},~\cite{holevo2024estimatesburesdistancebosonic}, who argued that while suboptimal for the trace distance, this behaviour is tight for the Bures distance~\cite{holevo2024estimatesburesdistancebosonic}. 
}

\paragraph{Proof sketch:} we start by employing Pinsker's inequality to upper-bound the trace distance between two Gaussian states in terms of their (symmetrized) relative entropies. We are then able to obtain a close expression for the relative entropy, yielding the expression:
\begin{align}\label{equ:trace_from_Hamilt_mr}
\|\rho-\hat{\rho}\|_1&\le\sqrt{\sum_{i,j}(\hat{H}_{i,j}-H_{i,j})(V_{i,j}-\hat{V}_{i,j})+2(t_{i}-\hat{t}_i)(t_{j}-\hat{t}_j)(H_{i,j}+\hat{H}_{i,j})}
\end{align}
Then, using our estimators for the covariance matrix and the mean vector, combined with the continuity estimates for the Hamiltonian entries in terms of the covariance, we are able to upper-bound the trace distance and get the advertised sample complexity. 

 We provide the first result on Hamiltonian learning for Gaussian states, as discussed in Sec.~\ref{sec:hamlearn}. In essence, this problem asks for recovering the Hamiltonian $H$ of $\rho$ in some norm up to $\epsilon>0$ in some metric from measurement on $\rho$ with probability of success $1-\delta$. Here we focus on recovery in operator norm and obtain (see Corollary~\ref{cor:generalgraph} for a more detailed account on the combined dependence on the parameters):
\begin{thm}[Gaussian Hamiltonian learning (informal)]\label{th:informallearning}
Let $\rho$ be a Gaussian state with Hamiltonian $H$ of maximal degree $\Delta-1$.
Then, it suffices to take
\begin{align}
N=\mathcal{O}\left(\frac{1}{\eps^{2+\gamma}}\ln\left(\frac{m}{\delta}\right)\right) \,\, \forall \gamma>0
\end{align}
copies of $\rho$ to obtain an estimate $\hat{H}$ satisfying $\|H-\hat{H}\|_{\infty}\leq \epsilon$ with probability of success at least $1-\delta${, using heterodyne measurements to estimate the covariance matrix and Algorithm~\ref{alg:learnHG} to reconstruct the estimate of the Hamiltonian.} 
\end{thm}

{\begin{rem}\label{rem:complexitygraphlearning}The computational complexity of the above procedure is $O(\mathrm{poly}(m))$ when the other parameters are fixed, dominated by the classical post-processing of Algorithm~\ref{alg:learnHG}, and with exponent of $m$ independent from the other parameters. In fact, each of the $m$ calls to the local inversion subroutine (Definition~\ref{constrappr1}) requires a matrix inversion of a submatrix of $2\hat{V}-i\Omega$, and the final logarithm can be implemented in $O(\mathrm{poly}(m))$ time, e.g., via Jordan decomposition of the matrix $\operatorname{LI}_{\hat{\mathcal{N}}}(2\hat{V}-i\Omega)i\Omega$.
\end{rem}

\begin{figure}
\centering
\begin{minipage}{.9\linewidth}
\begin{algorithm}[H]
\caption{\textsf{Hamiltonian reconstruction from local inversions}} 
	\label{alg:learnHG}
	\begin{algorithmic}[1]
        \State Input: estimate $\hat{V}$ of the covariance matrix $V$ of $m$-mode Gaussian state $\rho$, graph $G$, distance parameter $l$
        \State Set $\mathcal{N}=\{\mathcal{N}_i(l)\}_{i\in[m]}$, where $\mathcal{N}_i(l)$ is the set of nodes at distance at most $l$ from $i$ in the graph $G$.
        \State Compute the estimator $\operatorname{LI}_{\mathcal{N}}(2\hat{V}-i\Omega)$ of $(2V-i\Omega)^{-1}$ as in Definition~\ref{constrappr1}
        \State Compute the estimator $\hat{H}=\frac{1}{2}\log(I+2\operatorname{LI}_{\hat{\mathcal{N}}}(2\hat{V}-i\Omega)i\Omega)i \Omega$ of $H$, 
        \State Output $\hat{H}$.
	\end{algorithmic} 
\end{algorithm}
\caption{Algorithm for Hamiltonian learning for a known graph}
\end{minipage}
\end{figure}

}
We emphasize the fact that the sample complexity is only logarithmic in the number of modes. {Furthermore, we show that for graphs with polynomially growing neighborhoods, the sample complexity of our procedure in terms of precision is $\epsilon^{-2}$ up to polylog terms.} In Fig.~\ref{fig:simulation-errors-comparison} we showcase numerical simulations of our method for a 1D Hamiltonian with up to $1200$ modes, showing how it outperforms plug-in estimation methods. In addition, in Fig.~\ref{fig:simulation-errors-l} we show how truncating to nearest neighbors is not sufficient to obtain a good recovery even when we are not limited by finite sample-size.

\paragraph{Proof sketch:}  Recall that, given the covariance matrix $V$, the Hamiltonian of the Gaussian is given by $H=\frac{1}{2}\ln \left(I+\frac{2}{2i\Omega V-I}\right)i\Omega$. Thus, it is possible to infer the Hamiltonian if we have a sufficiently accurate approximation of the inverse of $2V-i\Omega $, which is given by $(2 V-i\Omega )^{-1}=-\frac{I-e^{+ 2Hi\Omega}}{2}(i\Omega)=\frac{1}{2}\sum_{n=1}^{\infty}\frac{\left(2Hi\Omega\right)^{n}}{n!}(i\Omega)$. Using these observations and the continuity estimates we develop in this work already suffice to obtain a Hamiltonian learning protocol with polynomial sample complexity in the number of  vertices. The main reason for the polynomial scaling is that this plug-in strategy requires us to estimate each entry of the covariance matrix to polynomial precision to ensure that the error in operator norm is of order $\epsilon$ once we perform the inversion. 

However, based on classical results~\cite{dempster1972covariance,friedman2008sparse}, we expect to be able to solve the Hamiltonian learning problem with a number of samples that scales logarithmically with the number of modes.
To achieve such a scaling, we pioneer what we call the local inversion technique. The intuition behind the technique is that, for local Gibbs states, the correlations with neighboring sites should determine the interactions. For classical Gaussian graphical models this is true by conditional independence, so it is possible to completely determine them by just considering constant-sized neighborhoods around each vertex. In contrast, we show that for Gaussian quantum states it suffices to post-process an estimate of $(2V-i\Omega)^{-1}$, constructed such that matrix entries corresponding to a pair of  vertices are obtained from inverting submatrices of the covariance matrix corresponding to vertices that are at most $\sim\frac{\log(\epsilon^{-1})}{\log\log(\epsilon^{-1})}$ away from the target pair of vertices, to estimate the interactions with error up to $\epsilon$. To show this, we resort to various bounds on the error in approximating inverses of approximately sparse matrices with local inversions, and the Taylor series representation of $(2V-i\Omega)^{-1}$ in terms of the Hamiltonian. These bounds are explained in Sec.~\ref{sec:sparse_approx}. But, crucially, the precision with which we need to estimate each of these submatrices to succesfully invert them locally does not scale with the number of modes, yielding our Hamiltonian learning result.

\begin{figure}[ht]
    \centering
    \begin{subfigure}[t]{0.48\textwidth}
        \centering
        \includegraphics[width=\textwidth]{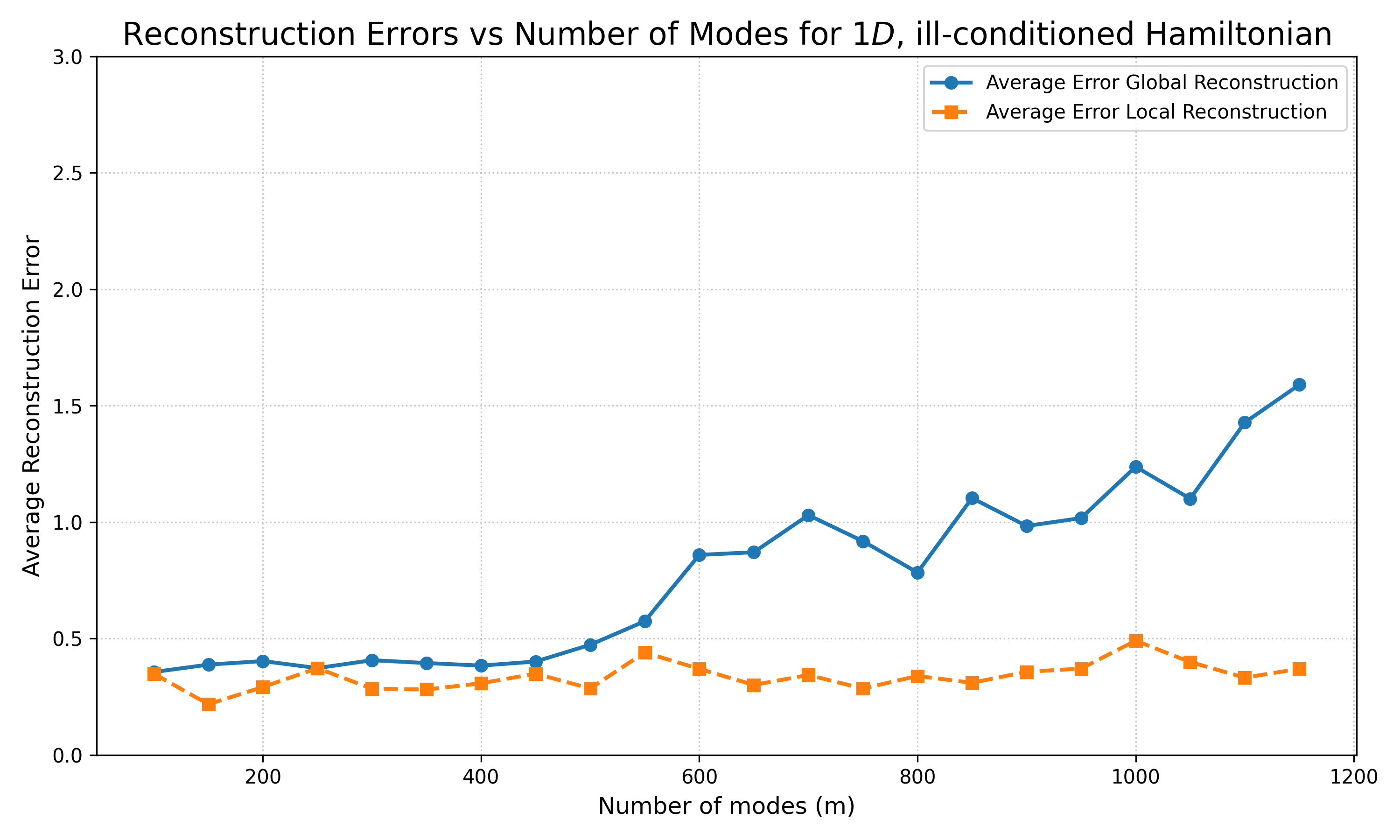}
        \caption{Ill-conditioned Hamiltonian ($c=0$)}
        \label{fig:simulation-errors-ill}
    \end{subfigure}
    \hfill
    \begin{subfigure}[t]{0.48\textwidth}
        \centering
        \includegraphics[width=\textwidth]{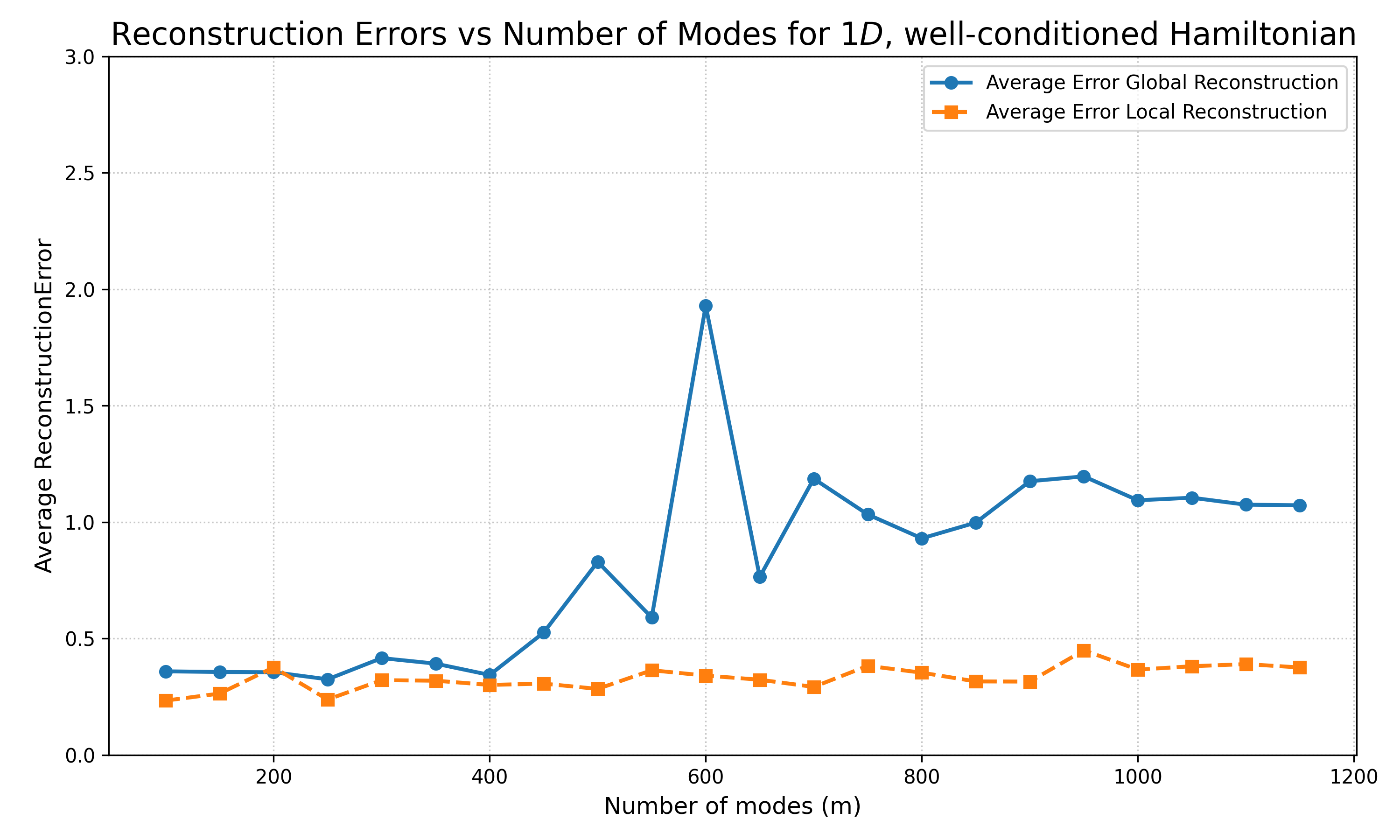}
        \caption{Well-conditioned Hamiltonian ($c=0.1$)}
        \label{fig:simulation-errors-well}
    \end{subfigure}
    \caption{Reconstruction errors vs. number of modes for ill- and well-conditioned Hamiltonians. For this figure, we considered two variations of 1D-Hamiltonians $H=(2+c)I-\ketbra{0}-(\sum_{i}\ket{i}\bra{i+1}+h.c.)$ and compared the reconstruction using the global plug-in method (estimate $\hat{H}=\frac{1}{2}\log(I+2(2\hat{V}-i\Omega)^{-1}i\Omega)i \Omega$ ) versus the local method ($\hat{H}=\frac{1}{2}\log(I+2\operatorname{LI}_{\hat{\mathcal{N}}}(2\hat{V}-i\Omega)i\Omega)i \Omega$) at inverse temperature $\beta=0.5$. The well-conditioned Hamiltonian corresponds to $c=10^{-1}$ and the ill-conditioned $c=0$, see App.~\ref{condnumbex} for a discussion on the condition number of the Hamiltonian for $c=0$  as a function of $m$.
    We performed numerical experiments generating $10^4$ samples for all number of modes and reconstructed the same Hamiltonian for $5$ different sets of samples for each number of modes. The $y$-axis displays the maximum difference $|H_{i,j}-\hat{H}_{i,j}|$ between the estimated Hamiltonian and the ideal one averaged over $5$ realizations of the samples. As they are both 1D models, we have $\Delta=2$ for both cases. 
    Importantly, we see that the quality of the recovery in both cases is essentially independent of the number of modes for both the well and ill-conditioned case when using our local inversion method and that it outperforms the error of the global inversion by roughly one order of magnitude for larger numbers of modes. }
    \label{fig:simulation-errors-comparison}
\end{figure}

\begin{figure}[ht]
    \centering
        \includegraphics[width=0.7\textwidth]{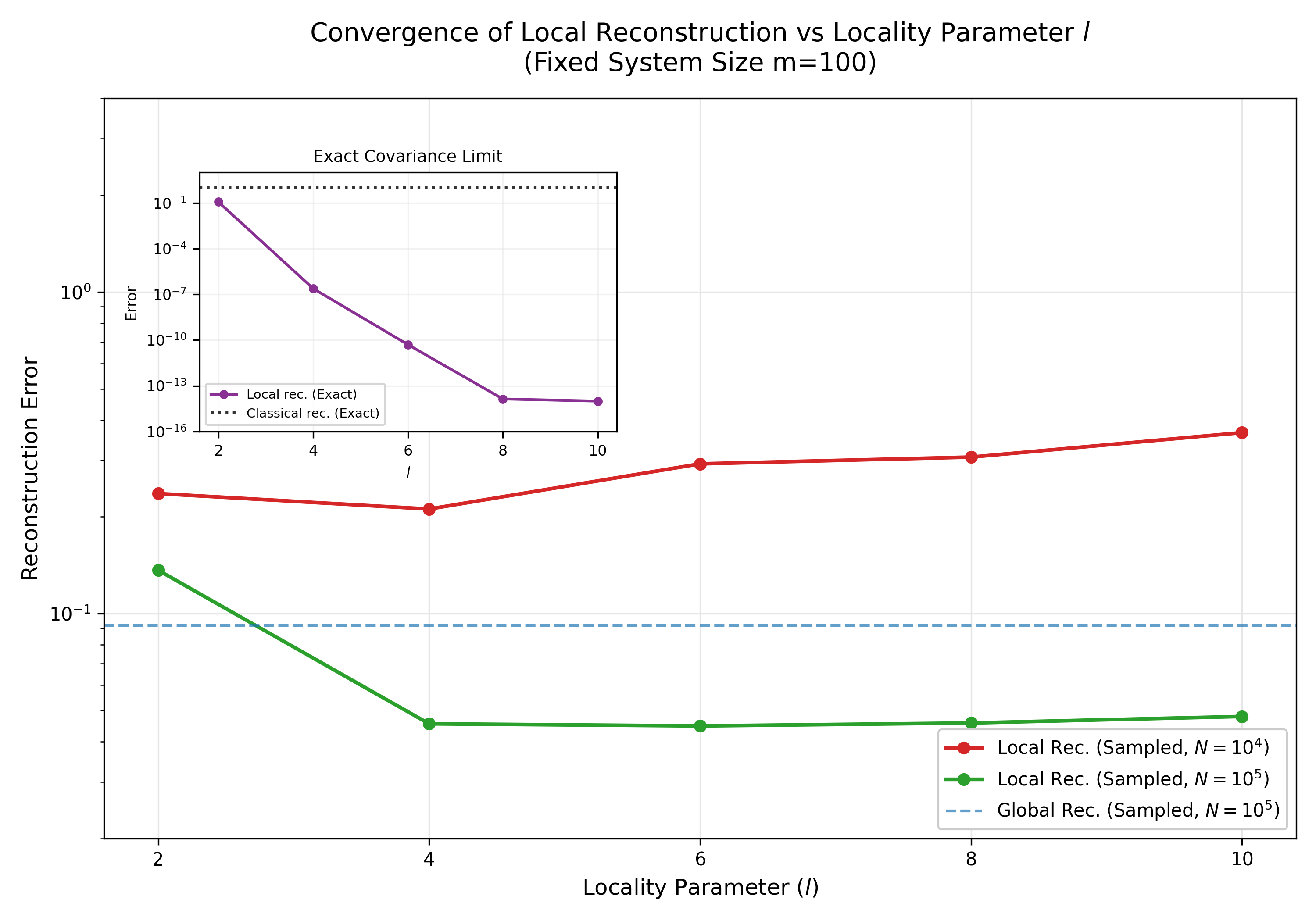}

    \caption{Entry-wise reconstruction errors vs. $l$ for the 1D well-conditioned Hamiltonian defined in Fig.~\ref{fig:simulation-errors-comparison} for $N=10^5$ and $N=10^4$ samples, against the error of the global reconstruction. We see that the local strategy outperforms the global one at $l=4$ and $N=10^5$. In the inset we plot the error given the exact  covariance matrix, showing the exponential improvement in precision as we increase $l$, together with the error obtained by estimating the Hamiltonian as $(2\hat{V}-i\Omega)^{-1}$, i.e. the classical formula.  }
    \label{fig:simulation-errors-l}
\end{figure}

Finally, we also show how to learn the graph of the Hamiltonian. In this problem we are given access to copies of the Gaussian states and the promise that its Hamiltonian is defined on a graph $G=([m],\mathsf{E})$ of maximal degree $\Delta-1$. Furthermore, we are given the promise that for each edge, at least one entry of the corresponding sub-matrix of the Hamiltonian is above a certain threshold $\kappa>0$. The goal is then to estimate $G$ with probability of success at least $1-\delta$.
The dependence of the sample complexity of the graph larning algorithm on the parameters $\|S\|_\infty$, $d_{\max}$, $d_{\min}$ is quite complex (and sometimes super-exponential), so we suppress it for simplicity in the next statement, which we show and state fully in Sec.~\ref{sec:graphlearn}:
\begin{thm}[Learning the graph of Gaussian states, informal]\label{th:informalgraph}
Let $H$ be a Hamiltonian with graph $G$ of degree $\Delta-1$ and edge set $\mathsf{E}$, the condition $0<\kappa\leq \min_{(i,j)\in \mathsf{E}}\max_{\delta_1,\delta_2\in\{0,1\}}|H_{2i-\delta_1,2j-\delta_2}|$. Then, it suffices to take
\begin{align}
    N= \mathcal{O}\left(\frac{1}{\kappa^{2+\gamma}}\ln\left(\frac{m}{\delta}\right)\right) \quad \forall \gamma>0 
\end{align}
copies of $\rho$ to learn the graph $G$ with probability of success at least $1-\delta${, using heterodyne measuerements to estimate the covariance matrix and Algorithm~\ref{alg:learnGG} to reconstruct the graph of the Hamiltonian.}
\end{thm}

{\begin{rem}\label{rem:complexityhamiltonlearning}The computational complexity of the above procedure is $O(\mathrm{poly}(m))$ when the other parameters are fixed, dominated by the classical post-processing of Algorithm~\ref{alg:learnGG}. While as in Algorithm~\ref{alg:learnHG} all the operation are only functions of matrices of size at most $2m\times 2m$, the number of such operations is a polynomial in $m$ with exponent of $m$ that depends on the other parameters. Most importantly, the loop between line 4 and 13 of Algorithm~\ref{alg:learnGG} can require a number of local inversions up to number of pairs of different neighborhood of a vertex $i$ of size $\Delta^l$. In the end, this gives an exponent of $m$ which is super-exponential in $\Delta$, exponential in $d_{\max}$ and subpolynomial in $\frac{1}{d_{\min}}$ and $\frac{1}{\kappa}$. See the proof of Theorem~\ref{theographlearning} for more details. 
\end{rem}
}
To the best of our knowledge, there are no results in this direction for quantum spin systems, marking the first quantum result of this kind in general. We once again emphasize the logarithmic scaling with the number of modes. Furthermore, even though the scaling with respect to other parameters is not efficient, this is a general feature of {Hamiltonian learning algorithms in the quantum case}.

\paragraph{Proof sketch:} The arguments to obtain our graph learning result also mostly rely on our local inversion techniques. The graph learning algorithm takes as an input an estimate of the covariance matrix $\hat{V}$, a promised lower bound on the strength of interaction $\kappa>0$, a promised upper bound on size of neighborhoods of distance $l~\sim\frac{\log(\kappa^{-1})}{\log\log(\kappa^{-1})}$ of vertices, denoted by $\xi(l)\le \Delta^l$, and a threshold parameter $\eta$, which depends on $\|S\|_{\infty}$ and the symplectic eigenvalues. 
It then iterates over all vertices $i$ and all possible neighborhoods of size $\xi(l)$, denoted by $\mathcal{N}_i$. For each of the $\mathcal{N}_i$, we then consider all possible neighborhoods enlarged by $\xi(l)$ additional vertices, defining a set $\overline{\mathcal{N}}_i$, and perform our local inversion with the enlarged neighborhood. If we see that in the resulting matrix there are no large entries (as measured by the threshold $\eta$) on row $i$ and vertices $j\in\overline{\mathcal{N}}_i$, we then set our neighborhood as $\mathcal{N}_i$. The intuition behind the procedure is as follows: recall that we assume that we have a lower bound on the size of all interactions and that matrix elements of $(2V-i\Omega)^{-1}$ corresponding to vertices that are far apart decay rapidly. 
Furthermore, performing local inversions with entries at a distance $l\sim \frac{\log(\eta^{-1})}{\log\log(\eta^{-1})}$ from a vertex suffices to obtain a good recovery of the entries up to error $\sim\eta$, which scales roughly as $\kappa$. If the $l$-neighborhood of $i$ is contained in $\mathcal{N}_i$, then all the estimates with the $\xi(l)$-enlarged neighborhood will be correct and all the entries added will be small because of the decay. Thus, we will correctly accept in this case. On the other hand, if we have the incorrect or incomplete $l$-neighborhood, then are $\xi(l)$ vertices we can add that will make it complete. 
For that joint neighborhood, we see two possibilities: either this will ``turn on'' a large entry, which will show us that we were not operating with the right neighborhood. In contrast, if we see that for all possibilities we add this never occurs, then we can also conclude that the vertices we are missing will only add a small error and can be neglected. By iterating over all such enlarged neighborhoods we are able to succesfully find submatrices to locally invert the covariance matrix for each vertex. After that, we inspect the resulting Hamiltonian and discard the interactions we find to be below the threshold given to us. Also here we are able to obtain a logarithmic scaling in the number of modes because we only need to invert locally.

{
\begin{rem}\label{rem:samplelattice}
Theorem~\ref{th:informallearning} and Theorem~\ref{th:informalgraph} both assume the worst-case growth for the neighborhoods at distance $l$ in a graph of degree $\Delta-1$, i.e. $\Delta^{l}$. However, $r$-dimensional lattices have a growth $gr^l$. By restricting to graphs with such polynomial growth of neighborhoods, one can obtain improved results, and in particular an $\mathcal{O}(\epsilon^{-2}\log^{2r}(\epsilon^{-1})\log(m)
)$ upper bounds on the sample complexity of Hamiltonian learning, see Theorem~\ref{th:polynomiallyham}. The same arguments could be used to obtain an $\tilde{\mathcal{O}}\left(\frac{1}{\kappa^2}\log m\right)$ sample complexity bound for graph learning with polynomially growing neighborhoods.
\end{rem}
}

\subsection{Comments and future directions}

The present work establishes a foundation for studying the problem of learning quantum Gaussian graphical models. Several key challenges remain, which may require refinements of the techniques proposed, or new approaches.

First of all, we make a remark on the (implicit) dependence of our result on the condition number of $H$. 
In the sample complexity for Hamiltonian and graph learning, we see a dependence on $d_{\min}$, which can be interpreted as a ``maximum temperature'' parameter of the Gaussian state. Since $d_{\min}$ gives a smaller contribution to the Hamiltonian matrix, it is natural to ask if this dependence can be removed. In other quantum Hamiltonian learning results for discrete variable systems \cite{bakshi2024learning,haah2022optimal,2004.07266,narayanan2024improved}, even if with they display polynomial dependence on the system size, a dependence of this type does not appear.
In fact, the dependence of our learning guarantees on the condition number can also be expressed in terms of potentially much smaller quantities, but 
we preferred to state our results as we did to avoid unnatural assumptions. Remarks~\ref{remarkV1},~\ref{remarkV2},~\ref{remarkcondition} in the technical sections clarify this point. 
In the analogue classical problem of learning the precision matrix (the term used in this context for the Hamiltonian matrix) or the graph, the dependence on the condition number was removed only by~\cite{misra20a}, crucially using multiplicative error bounds on covariance estimation. In the graph learning problem, existence of edges is there formulated in terms of relative strength: two vertices $(i,j)$ are connected if Hamiltonian matrix element $H_{2i-\delta,2j-\delta'}$ is sufficiently large compared with the diagonal elements $H_{2i-\delta,2i-\delta'}$ and $H_{2j-\delta,2j-\delta'}$ (here $\delta,\,\delta'$ can be $0$ or $1$ and index position and momentum quadratures). The sample complexity depends on the promise on the relative strength but not on the condition number. We leave exploring multiplicative error bounds in the quantum case for future work. See Appendix~\ref{condnumbex} for an argument suggesting that with additive precision guarantees one cannot avoid the condition number dependence.

We also point out the following natural developments of our work.

\begin{itemize}
\item The dependence of the sample complexity on $\epsilon$ is not exactly the expected $\epsilon^{-2}$ for the Hamiltonian and graph learning. While this can be remedied for graph whose neighborhoods increase polynomially with distance, in the general case this could be a limitation of our approach based on a Taylor expansion. It would be very interesting to understand if this can be overcome by other techniques.

\item A question we leave open is to find lower bounds on the sample complexity of both Hamiltonian estimation and graph learning. It would be especially interesting to understand how these bounds differ from the ones that can be obtained in the classical case \cite{Wang2010}.

\item We address the Hamiltonian learning problem by plugging-in estimates of the covariance matrix, via a simple single-copy measurement procedure. It would be interesting to investigate if entangled measurements could be beneficial to the estimation of the covariance matrix, and possibly even more in the estimation of the Hamiltonian matrix. 
\item In the classical literature, a variety of methods have been developed to estimate gaussian mixtures (e.g.~\cite{Liu2023}). It would be a natural generalization to develop these methods to the estimation of mixtures of Gaussian states, which are states that can be prepared as efficiently as Gaussian states themselves and can be relevant in realistic scenarios, for example in presence of phase-noise.
\item It would be interesting to see if our results can be extended to Gaussian fermionic states.
\item Finally, it would be natural if our ideas can be used to efficiently learn the parameters and structure of local Gaussian channels.
\end{itemize}

\section{Methods}
\subsection{Estimating covariance matrices and learning Gaussian states in trace distance}
In this section we will discuss the sample and computational complexity of two fundamental problems concerning the learning of Gaussian quantum states: estimating their covariance matrix and learning them in trace distance. The reason we will consider them together is that we show a new method to transfer a bound on the distance between covariance matrices to a bound in trace distance. Furthermore, estimates on the covariance matrix will be central to the other problems we consider in this paper, namely Hamiltonian and graph learning. Let us thus start with estimating the covariance matrix:

\subsubsection{Estimating the entries of the covariance matrix}\label{sec:learning_covariance}
Our approach to estimate the entries of the covariance matrix is based on performing heterodyne measurements on the Gaussian state~\cite{serafini2017quantum}, which are experimentally feasible on the majority of platforms. The heterodyne measurement is a POVM constructed from the overcomplete set of coherent states, defined as:

\begin{equation}\begin{aligned}\hspace{0pt}
\ket{x} := \W(x) \ket{0}\, ,
\label{coherent}
\end{aligned}\end{equation}
where $\ket{0}$ is the vacuum state, with characteristic function 
\begin{align*}
\chi_{|0\rangle\langle 0|}(u)= e^{-\frac{1}{4}\|u\|^2}\,.
\end{align*}

It follows that the characteristic function of the coherent state $|x\rangle\langle x|$ takes the form
\begin{align*}
\chi_{|x\rangle\langle x|}(u)= e^{-\frac{1}{4}\|u\|^2+iu^\intercal x}\,.
\end{align*}

From now on, for simplicity, we denote $V\equiv V[\rho]$ and $t\equiv t[\rho]$. Let us start by recalling the probability density of observing outcome $x\in \R^{2m}$ when measuring the Gaussian state $\rho$, which we denote by $\rho(x)$, and which evaluates to $\rho(x)=\frac{1}{(2\pi)^m}\Tr[|x\rangle\langle x|\rho]$. 

By Plancherel's theorem, 
\begin{align*}
\rho(x)&=\frac{1}{(2\pi)^m}\int \chi_{|x\rangle\langle x|}(u)^*\chi_\rho(u)\,d^{2m}u\\
&=\frac{1}{(2\pi)^m}\int  e^{-\frac{1}{4}\|u\|^2-iu^\intercal x}e^{it^\intercal u-\frac{1}{2}u^\intercal Vu}\,d^{2m}u\\
&=\frac{1}{(2\pi)^m}\,\int e^{-\frac{1}{2}u^\intercal(V+I/2)u+i(t-x)^\intercal u}d^{2m}u\\
&=\operatorname{det}(V+I/2)^{-\frac{1}{2}}\,e^{-\frac{1}{2}(t-x)^\intercal (V+I/2)^{-1}(t-x)}\,,
\end{align*}
where the last equation follows by direct Gaussian integration. In words, the heterodyne measurement of the state $\rho$ generates a $2m$ dimensional Gaussian random vector $\mathbf{Z}:=(Z_1,\dots, Z_{2m})$ with mean $t$ and covariance $\Sigma=V+I/2$:
\begin{align*}
t_i:=\mathbb{E}[Z_i],\qquad (V+I/2)_{i,j}=\mathbb{E}[Z_iZ_j]-t_it_j\,.
\end{align*}
We therefore have $2m+(2m)^2=2m(2m+1)$ coefficients which we can estimate by simple empirical averages: assuming access to $N$ copies of the state $\rho$, after heterodyne measurement leading to the i.i.d. vectors $\mathbf{Z}^{(1)},\dots ,\mathbf{Z}^{(N)}$, we construct the estimators
\begin{align*}
\hat{t}_i:=\frac{1}{N}\sum_{k=1}^N\,Z_i^{(k)}\,\qquad \hat{V}^{(0)}_{i,j}:=\frac{1}{N}\sum_{k=1}^NZ^{(k)}_iZ^{(k)}_j\,.
\end{align*}
By Gaussian concentration, we have that for any $i\in[2m]$
\begin{align*}
\mathbb{P}\big(|\hat{t}_i-t_i|\ge \epsilon\big)\le 2e^{-\frac{N\epsilon^2}{2V_{i,i}+1}}\,.
\end{align*}
Moreover, for any $i,j\in[2m]$, and $q\ge 1$,
\begin{align*}
\mathbb{E}\left[(Z_iZ_j)^{2q}\right]&\equiv \left\|Z_iZ_j\right\|_{2q}^{2q}\le \left\|Z_i\right\|_{4q}^{2q}\,\left\|Z_j\right\|_{4q}^{2q}\\
&\le \left(\|Z_i-t_i\|_{4q}+|t_i|\right)^{2q}\left(\|Z_j-t_j\|_{4q}+|t_j|\right)^{2q}\\
&\le (2q)!\, \left[\left(\sqrt{2V_{i,i}+1}+|t_i|\right)\left(\sqrt{2V_{j,j}+1}+|t_j|\right)\right]^{2q},
\end{align*}
where we used that the centered $4q$-th moment of a Gaussian random variable with variance $\sigma^2$ is $\sigma^{4q}(4q-1)!!$.
Denoting by $C[V,t]:=\max_{i}2\left(\sqrt{2V_{i,i}+1}+|t_i|\right)^2$, we have thus that, by triangle inequality and Jensen's inequality,

\begin{align*}
\mathbb{E}\left[\left(Z_iZ_j-\mathbb{E}\Big[Z_iZ_j\Big]\right)^{2q}\right]\le \mathbb{E}\left[(Z_iZ_j)^{2q}\right]+(\mathbb{E}[Z_iZ_j])^{2q} \le C[V,t]^{2q}(2q)!\,.
\end{align*}

and therefore the centered random variables $Z_iZ_j-\mathbb{E}\Big[Z_iZ_j\Big]$ are sub-Gamma with parameters $(4C[V,t]^2, 2C[V,t])$ by \cite[Theorem 2.3]{Boucheron2004}. This implies by Cramér-Chernoff method that, for $s<1$,
\begin{align*}
\mathbb{P}\big(|\hat{V}^{(0)}_{i,j}-\mathbb{E}[Z_iZ_j]|\ge 2\sqrt{8C[V,t]^2s}\big)&\le\mathbb{P}\big(|\hat{V}^{(0)}_{i,j}-\mathbb{E}[Z_iZ_j]|\ge \sqrt{8C[V,t]^2s}+2C[V,t]s\big)\\&\le 2e^{-Ns}\,.
\end{align*}

Thus, for $\epsilon=2\sqrt{8C[V,t]^2s}$,
\begin{align*}
\mathbb{P}\big(|\hat{V}^{(0)}_{i,j}-\mathbb{E}[Z_iZ_j]|\ge \epsilon\big)&\le 2e^{-\frac{N\epsilon^2}{32C[V,t]^2}}\,.
\end{align*}

Hence, by a union bound and using the fact that $C[V,t]^2\geq\max_i 2V_{i,i}+1$, for $\delta\in(0,1)$, and choosing $N\ge\frac{32C[V,t]^2}{\epsilon^2}\log\left(\frac{4m(2m+1)}{\delta}\right)$, we have that, with probability $1-\delta$,
\begin{align*}
|t_i-\hat{t}_i|\le \epsilon,\qquad |\mathbb{E}[Z_iZ_j]-\hat{V}^{(0)}_{i,j}|\le \epsilon\,,\qquad \forall i,j\in[2m]\,.
\end{align*}
This further implies that, denoting
\begin{align*}
\hat{V}_{i,j}:=\hat{V}_{i,j}^{(0)}-\hat{t}_i\hat{t}_j-\frac{\delta_{i,j}}{2}
\end{align*}
we get with the same probability,
\begin{align*}
\big|\hat{V}_{i,j}-V_{i,j}\big|\le \epsilon+\epsilon^2+2\epsilon\max_i |t_i|\,.
\end{align*}
By rescaling the error probability, we directly get the following
\begin{lem}\label{lem:cov_matrix}
In the notations of the previous paragraph, for any $\delta\in (0,1)$, $\epsilon\in (0,1/2)$ and $N\ge \frac{2^8C[V,t]^2(1+\max_i|t_i|)^2}{\epsilon^2}\log\left(\frac{4m+1}{\delta}\right)$ with $C[V,t]:=\max_{i}2\left(\sqrt{2V_{i,i}+1}+|t_i|\right)^2$,
\begin{align}\label{equ:estimates_close}
|t_i-\hat{t}_i|\le \epsilon,\qquad \text{ and }\qquad |\hat{V}_{i,j}-V_{i,j}|\le \epsilon\,,\qquad \forall i,j\in[2m]\,.
\end{align}
\end{lem}
We will later impose constraints on the set of Gaussian states we consider which will allow us to derive upper bounds on the value of $C[V,t]$.

In Lemma \ref{lem:cov_matrix} we established the sample complexity of obtaining an empirical estimate of the covariance matrix such that each entry is close to the true covariance matrix with high probability. Note, however, that the matrix $\hat{V}$ will not necessarily be a valid covariance matrix, as this would require that $2\hat{V}\geq \pm i\Omega$. But, given an estimate $\hat{V}$ and a precision parameter $\epsilon$, this issue can be readily fixed by solving the SDP:
\begin{align}\label{equ:SDP_Covariance}
    &\max\limits_{\tilde{V}\in\mathbb{R}^{2m\times 2m}} 1\\\nonumber
    \forall i,j\in[2m]: &\tilde{V}_{i,j}-\hat{V}_{i,j}\leq\pm \epsilon  \\
    & 2\widetilde{V}\geq \pm i\Omega\nonumber\\
    & \widetilde{V}\geq 0\nonumber
\end{align}
Note that conditioned on the estimates on $V_{i,j}$ being $\epsilon$ correct in the sense of Equation~\eqref{equ:estimates_close}, the SDP in Equation~\eqref{equ:SDP_Covariance}
is feasible, as the true $V$ is feasible. Furthermore, letting $\widetilde{V}$ be any feasible point of Equation~\eqref{equ:SDP_Covariance}, it readily follows from a triangle inequality that $|V_{i,j}-\widetilde{V}_{i,j}|\leq 2\epsilon$ for all $i,j\in[2m]$, and by the other constraints $\widetilde{V}$ is a valid covariance matrix.
As SDPs can be solved in time polynomial in the dimension of the underlying matrices and number of constraints, we conclude that:
\begin{lem}\label{lem:cov_matrix_2}
In the setting of Lemma~\ref{lem:cov_matrix}, for any $\delta\in (0,1)$, $\epsilon\in (0,1/2)$ and $$N=\mathcal{O}\left( \frac{C[V,t]^2(1+\max_i|t_i|)^2}{\epsilon^2}\log\left(\frac{4m+1}{\delta}\right)\right)\,,$$ we can obtain a bona-fide covariance matrix $\widetilde{V}$ satisfying 
\begin{align}\label{equ:estimates_closesdp}
|\widetilde{V}_{i,j}-V_{i,j}|\le 2\epsilon\,,\qquad \forall i,j\in[2m]\,.
\end{align}
in polynomial time in $m$ by solving the SDP in Equation~\eqref{equ:SDP_Covariance}.
\end{lem}

We will later {(for Theorem~\ref{th:plug-intrbound})} need similar results to convert empirical estimates of a Hamiltonian into a bona-fide Hamiltonian. Recall that Hamiltonians $H$ correspond to strictly positive symmetric matrices. However, $H>0$ is not a semidefinite constraint and we will impose the stronger constraint $H\geq \tau I$ on $H$ to make sure that finding a valid Hamiltonian from empirical estimates can be done efficiently. As we will discuss later in Sec.~\ref{sec:symplectic_diag}, given a lower bound $d_{\min}$ on the symplectic eigenvalues of $H$ and the symplectic matrix $S$ that diagonalizes $H$, we can take $\tau\geq \|S\|^{-2}d_{\min}$ and $H$ satisfies $H\geq \tau I$. Thus, given a (not necessarily positive) $\hat{H}$ satisfying:
\begin{align}
    \forall i,j\in[2m]: |\hat{H}_{i,j}-H_{i,j}|\leq \epsilon,
\end{align}
and a threshold $\tau>0$, we can obtain a symmetric and positive $\widetilde{H}$ (and thus valid Hamiltonian) satisfying $|\tilde{H}_{i,j}-H_{i,j}|\leq 2\epsilon$ by solving the SDP:
\begin{align}\label{equ:SDP_Hamiltonian}
    &\max\limits_{\widetilde{H}\in\mathbb{R}^{2m\times 2m}} 1\nonumber\\
    \forall i,j\in[2m]: &\widetilde{H}_{i,j}-\hat{H}_{i,j}\leq\pm \epsilon  \\
    & H\geq \tau I \,.\nonumber
\end{align}
Arguing as before, we conclude that this SDP can also be solved in time that is polynomial in the number of modes and logarithmic in the desired precision $\epsilon$ and threhsold $\tau$. Again, this is feasible as $H$ itself is a solution. In addition, it is possible to only optimize over Hamiltonians that have the same underlying graph structure as $H$ by imposing the additional linear constraints $\widetilde{H}_{i,j}=0$ for $(i,j)\in \mathsf{E}$.

\begin{rem}[Robustness to measurement imperfections]\label{rem:experimental_noise}
Our analysis assumes ideal heterodyne measurements where the samples follow a distribution with covariance $V+I/2$. In realistic experiments, the dominant imperfections are finite detection efficiency $\eta \in (0,1]$ and electronic noise $v_{\mathrm{el}} \ge 0$~\cite{Weedbrook2012}. These are modeled by Gaussian channels, resulting in a measured covariance $\Sigma_{\mathrm{meas}} = \eta V + (1 - \eta/2 + v_{\mathrm{el}}) I$~\cite{serafini2017quantum}. Since the measurement noise is Gaussian, the samples remain multivariate normal, preserving the validity of the estimators derived above.
These imperfections can be addressed in two ways. First, if $\eta$ and $v_{\mathrm{el}}$ are known, one can construct an unbiased estimator for $V$ by rescaling the empirical covariance, which increases the sample complexity by a constant factor $\eta^{-2}$ but does not alter the scaling with respect to the system size $m$. Second, for small, uncharacterized noise, our protocols are robust with respect to Gaussian noise that changes the covariance matrix by an entry-wise error that is comparable with the admissible statistical error, which is independent on system-size.  
Note also that our concentration bounds only used the sub-Gaussian and sub-Gamma properties of the estimators. If the implementation of the heterodyne measurement is not Gaussian (in fact, the standard implementation requires an infinite-energy local oscillator) but the estimators have expectation values sufficiently close to $t$ and $V$, and above sub-Gaussian and sub-Gamma properties are still satisfied, the conclusions of the analysis do not change.
\end{rem}

\subsubsection{Learning in trace distance}\label{sec:tracelearn}
In this section, we show how to learn an unknown Gaussian state in trace distance under certain constraints.

\begin{defi}[Problem 1: Gaussian trace distance learning]\label{def:trace_dist_learning}
Let $H\in M^{>0}_{2m}(\mathbb{R})$ be a symmetric, positive definite matrix and $t\in \mathbb{R}^{2m}$ (later referred to as Gaussian Hamiltonian). For any such data $(H,t)$, one can construct a Gaussian state on $m$ modes: $$\rho(H,t)=\frac{e^{-(R-t)^{\intercal}H(R-t)}}{\Tr[e^{-(R-t)^{\intercal}H(R-t)}]}\,.$$ 
The Gaussian trace distance learning problem consists in the following task: given a fixed precision $\epsilon>0$, provide estimates $\hat H$ of $H$ and $\hat t$ of $t$ such that $\|\rho(\hat{H},\hat{t})-\rho(H,t)\|_{1}\leq 2\epsilon$ with probability at least $1-\delta$, given $n$ copies of $\rho(H,t)$, for any $(H,t)$ in some specified subset $\mathcal{S}$. We denote by $\mathcal{N}_{\mathrm{TDL}}(\mathcal{S},m,\epsilon,\delta)$ the sample complexity for Gaussian trace distance learning. This corresponds to the minimum number $n$ of copies of the state $\rho(H,t)$ required for the task to succeed.  
\end{defi}

In this section we assume we have an estimate of the true covariance matrix and do not impose any graph structure for the underlying Hamiltonians (they could, e.g.,~be all-to-all connected).  Importantly, we will obtain sample complexities that scale polynomially with the number of modes and quadratically with $\epsilon^{-1}$, combined with polynomial postprocessing. 
This marks the first result on learning Gaussian states with such scaling~\cite{mele2024learning}, although unlike previous results, we need to impose some constraints on the Gaussian states we consider in addition to a bounded energy constraint, which is in fact implied by our constraints.

We start with the standard observation that it is possible to bound the trace distance between thermal states in terms of the entries of $V,t$ and $H$.
\begin{prop}\label{th:symrelent}
Given two Gaussian states $\rho$ and $\hat{\rho}$, $\rho\equiv \rho(t,H)$, $W\equiv \W(-\Omega^{-1}t)$, $\hat{\rho}\equiv \rho(\hat{t},\hat{H})$ and $\hat{W}\equiv \W(-\Omega^{-1}\hat{t})$, and denoting $T_{i,j}:=t_it_j+\hat{t}_i\hat{t}_j-\hat{t}_it_j-t_i\hat{t}_j=(t_{i}-\hat{t}_i)(t_{j}-\hat{t}_j)$
we have

\begin{align}\label{equ:trace_from_Hamilt}
\|\rho-\hat{\rho}\|_1&\le\sqrt{\sum_{i,j}(\hat{H}_{i,j}-H_{i,j})(V_{i,j}-\hat{V}_{i,j})+2T_{i,j}(H_{i,j}+\hat{H}_{i,j})}
\end{align}

\end{prop}

\begin{proof}
Using the symmetric relative entropy, we get from Pinsker's inequality that:
\begin{align*}
\|\rho-\hat{\rho}\|_1&\le \sqrt{2D(\rho\|\hat{\rho})+2D(\hat{\rho}\|\rho)}=\sqrt{2 \tr{(\rho-\hat{\rho})(\hat{W}R^\intercal \hat{H}R\hat{W}^\dagger-WR^\intercal HRW^\dagger)}}\\
&= \sqrt{2\sum_{i,j}\tr{(\rho-\hat{\rho})(\hat{H}_{i,j}(R_i-\hat{t}_i)(R_j-\hat{t}_j)-H_{i,j}(R_i-t_i)(R_j-t_j))}}\\
&=\sqrt{\sum_{i,j}\hat{H}_{i,j}\left(\tr{\rho\{R_i-\hat{t}_i,R_j-\hat{t}_j\}}-\hat{V}_{i,j}\right)-\sum_{i,j}H_{i,j}\left(V_{i,j}-\tr{\hat{\rho}\{R_i-t_i,R_j-t_j\}}\right) }\\
&=\sqrt{\sum_{i,j}\hat{H}_{i,j}\Big(V_{i,j}-\hat{V}_{i,j}+2T_{i,j}\Big)-\sum_{i,j}H_{i,j}\Big(V_{i,j}-\hat{V}_{i,j}-2T_{i,j}\Big)}\\
&=\sqrt{\sum_{i,j}(\hat{H}_{i,j}-H_{i,j})(V_{i,j}-\hat{V}_{i,j})+2T_{i,j}(H_{i,j}+\hat{H}_{i,j})}
\end{align*}
where recall that $WR_iW^\dagger=R_i- tI$\,.

\end{proof}
Clearly, if we can obtain estimates of the covariance matrix that entrywise close, we can also bound the trace distance. Furthermore, it is to be expected that, as covariance matrices get closer, the corresponding Hamiltonians also converge. Such a continuity estimate is crucial to obtain an $\epsilon^{-2}$ in the sample complexity. One possible path to showing such a statement is discussed for spin systems in~\cite{2004.07266}, where the authors obtain such statements in terms of the strong convexity properties of the partition function of states. We follow a different route and show several continuity bounds directly in Sec.~\ref{sec:conntinuity_bounds}. From these we obtain:

\begin{thm}\label{th:plug-intrbound}
Let $H=\frac{1}{2}\log\left(\frac{2i\Omega V+I}{2i\Omega V-I}\right)i\Omega$ and $\hat{H}=\frac{1}{2}\log\left(\frac{2i\Omega \hat{V}+I}{2i\Omega \hat{V}-I}\right)i\Omega$ for two bona fide covariance matrices $V$, $\hat{V}$. We also denote by $d_{\max}$, resp. $S$, the maximal symplectic eigenvalue of $H_1$, resp. the symplectic matrix associated to $H_1$ through Equation \eqref{eqsymplect}. 
Let $\rho$ and $\hat{\rho}$ be Gaussian states with $\rho\equiv \rho(t,H)$, $W\equiv \W(-\Omega^{-1}t)$, $\hat{\rho}\equiv \rho(\hat{t},\hat{H})$ and $\hat{W}\equiv \W(-\Omega^{-1}\hat{t})$.
Given an approximation $\hat{V}$ of the covariance matrix $V$ as well as an approximation $\hat{t}$ of the mean vector $t$ to precision $\epsilon$: $\max_{i,j}|V_{i,j}-\hat{V}_{i,j}|\le\epsilon$, and $\max_i|t_i-\hat{t}_i|\le \epsilon$, such that $2\|V-\hat{V}\|_{\infty}< \frac{e^{-2d_{\max}}}{1-e^{-2d_{\max}}} \|S\|_{\infty}^{-2}$
we have 
\begin{align}\label{equ:trace_dist_bound}
\|\rho-\hat{\rho}\|_1&\le m\epsilon\sqrt{24m(e^{2d_{\max}}-1)^2\|S\|_{\infty}^4+16\|S\|_{\infty}^2 d_{\max}}\,.
\end{align}

\end{thm}
\begin{proof}
Via Theorem~\ref{th:continuityhamv}, we have   
\begin{align*}
\|H-\hat{H}\|_\infty
&\leq 2\left(2\frac{e^{-2d_{\max}}}{1-e^{-2d_{\max}}} \|S\|_{\infty}^{-2}\right)^{-1}\left(2\frac{e^{-2d_{\max}}}{1-e^{-2d_{\max}}} \|S\|_{\infty}^{-2}-2\|V-\hat{V}\|_{\infty}\right)^{-1}\|V-\hat{V}\|_{\infty}\,,
\end{align*}

which implies
\begin{align*}
\|H-\hat{H}\|_\infty
&\leq 8m\left(2\frac{e^{-2d_{\max}}}{1-e^{-2d_{\max}}} \|S\|_{\infty}^{-2}\right)^{-2}\epsilon=2m(e^{2d_{\max}}-1)^2 \|S\|_{\infty}^{4}\epsilon\,.
\end{align*}

Therefore, since $|H_{k,k'}-\hat{H}_{k,k'}|\leq \|H_1-H_2\|_\infty$, we have via Proposition~\ref{th:symrelent} and its notation

\begin{align*}
\|\rho-\hat{\rho}\|_1&\le\sqrt{8 m^3(e^{2d_{\max}}-1)^2 \|S\|_{\infty}^{4}\epsilon^2+8m^2\epsilon^2(2\|S\|_{\infty}^2 d_{\max}+2m(e^{2d_{\max}}-1)^2\|S\|_{\infty}^4 \epsilon)}\nonumber\\
&\leq m\epsilon\sqrt{24m(e^{2d_{\max}}-1)^2\|S\|_{\infty}^4+16\|S\|_{\infty}^2 d_{\max}}\,.
\end{align*}
\end{proof}

Combined with our results on estimating the covariance matrix in Sec.~\ref{sec:learning_covariance}, this immediately yields:
\begin{cor}
Let $\rho(t,H)$ be a Gaussian state on $m$ modes s.t. the largest symplectic eigenvalue $d_{\max}$ satisfies $d_{\max}\leq \beta_{\max}$, the smallest symplectic eigenvalue $d_{\min}$ satisfies $d_{\min}\geq \beta_{\min}$, and $\|S\|_{\infty}\leq s$ for known $s,\beta_{\max},\beta_{\min}$. Then we have
\begin{align}\label{boundestimatevar}
C[V,t]^2(1+\max_i|t_i|)^2&\leq \max_{i}16\left(2V_{i,i}+1+|t_i|^2\right)^2(1+\max_i|t_i|)^2\\&=\mathcal{O}(s^4(1-e^{-2\beta_{\min}})^{-2}(1+\max_{i}|t_i|^6))\,,
\end{align}
and for 
$1>\epsilon,\delta>0$,
\begin{align}\label{samplecompltrace}
    N\geq \Omega(\epsilon^{-2}m^3s^6(e^{2\beta_{\max}}-1)^2(1-e^{-2\beta_{\min}})^{-2}(1+\max_{i}|t_i|^6)\log(m\delta^{-1})),
\end{align}
copies of $\rho$ suffice to output an estimate of a bona-fide covariance matrice $\hat{V}$, mean vector $\hat{t}$ and Hamiltonian $\hat{H}$ such that the corresponding Gaussian state is $\epsilon$-close in trace distance to $\rho$ with probability at least $1-\delta$. Furthermore, the postprocessing is polynomial in $m$. This implies $\mathcal{N}_{\mathrm{TDL}}(\mathcal{S},m,\epsilon,\delta)\leq N$ for the class of states $\mathcal{S}$ satisfying the above conditions.

\end{cor}
\begin{proof}
Equation~\eqref{boundestimatevar} is a consequence of the bound on $\|V\|_{\infty}$ from Equation~\eqref{eq:normbound1}. Equation \eqref{samplecompltrace} follows from noting that Equation~\eqref{equ:trace_dist_bound} implies that if the covariance matrix $\hat{V}$ satisfies $|V_{i,j}-\hat{V}_{i,j}|\leq \epsilon'$ with $\epsilon'=\mathcal{O}(\epsilon m^{-1.5}(s^2\beta_{\max}+(e^{2\beta_{\max}}-1)^2 s^{4})^{-\tfrac{1}{2}})$ (which also implies that the condition  $2\|V-\hat{V}\|_{\infty}< \frac{e^{-2d_{\max}}}{1-e^{-2d_{\max}}} \|S\|_{\infty}^{-2}$ is satisfied if $\epsilon<1$), then the two states are $\epsilon$-close in trace distance. It then follows from Lemma~\ref{lem:cov_matrix_2} that we can obtain such a covariance matrix with the claimed sample complexity. The fact that the postprocessing is efficient follows from the fact that we just need to take the empirical averages and then solve an SDP to obtain the covariance matrix. If we wish to output the Hamiltonian, instead, we only need to perform an additional matrix inversion step, which is again polynomial in $m$.
\end{proof}

\begin{rem}\label{remarkV1} Bounding $V_{ii}$ by $\|V\|_{\infty}$ may be overshooting in the regime of large $m$. For example, for a covariance matrix with many entries of $\mathcal{O}(1)$, the norm can still be of $\Omega(m)$. The sample complexity can be as well expressed in terms of a bound on the diagonal elements of $V$.
\end{rem}
{
\begin{rem}\label{remarkenergy}
The result above already illustrates the power of our methods, as this the first general result on learning Gaussian states in trace distance with sample complexity scaling quadratically in precision. Thanks to Eq.~(\ref{eq:energy}), we can also write $C[V,t]^2(1+\max_i|t_i|)^2=O(E^3)$ and using also Eq.~(\ref{ineq:energy2}) and the steps in the proof above, we obtain that $N=O\left(\frac{m^3 E^4(e^{2\beta_{\max}}-1)^2}{\epsilon^2}\log{\frac{m}{\delta}}\right)$ are sufficient for the task. 
\end{rem}
}

\subsection{Continuity bounds for Hamiltonians and covariances}\label{sec:conntinuity_bounds}
In this section, we derive useful perturbation bounds for Gaussian Hamiltonians in terms of the norm difference between their corresponding covariance matrices, and vice versa. As we saw in the previous section, such bounds are essential when learning Gaussian states.
We will use the following simple bound for the matrix inverse at multiple occasions:
\begin{equation}\label{inversebound}
\|A^{-1}-B^{-1}\|_{\infty}\leq \|A^{-1}\|_{\infty}\|B^{-1}\|_{\infty} \|A-B\|_{\infty},
\end{equation}
which readily follows from the identity $A^{-1}-B^{-1}=A^{-1}(B-A)B^{-1}$ and the submultiplicativity of the operator norm. We will also make use of the following integral representation of the matrix logarithm, which holds for $A$ such that the spectrum $\mathrm{spec}(A)$ does not intersect with $(-\infty,0]$:
\begin{equation}\label{logintegrep}
\log(A)=(A-I)\int_{0}^{\infty}\frac{1}{tA+I}\frac{dt}{t+1}\,.
\end{equation}
\begin{lem}\label{th:continuityhamv} For $i\in\{1,2\}$, let $H_i=\frac{1}{2}\log\left(\frac{2i\Omega V_i+I}{2i\Omega V_i-I}\right)i\Omega$. We also denote by $d_{\max}$, resp. $S$, the maximal symplectic eigenvalue of $H_1$, resp. the symplectic matrix associated to $H_1$ through Equation \eqref{eqsymplect}. Then, assuming that $\|V_1-V_2\|_{\infty}< \frac{e^{-2d_{\max}}}{1-e^{-2d_{\max}}} \|S\|_{\infty}^{-2}$, 
\begin{align*}
\|H_1-H_2\|_\infty
&\leq 2\left(2\frac{e^{-2d_{\max}}}{1-e^{-2d_{\max}}} \|S\|_{\infty}^{-2}\right)^{-1}\left(2\frac{e^{-2d_{\max}}}{1-e^{-2d_{\max}}} \|S\|_{\infty}^{-2}-2\|V_1-V_2\|_{\infty}\right)^{-1}\|V_1-V_2\|_{\infty},
\end{align*}

\end{lem}

\begin{proof}

Using \eqref{logintegrep}, for $H\in\{H_1,H_2\}$, we have

\begin{align*}
H&\equiv \frac{1}{2}\log \left(I+\frac{2}{2i\Omega V-I}\right)i\Omega\\
&=\left(\frac{1}{2i\Omega V-I}\int_{0}^{\infty}\frac{1}{t\frac{2i\Omega V+I}{2i\Omega V-I}+I}\frac{dt}{t+1}\right)i\Omega\\
&=\int_{0}^{\infty}\frac{1}{t(2i\Omega V+I)+(2i\Omega V-I)}\frac{dt}{t+1}i\Omega\\
&=\int_{0}^{\infty}\frac{1}{2i\Omega V+\frac{t-1}{t+1}I}\frac{dt}{(t+1)^2}i\Omega.
\end{align*}
Then
\begin{align*}
\|H_1-H_2\|_\infty&=\left\|\int_{0}^{\infty}\frac{1}{{2}i\Omega V_1+\frac{t-1}{t+1}I}\frac{dt}{(t+1)^2}-\int_{0}^{\infty}\frac{1}{{2}i\Omega V_2+\frac{t-1}{t+1}I}\frac{dt}{(t+1)^2}\right\|_{\infty}\\
&\le \, 2\|V_1-V_2\|_\infty\int_0^\infty\, 
\left\|\left({2}i\Omega V_2+\frac{t-1}{t+1}I\right)^{-1}\right\|_\infty \left\|\left({2}i\Omega V_1+\frac{t-1}{t+1}I\right)^{-1}\right\|_\infty 
\frac{dt}{(t+1)^2}\\
&\leq \left(2\frac{e^{-2d_{\max}}}{1-e^{-2d_{\max}}} \|S\|_{\infty}^{-2}\right)^{-1}\left(2\frac{e^{-2d_{\max}}}{1-e^{-2d_{\max}}} \|S\|_{\infty}^{-2}-2\|V_1-V_2\|_{\infty}\right)^{-1}2\|V_1-V_2\|_{\infty}\,,
\end{align*}
where the first inequality follows from \eqref{inversebound}, whereas the last one follows from \eqref{eq:minsing1} and the integral $\int \frac{dt}{(t+1)^2}=1$.
\end{proof}

We also have the following bound.

\begin{lem}\label{th:backpert}
 For $i\in\{1,2\}$, we let $K_i=\frac{2}{2i\Omega V_i-I}$. In the notations of the previous theorem, if $\|K_1-K_2\|_\infty< e^{-2d_{\max}}\|S\|_{\infty}^{-2}$, then
\begin{align*}
\|H_1-H_2\|_{\infty}&\leq {\frac{1}{2}}e^{{{2}}d_{\max}}\|S\|^2_{\infty}\|K_1-K_2\|_{\infty}\left(1+\frac{e^{2d_{\max}}\|S\|_{\infty}^2(1-e^{-2d_{\max}})+\|K_1-K_2 \|_{\infty}}{e^{-2d_{\max}} \|S\|_{\infty}^{-2}-\|K_1-K_2 \|_{\infty}}\right).
\end{align*}
\end{lem}

\begin{proof}
By writing $K_i=\frac{2}{2i\Omega V_i-I}$, we also have
\begin{align}\label{eq:backtoh}
H_i=\frac{1}{2}\log(I+K_i)i\Omega={\frac{1}{2}}\,K_i\int_{0}^{\infty}\frac{1}{I+\frac{t}{t+1}K_i}\frac{dt}{(t+1)^2}i\Omega\,.
\end{align}
Let $\sigma_{\min}(A)$ be the smallest singular vaule of $A$. Since \[\|K_1-K_2 \|_{\infty}< e^{-{2}d_{\max}} \|S\|_{\infty}^{-2}\leq \sigma_{\min}\left(\frac{t}{t+1}\frac{i\Omega V+I}{i\Omega V-I}+\frac{1}{t+1}\right)\] from \eqref{eq:minsing2}, then
\begin{align*}
H_2&=\frac{1}{2}{K_2}\int_{0}^{\infty}\frac{1}{I+\frac{t}{t+1}{K_2}}\frac{dt}{(t+1)^2}i\Omega\\
&=\frac{1}{2}{K_2}\int_{0}^{\infty}\frac{1}{\frac{t}{t+1}\frac{i\Omega V_1+I}{i\Omega V_1-I}+\frac{1}{t+1}+\frac{t}{t+1}(K_2 -K_1)}\frac{dt}{(t+1)^2}i\Omega
\end{align*}
is well defined and, by another use of \eqref{inversebound},
\begin{align*}
&\left\|\int_{0}^{\infty}\frac{1}{I+\frac{t}{t+1}K_1}\frac{dt}{(t+1)^2}-\int_{0}^{\infty}\frac{1}{I+\frac{t}{t+1}{K_2}}\frac{dt}{(t+1)^2}\right\|_\infty\\
&\leq  \|K_1-K_2 \|_{\infty} \max_{t\ge 0}\left(\sigma_{\min}\left(\frac{t}{t+1}\frac{i\Omega V_1+I}{i\Omega V_1-I}+\frac{1}{t+1}\right)-\|K_1-K_2 \|_{\infty}\right)^{-1}\!\!\!\!\!\!\!\sigma^{-1}_{\min}\left(\frac{t}{t+1}\frac{i\Omega V_1+I}{i\Omega V_1-I}+\frac{1}{t+1}\right)\\
& \leq \|K_1-K_2 \|_{\infty} \left(e^{-2d_{\max}} \|S\|_{\infty}^{-2}-\|K_1-K_2 \|_{\infty}\right)^{-1}\left(e^{2d_{\max}} \|S\|^2_{\infty}\right),
\end{align*}
where the last bound follows from \eqref{eq:minsing2}. Finally, by a triangle inequality,
\begin{align*}
&\|H_1-{H}_2\|_{\infty}=\frac{1}{2}\left\|K_1\int_{0}^{\infty}\frac{1}{I+\frac{t}{t+1}K_1}\frac{dt}{(t+1)^2}-K_2 \int_{0}^{\infty}\frac{1}{I+\frac{t}{t+1}{K_2}}\frac{dt}{(t+1)^2}\right\|_\infty\\
&\le \frac{1}{2}\|K_1-K_2\|_\infty e^{2d_{\max}}\|S\|_\infty^2\, \left(1+\|K_2\|_\infty \left(e^{-2d_{\max}} \|S\|_{\infty}^{-2}-\|K_1-K_2 \|_{\infty}\right)^{-1}\right) \\
&\le \frac{1}{2}\|K_1-K_2\|_\infty e^{2d_{\max}}\|S\|_\infty^2\, \left(1+\big(\|K_2-K_1\|_\infty +\|K_1\|_\infty\big)\left(e^{-2d_{\max}} \|S\|_{\infty}^{-2}-\|K_1-K_2 \|_{\infty}\right)^{-1}\right) \\
 &\leq \frac{1}{2}e^{2d_{\max}}\|S\|^2_{\infty}\|K_1-K_2 \|_{\infty}\left(1+\frac{e^{2d_{\max}}\|S\|_{\infty}^2(1-e^{-2d_{\max}})+\|K_1-K_2 \|_{\infty}}{e^{-2d_{\max}} \|S\|_{\infty}^{-2}-\|K_1-K_2 \|_{\infty}}\right).
\end{align*}
where we used Equation~\eqref{eq:minsing1} in the last line in order to bound $\|K_1\|_\infty$.
\end{proof}
We now prove a continuity bound for covariance matrices in terms of their associated Gaussian Hamiltonians.
\begin{lem}\label{th:boundfromhtov}
Let $H_1$ and $H_2$ be two Hamiltonian matrices, $H_1$ with normal form $H_1=S^{-\intercal}D{S^{-1}}$, and $\|H_1-H_2\|_{\infty}\leq \frac{1}{{4}\|S\|^{6}_{\infty}}\big(e^{-2d_{\max}-1}(1-e^{-2d_{\min}})\big)$. Then
\begin{align}
\|V_1-V_2\|_{\infty}&\leq 4\,\frac{\|S\|_{\infty}^8 e^{2d_{\max}+1}\|H_1-H_2\|_{\infty}}{(1-e^{-2d_{\min}})^2} \,.
\end{align}
\end{lem}
\begin{proof}
First of all, notice that $\|(2i\Omega V_1-I)^{-1}-(2i\Omega V_2-I)^{-1}\|_{\infty}=\frac{1}{2}\|e^{2H_1i\Omega}-e^{2H_2 i\Omega}\|_{\infty}$ by \eqref{eq:HtoV}. Moreover, for $s,t\in[0,1]$, 
\begin{align}
\left\|e^{2s(t H_1+(1-t)H_2)i\Omega}\right\|_{\infty}&=\left\|e^{2s(H_1 i\Omega+(1-t)(H_2-H_1)i\Omega)}\right\|_{\infty}\nonumber\\
&=\left\|S^{-\intercal} e^{2s(D_1 i\Omega+(1-t)S^{\intercal}(H_2-H_1)Si\Omega)} S^{\intercal}\right\|_{\infty}\nonumber\\
&\leq \|S\|_{\infty}^2e^{2s\|D_1 i\Omega+(1-t)S^{\intercal}(H_2-H_1)Si\Omega\|_{\infty}}\nonumber\\
&\leq \|S\|_{\infty}^2 e^{2sd_{\max}+2s\|S\|^2_{\infty}\|H_2-H_1\|_{\infty}}\,.
\end{align}
Hence, by Duhamel formula,
\begin{align}
e^{2H_1i\Omega}-e^{2H_2 i\Omega}=\int_{0}^1 \int_{0}^1  e^{{2}s(t H_1+(1-t)H_2)i\Omega}2(H_1-H_2)i\Omega e^{{2}(1-s)(t H_1+(1-t)H_2)i\Omega}dt \,ds\,,
\end{align}
we can estimate
\begin{align}
\left\|e^{2H_1i\Omega}-e^{2H_2 i\Omega}\right\|_{\infty}&\leq 2\int_{0}^1 \int_{0}^1  \left\|e^{2s(t H_1+(1-t)H_2)i\Omega}\right\|_{\infty}\|H_1-H_2\|_{\infty} \left\|e^{2(1-s)(t H_1+(1-t)H_2)i\Omega}\right\|_{\infty}dt\,ds\nonumber\\
& \leq 2\|S\|_{\infty}^4 e^{2d_{\max}+2\|S\|^2_{\infty}\|H_2-H_1\|_{\infty}}\|H_1-H_2\|_{\infty}\,.
\end{align}
We then have, by perturbation bounds on the inverse \eqref{inversebound} and the condition in the theorem statement,
\begin{align*}
\|V_1-V_2\|_{\infty}&=\frac{1}{2}\|2i\Omega V_1-I-(2i\Omega V_2-I)\|_{\infty}\\
&\leq \frac{1}{2}\|2i\Omega V_1-I\|_{\infty} \|2i\Omega V_2-I\|_{\infty}\|(2i\Omega V_1-I)^{-1}-(2i\Omega V_2-I)^{-1}\|_{\infty}\nonumber\\
&\leq \frac{1}{4}\frac{2\|S\|_{\infty}^2}{1-e^{-2d_{\min}}} \!\left[\left(\|S\|_{\infty}^2\frac{2}{1-e^{-2d_{\min}}}\right)^{-1}\!\!\!\!\!\!-{\frac{1}{2}}\|e^{2H_1i\Omega}-e^{2H_2 i\Omega}\|_{\infty} \right]^{-1}\!\!\!\!\!\!\left\|e^{2H_1i\Omega}-e^{2H_2 i\Omega}\right\|_{\infty}\nonumber\\
&\le  \frac{\|S\|_{\infty}^4}{(1-e^{-2d_{\min}})^2} \!\left[1-\frac{\|S\|_{\infty}^2}{1-e^{-2d_{\min}}}\|e^{2H_1i\Omega}-e^{2H_2 i\Omega}\|_{\infty} \right]^{-1}\left\|e^{2H_1i\Omega}-e^{2H_2 i\Omega}\right\|_{\infty}\nonumber  \\
&\leq \frac{2\|S\|_{\infty}^8 e^{2d_{\max}+{2}\|S\|^2_{\infty}\|H_2-H_1\|_{\infty}}\|H_1-H_2\|_{\infty}}{(1-e^{-2d_{\min}})^2\left[1-\frac{{2}\|S\|_{\infty}^6 e^{2d_{\max}+{2}\|S\|^2_{\infty}\|H_2-H_1\|_{\infty}}}{1-e^{-2d_{\min}}}\|H_1-H_2\|_{\infty}  \right]}  \nonumber\\
&\le \frac{2\|S\|_{\infty}^8 e^{2d_{\max}+{2}\|S\|^2_{\infty}\|H_2-H_1\|_{\infty}}\|H_1-H_2\|_{\infty}}{(1-e^{-2d_{\min}})^2\left[1-\frac{{2}\|S\|_{\infty}^6 e^{2d_{\max}+1}}{1-e^{-2d_{\min}}}\|H_1-H_2\|_{\infty}  \right]}  \nonumber\\
&\leq 4\frac{\|S\|_{\infty}^8 e^{2d_{\max}+1}\|H_1-H_2\|_{\infty}}{(1-e^{-2d_{\min}})^2} \,,
\end{align*}
where we also used \eqref{eq:normbound1} in the second inequality above, and where in the last two inequalities we also used that $\|H_1-H_2\|_\infty\le (2\|S\|^2_\infty)^{-1}$.
\end{proof}

\subsection{Sparse approximations of $(2V-i\Omega)^{-1}$}\label{sec:sparse_approx}

The main goal of this section is to derive a sparse approximation of $(2V-i\Omega)^{-1}$ from an estimated covariance matrix $\hat{V}$. We start by stating a simple Taylor expansion bound for the matrix $(2V-i\Omega)^{-1}$ which will turn out to be crucial for our later local perturbation analysis. Indeed, one of the main ideas of our work is to invert the matrix $2V-i\Omega$ only on submatrices and show that this still yields good estimates of the Hamiltonian if we follow this procedure on appropriate submatrices. This way we can make sure that the precision we need to estimate the covariance matrix to estimate the terms of the Hamiltonian scales logarithmically with the number of modes.

\begin{lem}\label{th:seriesapprox}
Let $V$ be a covariance matrix of a Gaussian state with Hamiltonian $H=S^{-\intercal}D{S^{-1}}$. We have
\begin{equation}\label{analyticform}
(2 V-i\Omega )^{-1}=-\frac{I-e^{+ 2Hi\Omega}}{2}(i\Omega)=\frac{1}{2}\sum_{n=1}^{\infty}\frac{\left(2Hi\Omega\right)^{n}}{n!}(i\Omega)=\frac{1}{2}\sum_{n=1}^{l}\frac{\left( 2Hi\Omega\right)^{n}}{n!}(i\Omega)+E\,,
\end{equation}
with $\|E\|_{\infty}\leq \frac{\|S\|_{\infty}^2}{{2}}\left(\frac{(2d_{\max})^{l+1}e^{2d_{\max}}}{(l+1)!}\right)$.

\end{lem}

\begin{proof}

We recall that $V=\frac{1}{2}i\Omega \coth(Hi\Omega)$ by \eqref{eq:HtoV}, so $2i\Omega V- I=\coth(Hi\Omega)- I$. Then,
\eqref{analyticform} comes from writing $\coth(x)- 1=\frac{2 e^{- x}}{e^{x}-e^{-x}}=- \frac{2}{1-e^{ 2x}}$ and using the series expansion of the exponential function. The bound on the remainder term $E$ in~\eqref{analyticform} arises from the identity
\begin{align*}
E(i\Omega)= S^{{-\intercal}}\left(\frac{1}{2}\sum_{n=l+1}^{\infty}\frac{\left( 2Di\Omega\right)^{n}}{n!}\right)S^{{\intercal}},
\end{align*}
and thus
\begin{align*}
EE^{\dagger}&=S^{{-\intercal}}\left(\frac{1}{2}\sum_{n=l+1}^{\infty}\frac{\left( 2Di\Omega\right)^{n}}{n!}\right)S^{{\intercal}}S\left(\frac{1}{2}\sum_{n=l+1}^{\infty}\frac{\left( 2Di\Omega\right)^{n}}{n!}\right)S^{{-1}}
\\&\leq \sigma_{\max}(S^{\intercal}S) S^{{-\intercal}}\left(\frac{1}{2}\sum_{n=l+1}^{\infty}\frac{\left( 2Di\Omega\right)^{n}}{n!}\right)^{2}S^{{-1}}\\
&\leq \frac{\sigma_{\max}(S^{\intercal}S)}{4}\left(\frac{(2d_{\max})^{l+1}e^{2d_{\max}}}{(l+1)!}\right)^2 S^{-\intercal}S^{-1}\\
&\leq \frac{\sigma^2_{\max}(S^{\intercal}S)}{4}\left(\frac{(2d_{\max})^{l+1}e^{2d_{\max}}}{(l+1)!}\right)^2 
\end{align*}
where in the second inequality we estimated the Lagrange remainder. Thus we obtain $\|E\|_{\infty}\leq \frac{\sigma_{\max}(SS^{\intercal})}{{2}}\left(\frac{(2d_{\max})^{l+1}e^{2d_{\max}}}{(l+1)!}\right)$.
\end{proof}

\begin{rem}
Assuming that $H$ has an underlying graph structure, note that $((2V-i\Omega)^{-1}-E)_{2i-\delta_{1},2j-\delta_2}$, with $\delta_1,\delta_2\in{0,1}$, is zero if $j$ is at distance more than $l$ from i. This property will be crucial in the following.
\end{rem}

\begin{lem}\label{theoapprinvsparse}
Let $M=\left(\begin{smallmatrix}
    A &B\\ C&D
\end{smallmatrix}\right)\in\mathbb{C}^{m_1+m_2}\times \mathbb{C}^{m_1+m_2}$ be a positive definite matrix with blocks $A\in \mathbb{C}^{m_1}\times\mathbb{C}^{m_1}$, $B\in \mathbb{C}^{m_1}\times\mathbb{C}^{m_2}$, $C\in \mathbb{C}^{m_2}\times\mathbb{C}^{m_1}$, $D\in \mathbb{C}^{m_2}\times\mathbb{C}^{m_2}$, and such that $D$ and $A-BD^{-1}C$ are invertible. Let $N\in \mathbb{C}^{m_1}\times\mathbb{C}^{m_1}$ be the upper-left block of $M^{-1}$. Denote by $\widetilde B$, resp. $\widetilde C$, the matrix obtained from $B$ by substituting the first row with zeros, resp. from $C$ by substituting the first column with zeros. Clearly $B=C^\dagger$, $\widetilde{B}=\widetilde{C}^\dagger$ and $\|B-\widetilde{B}\|_\infty=\|C-\widetilde{C}\|_\infty$. Moreover, 
\begin{align*}
|(N^{-1})_{1,j}-A_{1,j}|\leq \|B-\widetilde{B}\|_\infty \|M^{-1}\|_{\infty}\|M\|_{\infty}\,\,\forall j\in{m_1},\\ 
|(N^{-1})_{j,1}-A_{j,1}|\leq \|B-\widetilde{B}\|_\infty \|M^{-1}\|_{\infty}\|M\|_{\infty}\,\,\forall j\in{m_1}. 
\end{align*}

\end{lem}

\begin{proof}
We have by block-inversion that
\begin{equation*}
M^{-1}=\begin{pmatrix}
    N\equiv(A-BD^{-1}C)^{-1} & -(A-BD^{-1}C)^{-1}BD^{-1} \\ -D^{-1}C(A-BD^{-1}C)^{-1}& D^{-1}+ D^{-1}C(A-BD^{-1}C)^{-1}BD^{-1}
\end{pmatrix}\,.
\end{equation*}
By definition, $A-BD^{-1}C=N^{-1}$. Since by construction the first row of $\widetilde B$ is identically zero, $\widetilde B_{1,j}=0\,\,\forall j\in[m_2]$. Similarly, $\widetilde C_{j,1}=0\,\,\forall j\in[m_2]$. Then $(\widetilde {B}D^{-1}C)_{1,j}=0 \,\,\forall j\in[m_1]$ and $(BD^{-1}\widetilde{C})_{j,1}=0 \,\,\forall j\in[m_1]$. This means that $(A-\widetilde{B}D^{-1}C)_{1,j}=A_{1,j} \,\,\forall j\in[m_1]$ and $(A-BD^{-1}\widetilde{C})_{j,1}=A_{j,1} \,\,\forall j\in[m_1]$. Furthermore, $\|A-BD^{-1}C-(A-\widetilde{B}D^{-1}C)\|_{\infty}\leq \|B-\widetilde{B}\|_{\infty}\|D^{-1}\|_{\infty}\|C\|_{\infty}$. Moreover, $\|C\|_{\infty}\leq \|M\|_{\infty}$ and $\sigma_{\min}(D)\geq \sigma_{\min}(M)$ as the matrix $M$ is positive definite, and hence $\|D^{-1}\|_{\infty}\leq \|M^{-1}\|_{\infty}$. Then $\|A-BD^{-1}C-(A-\widetilde{B}D^{-1}C)\|_{\infty}\leq  \|B-\widetilde{B}\|_\infty\|M^{-1}\|_{\infty}\|M\|_{\infty}$. In the same way,  $\|A-BD^{-1}C-(A-BD^{-1}\widetilde{C})\|_{\infty}\leq \|C-\widetilde{C}\|_\infty\|M^{-1}\|_{\infty}\|M\|_{\infty}$. The bounds in the statement follow.
\end{proof}

\begin{rem}\label{remarkcondition}{}$\|M\|_{\infty}\|M^{-1}\|_{\infty}$ is the condition number of $M$. Looking at the proof, it is clear that instead of $\|M\|_{\infty}\|M^{-1}\|_{\infty}$ we can leave $\|D^{-1}C\|_{\infty}$ in the bound statement. When $C$ contains very few columns with respect to the size of $D$, and $D$ has a large condition number, it would be natural to guess that $\|D^{-1}C\|_{\infty}$ does not scale as  $\|D^{-1}\|_{\infty}$, as the subspace spanned by the column space of $C$ may be not aligned with the eigenspaces with largest eigenvalues of $\|D^{-1}\|_{\infty}$. All the following bounds can be in principle improved using this fact. We prefer not to do it explicitly since a bound on $\|D^{-1}C\|_{\infty}$ is not a natural assumption in our context.
\end{rem}

We will now introduce some definitions and notations necessary to explain our results on the local inversion technique.
\begin{defi}[Neighborhood structure] \label{neighborhoodstruct} A collection $\mathcal{N}=\{\mathcal{N}_i\}_{i\in[m]}$ of subsets of $[m]$ is called a neighborhood structure if for all $i\in[m]$, $i\in\mathcal{N}_i$; (b) for all $i,j\in [m]$, $i\in\mathcal{N}_j\Rightarrow j\in\mathcal{N}_i$.
\end{defi}
\begin{rem}
We use the term neighborhood to follow the convention used in classical papers on learning graphical models. But we emphasize that the $\mathcal{N}_i$ will \emph{not} be the neighborhood of a vertex $i$ of the interaction graph. Rather, the $\mathcal{N}_i$ will often correspond to candidates for the set of all vertices at most a certain distance away from $i$.
\end{rem}

\begin{defi}[Localization] 
Given a $2m\times 2m$ self-adjoint matrix $M$ and a neighborhood structure $\mathcal{N}=\{\mathcal{N}_i\}_{i\in [m]}$, let $\operatorname{LOC}_{\mathcal{N}}(M)$ be the matrix obtained from $M$ by keeping the matrix elements given by the neighborhood structure $\mathcal{N}$, that is, we define $\operatorname{LOC}_{\mathcal{N}}(M)_{2i-\delta_1,2j-\delta_2}=M_{2i-\delta_1,2j-\delta_2}$ whenever $j\in \mathcal{N}_i$ and $\delta_1,\delta_2\in\{0,1\}$, and set all other entries to $0$. Note that, by construction, $\operatorname{LOC}_{\mathcal{N}}(M)$ is self-adjoint and has at most $2\max_i|\mathcal{N}_i|$ non-zero element per row.
\end{defi}
Given a neighborhood structure, we will construct approximate inverses of matrices by "sticthing together" the inverses of various submatrices with only entries in the neighborhood structure. More formally, we have:
\begin{defi}[Local inversion]\label{constrappr1}
Let $N$ be a positive definite $2m\times 2m$ matrix and let $\mathcal{N}=\{\mathcal{N}_i\}_{i\in[m]}$ a neighborhood structure of $[m]$. We denote by $\operatorname{LI}_{\mathcal{N}}(N)$ the matrix constructed in the following way: for each $i\in[m]$, let $N_i$ be the positive definite, $2|\mathcal{N}_i|\times 2|\mathcal{N}_i|$ matrix with entries $N_{2j_1-\delta_1,2j_2-\delta_2}$ corresponding to indices $j_1,j_2\in \mathcal{N}_{i}$, $\delta,\delta'\in \{0,1\}$ (i.e. we select modes in $\mathcal{N}_i$). Then, for each $i\in[m]$, $j\in \mathcal{N}_{i}$ and $\delta,\delta'\in \{0,1\}$, we define $\left(\operatorname{LI}_{\mathcal{N}_i}(N)\right)_{2i-\delta,2j-\delta'}$ to be $\left({N_{i}^{-1}}\right)_{2i-\delta,2j-\delta'}$ and $\left(\operatorname{LI}_{\mathcal{N}}(N)\right)_{2i-\delta,2j-\delta'}$ to be equal to $\tfrac{1}{2}\left[\left({N_{i}^{-1}}\right)_{2i-\delta,2j-\delta'}+\left({N_{j}^{-1}}\right)^*_{2j-\delta',2i-\delta}\right]$, and set the other entries of $\operatorname{LI}_{\mathcal{N}}(N)$ to $0$. Note that, by construction, $\operatorname{LI}_{\mathcal{N}}(N)$ is self-adjoint, as both $\left({N_{i}^{-1}}\right)$ and $\left({N_{j}^{-1}}\right)$ are.
\end{defi}
Note that all inverses exist, as they are inverses of submatrices of a positive matrix.
We then have the following first lemma bounding the difference between the local inverse and the true inverse:

\begin{lem}\label{th:Hsparsebound}
Let $M$ be a positive definite matrix and let $\mathcal{N}=\{\mathcal{N}_i\}_{i\in[m]}$ be a neighborhood structure of $[m]$. For each $i\in[m]$ and $\delta\in\{0,1\}$, let $\left(\begin{smallmatrix}
    A_{i,\delta} &B_{i,\delta}\\ C_{i,\delta}&D_{i,\delta}
\end{smallmatrix}\right)\in\mathbb{C}^{2m}\times \mathbb{C}^{2m}$, with blocks $A_{i,\delta}\in \mathbb{C}^{2|\mathcal{N}_i|}\times\mathbb{C}^{2|\mathcal{N}_i|} ,B_{i,\delta}\in \mathbb{C}^{2|\mathcal{N}_i|}\times\mathbb{C}^{2m-2|\mathcal{N}_i|} ,C_{i,\delta}\in \mathbb{C}^{2m-2|\mathcal{N}_i|}\times\mathbb{C}^{2|\mathcal{N}_i|} ,D_{i,\delta}\in \mathbb{C}^{2m-2|\mathcal{N}_i|}\times\mathbb{C}^{2m-2|\mathcal{N}_i|} $ be the matrix $\pi_{i,\delta}M \pi_{i,\delta}^{-1}$ where $\pi_{i,\delta}$ is a permutation that brings the coefficients $2j-\delta'$ with $\delta'\in\{0,1\}$ and $j\in\mathcal{N}_i$ to the upper-left block $A_{i,\delta}$, in such a way that index $2i-\delta$ is mapped to $2-\delta$. Let $\widetilde{B}_{i,\delta}$ the matrix constructed from $B_{i,\delta}$ by putting the first row to zero. Let $E$ be a self-adjoint matrix such that $\operatorname{LOC}_{\mathcal{N}}(M-E)=M-E$. Then,
\begin{equation}\label{eq:opnormbound}
\|M-\operatorname{LI}_{\mathcal{N}}(M^{-1})\|_{\infty}\leq 2\max_{i\in [m]}\max_{\delta\in\{0,1\}}{|\mathcal{N}_i|}\|B_{i,\delta}-\widetilde{B}_{i,\delta}\|_{\infty} \|M^{-1}\|_{\infty}\|M\|_{\infty}+(2\max_{i\in [m]}{|\mathcal{N}_i|}+1)\|E\|_{\infty}
\end{equation}
\end{lem}
\begin{proof}

By a triangle inequality, we have 
 \begin{align*}
 \|M-\operatorname{LI}_{\mathcal{N}}(M^{-1})\|_{\infty}&\le\|E\|_\infty+\|M-E-\operatorname{LI}_{\mathcal{N}}(M^{-1})\|_{\infty}\,.
 \end{align*}
 It remains to bound the last operator norm above. Since $\operatorname{LOC}_{\mathcal{N}}(M-E)-\operatorname{LI}_{\mathcal{N}}(M^{-1})=M-E-\operatorname{LI}_{\mathcal{N}}(M^{-1})$ is self-adjoint, it suffices to find a uniform bound on its eigenvalues. By the Gershgorin circle theorem, denoting the matrix $L:=\operatorname{LOC}_{\mathcal{N}}(M-E)-\operatorname{LI}_{\mathcal{N}}(M^{-1})$, these eigenvalues are situated in the $2m$ intervals
 \begin{align*}
     \left[L_{k,k}-\sum_{k'\ne k}\big|L_{k,k'}\big|,\,L_{k,k}+\sum_{k'\ne k}\big|L_{k,k'}\big|\right],\, k\in[2m]\,.
 \end{align*}
 
 Now for any $i\in[m]$, $j\in\mathcal{N}_i$ and $\delta'\in\{0,1\}$, we have by the construction of the local inverse matrix, denoting $N=M^{-1}$:
 \begin{align}
 \Big|M_{2i-\delta,2j-\delta'}-\big(\operatorname{LI}_{\mathcal{N}}(M^{-1})\big)_{2i-\delta,2j-\delta'}\Big|&\leq\frac{1}{2}\Big|M_{2i-\delta,2j-\delta'}-\big(N_{i}^{-1}\big)_{2i-\delta,2j-\delta'}\Big|\nonumber\\
 &+\frac{1}{2}\Big|M_{2i-\delta,2j-\delta'}-\big(N_{j}^{-1}\big)^*_{2j-\delta',2i-\delta}\Big|\nonumber\,.
 \end{align}
  Via Lemma~\ref{theoapprinvsparse} applied to the matrices  $\pi_{i,\delta}M\pi_{i,\delta}^{-1}$ and $\pi_{j,\delta}M\pi_{j,\delta}^{-1}$,
\begin{align*}
\Big|M_{2i-\delta,2j-\delta'}-\big(N_{i}^{-1}\big)_{2i-\delta,2j-\delta'}\Big|&\leq  \|B_{i,\delta}-\widetilde{B}_{i,\delta}\|_\infty \|M^{-1}\|_{\infty}\|M\|_{\infty}\, ,\\
\Big|M_{2i-\delta,2j-\delta'}-\big(N_{j}^{-1}\big)^*_{2j-\delta',2i-\delta}\Big|&=\Big|M^*_{2j-\delta',2i-\delta}-\big(N_{j}^{-1}\big)^*_{2j-\delta',2i-\delta}\Big|
&\!\!\le  \!\|B_{j,\delta}-\widetilde{B}_{j,\delta}\|_\infty \|M^{-1}\|_{\infty}\|M\|_{\infty}
\end{align*}
and hence $|L_{2i-\delta,2j-\delta'}|\leq |(M-\big(\operatorname{LI}_{\mathcal{N}}(M^{-1})\big)_{2i-\delta,2j-\delta'}|+|E_{2i-\delta,2j-\delta'}|\leq \max_{i\in[m]}\max_{\delta\in\{0,1\}}\|B_{i,\delta}-\widetilde{B}_{i,\delta}\|_\infty \|M^{-1}\|_{\infty}\|M\|_{\infty} +\|E\|_{\infty}$. Using Gershgorin theorem together with the fact that the matrix $L$ has at most $2\max_i|\mathcal{N}_i|$ non-zero coefficients per row, we get
\begin{align*}
 \|L\|_\infty\le 2\max_{i\in[m]}\max_{\delta\in\{0,1\}}|\mathcal{N}_i|\left(\|B_{i,\delta}-\widetilde{B}_{i,\delta}\|_\infty \|M^{-1}\|_{\infty}\|M\|_{\infty} +\|E\|_{\infty}\right)\,.
\end{align*}
The result follows.
 
\end{proof}

However, in our setting we will only have access to approximations of the true covariance matrix when performing the local inversions. In the next lemma we bound the effect caused by this, i.e. when we only have access to an entrywise approximation of the inverse of a matrix:

\begin{lem}\label{th:approxlocinvbound}
In the notations of Lemma \ref{th:Hsparsebound}, given a $2m\times 2m$ positive definite matrix $M$, $N=M^{-1}$ and $\hat{N}$ an entrywise approximation of $N$, where $|\hat{N}_{k,k'}-N_{k,k'}|\le \zeta$ for all $k,k'\in[2m]$ with $\zeta\le \frac{1}{4\max_{i\in[m]}|\mathcal{N}_i|}\|M\|_\infty^{-1}$,

\begin{align}
\|M-\operatorname{LI}_{\mathcal{N}}(\hat{N})\|_{\infty}&\leq 2\max_{i\in [m]}\max_{\delta\in\{0,1\}}{|\mathcal{N}_i|}\|B_{i,\delta}-\widetilde{B}_{i,\delta}\|_{\infty} \|M^{-1}\|_{\infty}\|M\|_{\infty}+(2\max_{i\in [m]}{|\mathcal{N}_i|}+1)\|E\|_{\infty}\nonumber\\
&\,+8\|M\|_{\infty}^{2}\max_{i}|\mathcal{N}_i|^2\zeta\,\nonumber.
\end{align}

\end{lem}

\begin{proof}
By the triangle inequality,
\begin{align*}
\|M-\operatorname{LI}_{\mathcal{N}}(\hat{N})\|_{\infty}\leq \|M-\operatorname{LI}_{\mathcal{N}}(N)\|_{\infty}+\|\operatorname{LI}_{\mathcal{N}}(N)-\operatorname{LI}_{\mathcal{N}}(\hat{N})\|_{\infty}.
\end{align*}
The first term can be bounded as in Lemma~\ref{th:Hsparsebound}. The second term can be bounded using the fact that both $\operatorname{LI}_{\mathcal{N}}(N)$ and $\operatorname{LI}_{\mathcal{N}}(\hat{N})$ have at most $2{\max_{i\in[m]}|\mathcal{N}_i|}$ non-zero elements on each column, and each non-zero entry can be bounded using the perturbation bound for the inverse (cf. Equation~\eqref{inversebound}): for each $i\in[m]$ $j\in \mathcal{N}_i$, $\delta_1,\delta_2\in \{0,1\}$, we have by the equivalence between the max norm and operator norm of matrices:
\begin{align}
\Big|\big(\operatorname{LI}_{\mathcal{N}}(\hat{N})\big)_{2i-\delta_1,2j-\delta_2}-\big(\operatorname{LI}_{\mathcal{N}}({N})\big)_{2i-\delta_1,2j-\delta_2}\Big|
&\leq \frac{1}{2}\left(\|N_i^{-1}-(\hat{N}_i)^{-1}\|_{\infty}    +\|N_j^{-1}-(\hat{N}_j)^{-1}\|_{\infty} \right)\nonumber\\
&\leq \max_{i\in[m]}\|N_i^{-1}\|_{\infty}\|(\hat{N}_i)^{-1}\|_{\infty}\|N_i-\hat{N}_i\|_{\infty}\nonumber\\&\leq \max_{i\in[m]}\|N_i^{-1}\|_{\infty}\|(\hat{N}_i)^{-1}\|_{\infty}\|N_i-\hat{N}_i\|_{\infty}\nonumber\\
&\leq \max_{i\in[m]}\sigma_{\min}(N_i)^{-1}\sigma_{\min}(\hat{N}_i)^{-1}2|\mathcal{N}_i|\zeta\nonumber\\
&\leq \max_{i\in[m]}\sigma_{\min}(M^{-1})^{-1} (\sigma_{\min}({N}_i)-\|N_i-\hat{N}_i\|_{\infty})^{-1}2|\mathcal{N}_i|\zeta\nonumber\\
&\leq\max_{i\in[m]}\|M\|_{\infty}(\|M\|^{-1}_{\infty}-2|\mathcal{N}_i|\zeta)^{-1}2|\mathcal{N}_i|\zeta\nonumber\\
&\le \max_{i\in[m]}4\|M\|_{\infty}^2|\mathcal{N}_i|\zeta \,,\label{eq:inversentry}
\end{align}
where we used that $\|N_i-\hat{N}_{i}\|_{\infty}\leq \|N_{i}-\hat{N}_{i}\|_{2}\leq2|\mathcal{N}_i|\zeta$, $\sigma_{\min}(N_{i})\geq \sigma_{\min}(M^{-1})$, and the condition on $\zeta$ in the theorem statement. The bound on $\|\operatorname{LI}_{\mathcal{N}}(N)-\operatorname{LI}_{\mathcal{N}}(\hat{N})\|_{\infty}$ follows again from the use of Gershgorin's theorem.
\end{proof}
We now apply these perturbation bounds to the estimation of $2i\Omega V-I=i\Omega(2V-i\Omega)$. First we recall the following definition

\begin{defi}[Interaction graph]
 Given a Hamiltonian matrix $H$, its interaction graph is a graph $G=([m],\mathsf{E})$ made of $m$ vertices, where two vertices $i,j\in [m]$ are connected by an edge whenever the $2\times 2$ sub-matrix
$\left\{H_{2i-\delta,2j-\delta'}\right\}_{\delta,\delta'\in\{0,1\}}$ is non-zero.
\end{defi}
We also define a neighborhood structure that will play an important role in our proofs.

\begin{defi}[$l$-neighborhoods]\label{nbhdl}
Given $l\in\mathbb{N}$ and a Hamiltonian matrix $H$, let $G$ be interaction graph of $H$. 
The $l$-neighborhood of a site $i\in [m]$ is the set of indices $j$ in $[m]$ that are at distance at most $l$ from $i$, according to the graph $G$. We denote the set of $l$-neighborhoods as $\mathcal{N}(l)=\{\mathcal{N}(l)_i\}_{i\in [m]}$ and $\xi(l):=\max_{i\in[m]}|\mathcal{N}(l)_i|$. Clearly, any set of $l$-neighborhood satisfies the definition \ref{neighborhoodstruct} of a neighborhood structure. 
\end{defi}

\begin{lem}\label{boundkernellocal}

Let $V$ be a bona-fide $m$-mode covariance matrix with associated symplectic matrix $S$ given by Equation \eqref{eq:SdecV}. For $l\geq1$, let {$\zeta<\frac{1}{2(\xi(l))}\frac{e^{-2d_{\max}}}{\|S\|_\infty^2(1-e^{-2d_{\max}})}$} and a matrix $\hat{V}$ such that $|\hat{V}_{k,k'}-V_{k,k'}|\leq{\zeta}$ for and any $k,k'\in[2m]$. 
Then, $\operatorname{LI}_{\mathcal{N}(l)}(2\hat{V}-i\Omega)$ is well defined, and

\begin{align}\label{equ:error_localinverse}
&\|(2V-i\Omega)^{-1}-\operatorname{LI}_{\mathcal{N}(l)}(2\hat{V}-i\Omega)\|_{\infty} \nonumber\\
&\quad \le \xi(l)\|S\|_{\infty}^6\frac{(2d_{\max})^{l+1}e^{4d_{\max}}}{(l+1)!}{\frac{1-e^{-2d_{\max}}}{1-e^{-2d_{\min}}}}+(2\xi(l)+1)\frac{\|S\|_{\infty}^2}{2}\frac{(2d_{\max})^{l+1}e^{2d_{\max}}}{(l+1)!}\\&+
{2\|S\|^4_\infty e^{4 d_{\max}}(1-e^{-2d_{\max}})^2}\xi(l)^2\zeta \nonumber\,.
\end{align}
\end{lem}

\begin{proof}
First we note that by the condition on $\zeta$ we have that the restriction to each neighborhood $\mathcal{N}(l)_i$ of $2\hat{V}-i\Omega$ is strictly positive and, as such, all the local inverses exist.  
More specifically, by Lemma~\ref{th:approxlocinvbound} applied to the matrices $M=(2V-i\Omega)^{-1}$, $\hat{N}=2\hat{V}-i\Omega$, the neighborhood $\mathcal{N}(l)$ of Definition \ref{nbhdl} and $E$ as defined in 
Lemma \ref{th:seriesapprox}, we have that 
\begin{align*}
&\|(2V-i\Omega)^{-1}-\operatorname{LI}_{\mathcal{N}(l)}(2\hat{V}-i\Omega)\|_{\infty}\\
& \qquad \leq 2\max_{\delta\in\{0,1\}}{\xi(l)}\|B_{i,\delta}-\widetilde{B}_{i,\delta}\|_{\infty} \|M^{-1}\|_{\infty}\|M\|_{\infty} +(2{\xi(l)}+1)\|E\|_{\infty}+8\|M\|_{\infty}^{2}\xi(l)^2\zeta 
\end{align*}
as long as $\zeta\le \frac{1}{4\xi(l)}\|M\|_\infty^{-1}$ and $M-E=\operatorname{LOC}_{\mathcal{N}(l)}(M-E)$. The first condition is satisfied by Equation {\eqref{eq:minsing1} in Lemma~\ref{lem:bound}, which states that $\|M\|_\infty^{-1}\ge \frac{2e^{-2d_{\max}}}{\|S\|_\infty^2(1-e^{-2d_{\max}})}$}. 
The second condition follows from the series approximation of Lemma \ref{th:seriesapprox}, since 
\begin{align*}
M-E=(2V-i\Omega)^{-1}-E=\frac{1}{2}\sum_{n=1}^{l}\frac{\left( 2Hi\Omega\right)^{n}}{n!}(i\Omega)\,,
\end{align*}
so that, for any two vertices $i,j$ of the graph $G$, $(M-E)_{2i-\delta,2j-\delta'}$ is non-zero only if $j$ is at distance at most $l$ from $i$. Using once again Equation \eqref{eq:normbound1} together with \eqref{eq:minsing1} of Lemma~\ref{lem:bound}, as well as by noticing that, for any $i\in[m]$ and $\delta\in\{0,1\}$, $B_{i,\delta}-\widetilde{B}_{i,\delta}$ is a sub-row of $E$ so that $\|B_{i,\delta}-\widetilde{B}_{i,\delta}\|_\infty\le \|E\|_\infty$, we find the following bound:  
\begin{align*}
\|M-\operatorname{LI}_{\mathcal{N}(l)}(\hat{N})\|_\infty&\le 2\xi(l)\|E\|_\infty\|S\|_\infty^4\frac{(1-e^{-2d_{\max}})e^{2d_{\max}}}{1-e^{-2d_{\min}}}+(2\xi(l)+1)\|E\|_\infty\\&+{2\|S\|^4_\infty e^{4 d_{\max}}(1-e^{-2d_{\max}})^2}\xi(l)^2\zeta\,.
\end{align*}
Finally, the result follows after using the bound $\|E\|_{\infty}\leq \frac{\|S\|_{\infty}^2}{2}\left(\frac{(2d_{\max})^{l+1}e^{2d_{\max}}}{(l+1)!}\right)$ derived in Lemma~\ref{th:seriesapprox}.

\end{proof}

We now need the following lemma.
\begin{lem}\label{lem:Lambert}
Let $C_1,C_2,\epsilon>0$ and let $W(x)$ be the principal branch of the Lambert W function defined through $e^{W(x)}W(x)=x$ for $x>0$. 
If we take $l=\lfloor{(C_2e)\exp{\left(W\left(\frac{\log\frac{C_1}{\eps}}{C_2e}\right)\right)}}\rfloor$, then $C_1\frac{{C_2}^{l+1}}{(l+1)!}\leq \epsilon$ and $\lim_{\eps\rightarrow 0}\Delta^l\eps^\gamma=0$ for every $\gamma>0$.
\end{lem}
\begin{proof}
We have 
\[C_1\frac{(C_2)^{l+1}}{(l+1)!}\leq C_1\left(\frac{C_2e}{l+1}\right)^{l+1},\]
and taking logarithms, we have 
\begin{align}
\log\frac{C_1\left(\frac{C_2e}{l+1}\right)^{l+1}}{\epsilon}&= -(l+1)\log \left(\frac{l+1}{C_2e}\right)+\log\frac{C_1}{\epsilon}\nonumber\\&\leq-(C_2e)e^{W\left(\frac{\log\frac{C_1}{\eps}}{C_2e}\right)}W\left(\frac{\log\frac{C_1}{\eps}}{C_2e}\right)+\log\frac{C_1}{\epsilon}\nonumber\\
&=-C_2e\left(\frac{\log\frac{C_1}{\eps}}{C_2e}\right)+\log\frac{C_1}{\epsilon}=0,\label{eq:WLamb}
\end{align}
so that $C_1\left(\frac{C_2e}{l+1}\right)^{l+1}\leq \epsilon$. An upper bound on the Lambert W function is $W(x)\leq \log\left(\frac{2x+1}{1+\log (x+1)}\right)$, see~\cite{hoorfar2008inequalities}. So, with our choice,  $l\leq \frac{2\log\frac{C_1}{\eps} +C_2e}{1+\log\left({\frac{\log\frac{C_1}{\eps}}{C_2e}}+1\right)}$, and $\lim_{\eps\rightarrow 0}\Delta^l\eps^\gamma=0$ for every $\gamma>0$.

\end{proof}

The following remark is the basis for our sample complexity analysis.
\begin{rem}\label{rem:hbound} In the worst case, $\xi(l)\leq \{m,\Delta^{l}\}$, where $\Delta-1$ is the degree of the graph. Hence, the bound found in Lemma \ref{boundkernellocal} simplifies to

\begin{align*}
&\|(2V-i\Omega)^{-1}-\operatorname{LI}_{\mathcal{N}(l)}(2\hat{V}-i\Omega)\|_{\infty} \quad \le  C(S,d_{\min},d_{\max})\frac{(2\Delta d_{\max})^{l+1}}{(l+1)!}+C'(S,d_{\min},d_{\max}) \,\Delta^{2l} \zeta\,,
\end{align*}
with $C(S,d_{\min},d_{\max}):=4\|S\|_\infty^6{e^{4d_{\max}}\frac{1-e^{-2d_{\max}}}{1-e^{-2d_{\min}}}}$ and $C'(S,d_{\min},d_{\max}):={2\|S\|^4_\infty e^{4 d_{\max}} (1-e^{-2d_{\max}})^2}$. 
Let us also denote $C''(\Delta,d_{\max})=2\Delta d_{\max}e$.
Therefore, by choosing 

\begin{align*}
&l= \lfloor{C''\exp{\left(W\left(\frac{\log\frac{C}{\eps'}}{C''}\right)\right)}}\rfloor,\, \quad \zeta=\min\left\{\frac{\epsilon'}{C'\Delta^{2l}},\,\frac{e^{-2d_{\max}}}{2\Delta^l\|S\|_\infty^2(1-e^{-2d_{\max}})}\right\}\,,
\end{align*}
we obtain, via Lemma~\ref{lem:Lambert}
\begin{align*}
\|(2V-i\Omega)^{-1}-\operatorname{LI}_{\mathcal{N}(l)}(2\hat{V}-i\Omega)\|_{\infty}\le 2\epsilon'\,.
\end{align*}
Using Equation~\eqref{eq:backtoh} as well as the perturbation bound in Lemma~\ref{th:backpert}, and denoting by $\hat{H}$ the matrix associated to $\hat{V}$ through the following integral representation:

\begin{align}
\hat{H}&=\operatorname{LI}_{\mathcal{N}}(2\hat{V}-i\Omega)i\Omega\int_{0}^{\infty}\frac{I}{I+\frac{t}{t+1}2\operatorname{LI}_{\mathcal{N}}(2\hat{V}-i\Omega)i\Omega}\frac{dt}{(t+1)^2}i\Omega\,,\label{eq:locinvtoh}
\end{align}
we get 
\begin{align*}
\|H-\hat{H}\|_{\infty}&\leq e^{2d_{\max}}\|S\|^2_{\infty}2\epsilon'\left(1+\frac{e^{2d_{\max}}\|S\|_{\infty}^2(1-e^{-2d_{\max}})+4\epsilon'}{e^{-2d_{\max}} \|S\|_{\infty}^{-2}-4\epsilon'}\right)\,,
\end{align*}
provided $4\epsilon'< e^{-2d_{\max}}\|S\|_{\infty}^{-2}$.
Choosing 
$\epsilon'=\epsilon \frac{e^{-2d_{\max}}\|S\|^{-2}_{\infty}}{8\left(1+{e^{4d_{\max}}\|S\|_{\infty}^4(1-e^{-2d_{\max}})}\right)}<\frac{e^{-2d_{\max}}\|S\|_{\infty}^{-2}}{4}$ we have

\begin{align*}
\left(1+\frac{e^{2d_{\max}}\|S\|_{\infty}^2(1-e^{-2d_{\max}})+4\epsilon'}{e^{-2d_{\max}} \|S\|_{\infty}^{-2}-4\epsilon'}\right)\leq 2\left(1+{e^{4d_{\max}}\|S\|_{\infty}^4(1-e^{-2d_{\max}})}\right)\,,
\end{align*}
and thus

\begin{align*}
\|H-\hat{H}\|_{\infty}&\leq \epsilon\,.
\end{align*}

\end{rem}

\subsection{Gaussian Hamiltonian learning}\label{sec:hamlearn}
The Hamiltonian learning problem consists on learning an estimate of the Hamiltonian of a Gaussian state from copies of it. More formally:
\begin{defi}[Problem 2: Gaussian Hamiltonian learning]\label{def:Hamiltonian_learning}
Let $H\in M^{>0}_{2m}(\mathbb{R})$ be a symmetric, positive definite matrix (later referred to as Gaussian Hamiltonian) and $t\in \mathbb{R}^{2m}$. For any such data $(H,t)$, one can construct a Gaussian state on $m$ modes: $$\rho(H,t)=\frac{e^{-(R-t)^{\intercal}H(R-t)}}{\Tr[e^{-(R-t)^{\intercal}H(R-t)}]}\,.$$ 
The Gaussian Hamiltonian learning problem consists in the following task: given a fixed precision $\epsilon>0$, provide estimates $\hat H$ of $H$ and $\hat t$ of $t$ such that $\|\hat{H}-H\|_{\infty}\leq \epsilon$ and $\|\hat t-t\|_\infty\le \epsilon$ with probability at least $1-\delta$, given $n$ copies of $\rho(H,t)$, for any $(H,t)$ in some specified subset $\mathcal{S}$. We denote by $\mathcal{N}_{\mathrm{HL}}(\mathcal{S},m,\epsilon,\delta)$ the sample complexity for Gaussian Hamiltonian learning. This corresponds to the minimum number $n$ of copies of the state $\rho(H,t)$ required for the task to succeed.  
\end{defi}
We consider two regimes of interest. It follows from Equation~\eqref{equ:error_localinverse} that the growth rate of $\xi(l)$, i.e. how many vertices are at a distance at most $l$, governs how good the approximation of the local inverse is. Thus, we will first work in the physically relevant regime where it is polynomial in $l$, such as in the case of lattices, and then handle the general case where we just assume a bound on the maximal degree of the graph.

\begin{thm}[Hamiltonian learning, polynomially growing neighborhoods]\label{th:polynomiallyham}
Let $H$ be a Gaussian Hamiltonian with interaction graph $G$ such that $\xi(l)\leq gl^r$ for some $g,r\geq1$, and such that $H={S^{-\intercal}DS^{-1}}$, with $S$ symplectic, $D$ positive definite, diagonal with entries $(d_1,d_1,..., d_m,d_m)$, and $d_{\max}=\max_{i\in{[m]}}d_i$, $d_{\min}=\min_{i\in{[m]}}d_i$. 
Furthermore, to simplify the presentation, we let:
\begin{align}
&C_1:=\frac{e^{4d_{\max}}\|S\|_{\infty}^6}{2}\frac{1-e^{-2d_{\max}}}{1-e^{-2d_{\min}}},\quad C_2:=\frac{e^{2d_{\max}}\|S\|_{\infty}^2}{2},\quad
   C_3:=2\|S\|^4_\infty e^{4 d_{\max}}(1-e^{-2d_{\max}})^2,\quad\\ &C_4:=3C_1 g(2d_{\max}e)^r\nonumber\,.
\end{align}
Given $1>\epsilon>0$ and $\delta\in (0,1)$ and for any $t\in\mathbb{R}^{2m}$,
\begin{align}
&N= \left(g\left[r+(2d_{\max}  e)e^{W(\frac{1}{2d_{\max}  e}\log{\frac{C_4}{\epsilon})}}\right]^{2r}\right)^{2}\times\frac{e^{4d_{\max}}\|S\|^{4}_{\infty}}{(8C_3\left(1+{e^{4d_{\max}}\|S\|_{\infty}^4(1-e^{-2d_{\max}}))}\right)^{-2}}\times \\&\frac{2^{14} (\|S\|^2_{\infty}\frac{1}{1 -e^{-2d_{\min}}}+\max_i|t_i|^2)^2(1+\max_i|t_i|)^2}{\epsilon^2}\log\left(\frac{4m+1}{\delta}\right)=\\
&\mathcal{O}\left(\epsilon^{-2}\log(m\delta^{-1})\left(g\left[r+(2d_{\max}  e)e^{W(\frac{1}{2d_{\max}  e}\log{\frac{C_4}{\epsilon})}}\right]^{2r}\right)^{2}\frac{\|S\|^{24}e^{12d_{\max}}(1-e^{-2d_{\max}})^2}{(1-e^{-2d_{\min}})^{2}}(1+\max_i|t_i|)^6)\right)
\end{align}
copies of the Gaussian state with parameters $(H,t)$ suffice to derive estimates $\hat H$ and $\hat t$ such that
 $\|\hat{H}-H\|_{\infty},\,\|\hat t-t\|_\infty\leq \epsilon$, with probability larger than $1-\delta$.
In particular, for $r=\mathcal{O}(1)$, focusing on the scaling in terms of $\epsilon$ and $m$ we get:
\begin{align}
N=\mathcal{O}(\epsilon^{-2}\log^{2r}(\epsilon^{-1})\log(m)).
\end{align}
\end{thm}
\begin{proof}
Let us start by controlling how to pick $l$ to ensure that the r.h.s. of Equation~\eqref{equ:error_localinverse} is at most $\epsilon'$ for some some $\epsilon_1>0$ to be fixed later.

The bound then becomes:
\begin{align}
\|(2V-i\Omega)^{-1}-\operatorname{LI}_{\mathcal{N}(l)}(2\hat{V}-i\Omega)\|_{\infty}&\le C_1\frac{\xi(l)(2d_{\max})^{l+1}}{(l+1)!}+C_2\frac{(2\xi(l)+1)(2d_{\max})^{l+1}}{(l+1)!}+C_3\xi(l)^2\zeta \nonumber\\
&\leq 3C_1 g (l+1)^r\frac{(2d_{\max})^{l+1}}{(l+1)!}+C_3\xi(l)^2\zeta,
\end{align}

where we used $g>1$ and $C_1\geq C_2$.

Taking $l=\lfloor r+(2d_{\max}  e)e^{W(\frac{1}{2d_{\max}  e}\log{\frac{C_4}{\epsilon})}}\rfloor$, we have that:

\begin{align}
3C_1 g (l+1)^r\frac{(2d_{\max})^{l+1}}{(l+1)!}&\leq 3C_1 g(l+1)^r\frac{(2d_{\max}e)^{l+1}}{(l+1)^{l+1}}\nonumber\\&\leq3C_1 g(2d_{\max}e)^r\frac{(2d_{\max}e)^{l+1-r}}{(l+1-r)^{l+1-r}}\leq \epsilon,
\end{align}

similarly to Eq.~(\ref{eq:WLamb}).

Thus, we see that if we pick $\zeta=\epsilon_1C_3^{-1}\left(g\left[r+(2d_{\max}  e)e^{W(\frac{1}{2d_{\max}  e}\log{\frac{C_4}{\epsilon})}}\right]^{2r}\right)^{-1}$, then we have that:
\begin{align}
\|(2V-i\Omega)^{-1}-\operatorname{LI}_{\mathcal{N}(l)}(2\hat{V}-i\Omega)\|_{\infty}\leq 2\epsilon_1.
\end{align}
Following the same analysis as in Remark~\ref{rem:hbound}, we get that setting $\epsilon_1$ as 
\begin{align}
\epsilon_1=\epsilon \frac{e^{-2d_{\max}}\|S\|^{-2}_{\infty}}{8\left(1+{e^{4d_{\max}}\|S\|_{\infty}^4(1-e^{-2d_{\max}})}\right)}
\end{align}
is sufficient to ensure that $\|H'-H\|_{\infty}\leq \epsilon$.
Putting all the elements together, we see that we need to estimate the entries of the covariance matix up to a precision 
\begin{align}
\zeta=\left(g\left[r+(2d_{\max}  e)e^{W(\frac{1}{2d_{\max}  e}\log{\frac{C_4}{\epsilon})}}\right]^{2r}\right)^{-1}\times\epsilon \frac{e^{-2d_{\max}}\|S\|^{-2}_{\infty}}{8C_3\left(1+{e^{4d_{\max}}\|S\|_{\infty}^4(1-e^{-2d_{\max}})}\right)},
\end{align}
which yields the advertised bound.
\end{proof}

\begin{rem}\label{remarklattice} For Hamiltonians associated to graphs with polynomial neighborhood growth, one can expect an exponential decay of correlation if $d_{\min}$ is uniformly bounded. For example~\cite{Cramer_2006, Schuch_2006}) consider some subclasses of such Hamiltonians and show that the correlation length, in our notation, is roughly $\xi\sim \frac{\Delta}{\left(\frac{d_{\min}}{d_{\max}}\right)^{\alpha}}$ for some $\alpha>0$ (if the inverse temperature is set to $d_{\max}$, $(\frac{d_{\min}}{d_{\max}})$ can also be interpreted as the energy gap of the Hamiltonian). 
If $\xi$ does not scale with $m$, the correlation decay is strong enough so that one could estimate only a number of entries of $V$ that does not scale with $m$ at error $\epsilon$, set the entries corresponding to far vertices to zero, and still obtain an estimate of $V$ which is of error $\mathcal{O}(\epsilon)$. This would allow to use the continuity bound of Lemma~\ref{th:continuityhamv} to obtain an $\mathcal{O}(\epsilon)$ error on the estimate of $H$. However, $\xi$ is $\mathcal{O}(m^\alpha)$ if the energy gap scales with $1/m$, so that the above strategy would not be enough to guarantee $\mathcal{O}(\mathrm{poly}\log(m))$ sample complexity. 
Instead, the estimation algorithm using local inversions has a sample complexity upper bounded by a polynomial function of $\max_{i}|V_{ii}|$ and $\log\frac{1-e^{-2d_{\max}}}{1-e^{-2d_{\min}}}\leq 2 d_{\min}+\log\frac{d_{\max}}{d_{\min}}\leq 2 d_{\max}+\log\frac{d_{\max}}{d_{\min}}$, making the sample complexity still $\mathcal{O}(\mathrm{poly}\log(m))$ if $\max_{i}|V_{ii}|$ does not scale with $m$ and $d_{\max}$ is uniformly bounded. 
\end{rem}

Thus, we see that assuming polynomially-growing neighborhoods, it is possible to obtain a sample complexity that scales like $\epsilon^{-2}\log(m)$, up to poly-log factors and, of course, neglecting the dependency on the various other parameters of the problem.
In the setting where we just assume a bound on the maximal degree of the interaction graph we obtain a higher sample complexity:

\begin{thm}[Hamiltonian learning, bounded degree]
Let $H$ be a Gaussian Hamiltonian with interaction graph $G$, with degree $\Delta-1$, such that $H={S^{-\intercal}DS^{-1}}$, with $S$ symplectic, $D$ positive definite, diagonal with entries $(d_1,d_1,..., d_m,d_m)$, and $d_{\max}=\max_{i\in{[m]}}d_i$, $d_{\min}=\min_{i\in{[m]}}d_i$. Given $1>\epsilon>0$ and $\delta\in (0,1)$, choose
\begin{align*}
\epsilon'&\leq\frac{\epsilon e^{-2d_{\max}} \|S\|_{\infty}^{-2}}{8}\frac{1}{1+e^{4d_{\max}}\|S\|_{\infty}^4(1-e^{-2d_{\max}})}\nonumber\\
l&=\lfloor(2\Delta d_{\max} e)e^{\left(W\left(\frac{\log\frac{C}{\eps'}}{2\Delta d_{\max }e}\right)\right)}\rfloor\\
\zeta&\leq\min\left\{\frac{\epsilon'}{C'\Delta^{2l}},\,{\frac{e^{-2d_{\max}}}{2\Delta^l\|S\|_\infty^2(1-e^{-2d_{\max}})}}\right\}\,,
\end{align*}
where $C(S,d_{\min},d_{\max}):=2\|S\|_\infty^6{e^{4d_{\max}}\frac{1-e^{-2d_{\max}}}{1-e^{-2d_{\min}}}}$ and $C'(S,d_{\min},d_{\max}):= {2\|S\|^4_\infty e^{4 d_{\max}} (1-e^{-2d_{\max}})^2}$. 

Then, for any $t\in\mathbb{R}^{2m}$, 

\begin{align}
N= \frac{2^{14} (\|S\|^2_{\infty}\frac{1}{1 -e^{-2d_{\min}}}+\max_i|t_i|^2)^2(1+\max_i|t_i|)^2}{\zeta^2}\log\left(\frac{4m+1}{\delta}\right)
\end{align}
copies of the Gaussian state with parameters $(H,t)$ suffice to derive estimates $\hat t$ and $\hat H$ (via Algorithm~\ref{alg:learnHG} with parameter $l$)  such that
 $\|\hat{H}-H\|_{\infty},\,\|\hat t-t\|_\infty\leq \epsilon$, with probability larger than $1-\delta$.

\end{thm}

\begin{proof}
It follows from Remark~\ref{rem:hbound} and Lemma~\ref{lem:cov_matrix}.
\end{proof}

The following corollary is an immediate consequence:

\begin{cor}[Gaussian Hamiltonian learning, general graphs]\label{cor:generalgraph}
Let $\mathcal{S}$ be the class of Hamiltonians such that $\|S\|_{\infty}\leq s$, $d_{\max}\leq \beta_{\max}$, $d_{\min}\geq \beta_{\min}$, $\max_{i}|t_i|\leq t_{\max}$, and degree $\Delta-1$. Then, with $C=\frac{s(1-e^{-2\beta_{\max}})}{1-e^{-2\beta_{\min}}}e^{2\beta_{\max}}$, for every $\gamma>0$,

\begin{align}
\mathcal{N}(\mathcal{S},m,\epsilon,\delta)=\log\left(\frac{4m+1}{\delta}\right)\mathcal{O}\left(\frac{\mathrm{poly}\{\Delta^{\Delta\beta_{\max}}
C\left(\frac{C}{\epsilon}\right)^{\gamma/\Delta \beta_{\max}} (1+t_{\max})\}}{\epsilon^{2}}\right)\,.
\end{align}
\end{cor}

\begin{rem}\label{remarkV2} As already highlighted in Remark~\ref{remarkV1}, bounding $V_{ii}$ by $\|V\|_{\infty}$ and thus in terms of $d_{\min}$ may be overshooting in the regime of large $m$. Here, since a dependence on $d_{\min}$ appears anyway from the rest of the proof, we expressed everything in terms of $d_{\min}$. However, the only other dependence on $d_{\min}$ comes from bounding $\|D^{-1}C\|_{\infty}$ in Lemma~\ref{theoapprinvsparse}. As noted in Remark~\ref{remarkcondition}, one could expect that bounding $\|D^{-1}C\|_{\infty}$ via the condition number of $(2V-i\Omega)^{-1}$ is also a loose bound, and the proof can be expressed as well just in terms of a uniform bound on all the instances of $\|D^{-1}C\|_{\infty}$ from each call to Remark~\ref{theoapprinvsparse}. Putting together these observations suggests that in practice the number of copies needed could be much smaller than what the corollary above guarantees. The same observation holds for our later Theorem~\ref{theographlearning}.
\end{rem}

\subsection{Learning the graph}\label{sec:graphlearn}

We now consider the problem of learning the interaction graph, defined as follows. 

\begin{defi}[Problem 3: Gaussian graph learning]\label{def:graph_learning}
Let $H\in M^{>0}_{2m}(\mathbb{R})$ be a symmetric, positive definite matrix (later referred to as Gaussian Hamiltonian) and $t\in \mathbb{R}^{2m}$. For any such data $(H,t)$, one can construct a Gaussian state on $m$ modes: $$\rho(H,t)=\frac{e^{-(R-t)^{\intercal}H(R-t)}}{\Tr[e^{-(R-t)^{\intercal}H(R-t)}]}\,.$$ 
The Gaussian graph learning problem consists in the following task: 
 provide an estimate $\hat{\mathsf{E}}$ of the edge set $\mathsf{E}$ of the graph of $H$, such that $\hat{\mathsf{E}}=\mathsf{E}$ with probability larger than $1-\delta$ given $n$ copies of $\rho(H,t)$, for any $(H,t)$ in some specified subset $\mathcal{S}$ which satisfies a lower bound $\kappa$ on the maximum absolute value of the entries of non-zero sub-matrices of the Hamiltonian:
\begin{align*}
\kappa\leq \min_{(i,j)\in \mathsf{E}}\max_{\delta_1,\delta_2\in\{0,1\}}|H_{2i-\delta_1,2j-\delta_2}|\,.
\end{align*}
We denote by $\mathcal{N}_{\mathrm{GL}}(\mathcal{S},\kappa,m,\epsilon,\delta)$ the sample complexity for Gaussian Hamiltonian learning. This corresponds to the minimum number $n$ of copies of the state $\rho(H,t)$ required for the task to succeed.  
\end{defi}
Again, this is the natural quantum variation of learning the structure of a Gaussian graphical model~\cite{misra20a}.
We will need the following variant of Lemma~\ref{th:Hsparsebound} where the neighborhood structure is also being approximated.

\begin{lem}\label{lem:notgoodneigh}
 
Let $M$ be a positive definite matrix and let $\mathcal{N}=\{\mathcal{N}_i\}_{i\in[m]}$, $\hat{\mathcal{N}}=\{\hat{\mathcal{N}}_i\}_{i\in[m]}$,  be two neighborhood structures of $[m]$. For each $i\in[m]$ and $\delta\in\{0,1\}$, let $\left(\begin{smallmatrix}
    A_{i,\delta} &B_{i,\delta}\\ C_{i,\delta}&D_{i,\delta}
\end{smallmatrix}\right)\in\mathbb{C}^{2m}\times \mathbb{C}^{2m}$, with blocks $A_{i,\delta}\in \mathbb{C}^{2|\hat{\mathcal{N}}_i|}\times\mathbb{C}^{2|\hat{\mathcal{N}}_i|} ,B_{i,\delta}\in \mathbb{C}^{2|\hat{\mathcal{N}}_i|}\times\mathbb{C}^{2m-2|\hat{\mathcal{N}}_i|} ,C_{i,\delta}\in \mathbb{C}^{2m-2|\hat{\mathcal{N}}_i|}\times\mathbb{C}^{2|\hat{\mathcal{N}}_i|} ,D_{i,\delta}\in \mathbb{C}^{2m-2|\hat{\mathcal{N}}_i|}\times\mathbb{C}^{2m-2|\hat{\mathcal{N}}_i|} $ be the matrix $\pi_{i,\delta}M \pi_{i,\delta}^{-1}$ where $\pi_{i,\delta}$ is a permutation that brings the coefficients $2j-\delta'$ with $\delta'\in\{0,1\}$ and $j\in\hat{\mathcal{N}}_i$ to the upper-left block $A_{i,\delta}$, in such a way that index $2i-\delta$ is mapped to $2-\delta$ for $\delta=0,1$. Let $\widetilde{B}_{i,\delta}$ the matrix constructed from $B_{i,\delta}$ by putting the first row to zero. Let $E$ be a self-adjoint matrix, such that $\operatorname{LOC}_{\mathcal{N}}(M-E)=M-E$. Then,
\begin{align}\label{eq:opnormbound23}
\|M-\operatorname{LI}_{\hat{\mathcal{N}}}(M^{-1})\|_{\infty}&\leq 2\max_{i\in [m]}\max_{\delta\in\{0,1\}}{|\hat{\mathcal{N}}_i|}\|B_{i,\delta}-\widetilde{B}_{i,\delta}\|_{\infty} \|M^{-1}\|_{\infty}\|M\|_{\infty}+(2\max_{i\in [m]}{|\hat{\mathcal{N}}_i|}+1)\|E\|_{\infty}\nonumber\\
&+\max_{i\in[m]}\max_{\delta_1\in\{0,1\}}\sum_{j\in \mathcal{N}_i, j\notin \hat{\mathcal{N}}_i,\delta_2\in\{0,1\}}|(M-E)_{2i-\delta_1,2j-\delta_2}|.
\end{align}
\end{lem}

\begin{proof}
The proof is essentially the same as that of Lemma~\ref{th:Hsparsebound}, except that in the final application of Gershgorin's theorem the matrix $M-E-\operatorname{LI}_{\hat{\mathcal{N}}}(M^{-1})$ can have non-zero elements for entries corresponding to the neighborhood structure $\mathcal{N}$ but not $\hat{\mathcal{N}}$. We reproduce the aformentioned proof with the necessary changes.
By triangle inequality, we have 
 \begin{align*}
 \|M-\operatorname{LI}_{\hat{\mathcal{N}}}(M^{-1})\|_{\infty}&\le\|E\|_\infty+\|M-E-\operatorname{LI}_{\hat{\mathcal{N}}}(M^{-1})\|_{\infty}\,.
 \end{align*}
 It remains to bound the last operator norm above. Since $\operatorname{LOC}_{\mathcal{N}}(M-E)-\operatorname{LI}_{\hat{\mathcal{N}}}(M^{-1})=M-E-\operatorname{LI}_{\hat{\mathcal{N}}}(M^{-1})$ is self-adjoint, it suffices to find a uniform bound on its eigenvalues. By Gershgorin circle theorem, denoting the matrix $L:=\operatorname{LOC}_{\mathcal{N}}(M-E)-\operatorname{LI}_{\hat{\mathcal{N}}}(M^{-1})$, these eigenvalues are situated in the $2m$ intervals
 \begin{align*}
     \left[L_{k,k}-\sum_{k'\ne k}\big|L_{k,k'}\big|,\,L_{k,k}+\sum_{k'\ne k}\big|L_{k,k'}\big|\right],\, k\in[2m]\,.
 \end{align*}

 Now for any $i\in[m]$, $j\in\hat{\mathcal{N}}_i$ and $\delta'\in\{0,1\}$, we have by the construction of the local inverse matrix, denoting $N=M^{-1}$, and $N_i$ the block matrices defined similarly to Definition \ref{constrappr1} but 
 corresponding to the neighborhood structure $\hat{\mathcal{N}}$:
 \begin{align}
 \Big|M_{2i-\delta,2j-\delta'}-\big(\operatorname{LI}_{\hat{\mathcal{N}}}(M^{-1})\big)_{2i-\delta,2j-\delta'}\Big|&\leq\frac{1}{2}\Big|M_{2i-\delta,2j-\delta'}-\big(N_{i}^{-1}\big)_{2i-\delta,2j-\delta'}\Big|\nonumber\\
 &+\frac{1}{2}\Big|M_{2i-\delta,2j-\delta'}-\big(N_{j}^{-1}\big)^*_{2j-\delta',2i-\delta}\Big|\nonumber\,.
 \end{align}
  Via Lemma~\ref{theoapprinvsparse} applied to the matrices  $\pi_{i,\delta}M\pi_{i,\delta}^{-1}$ and $\pi_{j,\delta}M\pi_{j,\delta}^{-1}$,
\begin{align*}
\Big|M_{2i-\delta,2j-\delta'}-\big(N_{i}^{-1}\big)_{2i-\delta,2j-\delta'}\Big|&\leq  \|B_{i,\delta}-\widetilde{B}_{i,\delta}\|_\infty \|M^{-1}\|_{\infty}\|M\|_{\infty}\, ,\\
\Big|M_{2i-\delta,2j-\delta'}-\big(N_{j}^{-1}\big)^*_{2j-\delta',2i-\delta}\Big|&=\Big|M^*_{2j-\delta',2i-\delta}-\big(N_{j}^{-1}\big)^*_{2j-\delta',2i-\delta}\Big|
&\!\!\le\!  \|B_{j,\delta}-\widetilde{B}_{j,\delta}\|_\infty \|M^{-1}\|_{\infty}\|M\|_{\infty}
\end{align*}
and hence $|L_{2i-\delta,2j-\delta'}|\leq |(M-\big(\operatorname{LI}_{\hat{\mathcal{N}}}(M^{-1})\big)_{2i-\delta,2j-\delta'}|+|E_{2i-\delta,2j-\delta'}|\leq \max_{i\in[m]}\max_{\delta\in\{0,1\}}\|B_{i,\delta}-\widetilde{B}_{i,\delta}\|_\infty \|M^{-1}\|_{\infty}\|M\|_{\infty} +\|E\|_{\infty}$. The other non-zero entries of $L$ are $L_{2i-\delta_1,2j-\delta_2}$ such that $j\in\mathcal{N}_i$ but $j\notin\hat{\mathcal{N}}_i$, and $L_{2i-\delta_1,2j-\delta_2}=(M-E)_{2i-\delta_1,2j-\delta_2}$. Using Gershgorin's theorem together with the fact that the matrix $L$ has at most $2\max_i|{\hat{\mathcal{N}}}_i|$ non-zero coefficients per column that corresponds to $j\in\hat{\mathcal{N}_i}$, we get
\begin{align*}
 \|L\|_\infty&\le 2\max_{i\in[m]}\max_{\delta\in\{0,1\}}|\hat{\mathcal{N}}_i|\left(\|B_{i,\delta}-\widetilde{B}_{i,\delta}\|_\infty \|M^{-1}\|_{\infty}\|M\|_{\infty} +\|E\|_{\infty}\right)\\
 &+\max_{i\in[m]}\max_{\delta_1\in\{0,1\}}\sum_{j\in \mathcal{N}_i, j\notin \hat{\mathcal{N}}_i,\delta_2\in\{0,1\}}|(M-E)_{2i-\delta_1,2j-\delta_2}| \,.
\end{align*}
The result follows.
\end{proof}

The error made during the local inversion can be bounded as in Lemma~\ref{th:approxlocinvbound}:

\begin{lem}\label{lem:estpertaprrneigh}
In the notations of Lemma \ref{lem:notgoodneigh}, given a $2m\times 2m$ positive definite matrix $M$, $N=M^{-1}$ and $\hat{N}$ a componentwise approximation of $N$, where $|\hat{N}_{k,k'}-N_{k,k'}|\le \zeta$ for all $k,k'\in[2m]$ with $\zeta\le \frac{1}{4\max_{i\in[m]}|{\hat{\mathcal{N}}}_i|}\|M\|_\infty^{-1}$,

\begin{align}
\|M-\operatorname{LI}_{\hat{\mathcal{N}}}(\hat{N})\|_{\infty}&\leq 2\max_{i\in [m]}\max_{\delta\in\{0,1\}}{|\hat{\mathcal{N}}_i|}\|B_{i,\delta}-\widetilde{B}_{i,\delta}\|_{\infty} \|M^{-1}\|_{\infty}\|M\|_{\infty}+(2\max_{i\in [m]}{|\hat{\mathcal{N}}_i|}+1)\|E\|_{\infty}\\
 &+\max_{i\in[m]}\max_{\delta_1\in\{0,1\}}\sum_{j\in \mathcal{N}_i, j\notin \hat{\mathcal{N}}_i,\delta_2\in\{0,1\}}|(M-E)_{2i-\delta_1,2j-\delta_2}|+8\|M\|_{\infty}^{2}\max_{i}|\hat{\mathcal{N}}_i|^2\zeta\,\nonumber.
\end{align}

\end{lem}

\begin{proof}
By the triangle inequality,
\begin{align*}
\|M-\operatorname{LI}_{\hat{\mathcal{N}}}(\hat{N})\|_{\infty}\leq \|M-\operatorname{LI}_{\hat{\mathcal{N}}}(N)\|_{\infty}+\|\operatorname{LI}_{\hat{\mathcal{N}}}(N)-\operatorname{LI}_{\hat{\mathcal{N}}}(\hat{N})\|_{\infty}.
\end{align*}
The first term is bounded as in Lemma~\ref{lem:notgoodneigh}, the second as in Lemma~\ref{th:approxlocinvbound}.

\end{proof}
With these estimates at hand we can show that Algorithm \ref{alg:learnGG} stated in pseudocode below outputs a correct estimate of the graph with high probability. Before stating the theorem, let us give an overview of the algorithm. The algorithm takes as an input an estimate of the covariance matrix $\hat{V}$ and size of neighborhoods $\xi(l)\le \Delta^l$, plus  thresholding parameters $\eta,\kappa$. It then iterates over all vertices and all neighborhoods of size $\xi(l)$, denoted by $\mathcal{N}_i$. For each of the $\mathcal{N}_i$, we then consider all possible neighborhoods enlarged by $\xi(l)$ additional vertices, defining a set $\overline{\mathcal{N}}_i$, and perform our local inversion with the enlarged neighborhood. If we see that in the resulting matrix there are no large entries (larger than $3\eta$) on row $i$ and vertices $j\in\overline{\mathcal{N}}_i$, we then set our neighborhood as $\mathcal{N}_i$. The intuition behind the procedure is as follows: recall that we assume that we have a lower bound on the size of all interactions and that matrix elements of $(2V-i\Omega)^{-1}$ corresponding to vertices that are far apart decay rapidly. If we have the right $l$-neighborhood already for some vertex $i$, then all the estimates with the enlarged neighborhood will be approximately correct and all the entries added will be smaller than $3\eta$ because of the decay. Thus, we will correctly accept in this case. On the other hand, if we have the incorrect or incomplete neighborhood, then there must be $\xi(l)$ vertices we can add that will make it complete. For that joint neighborhood, we see two possibilities: either this will ``turn on'' a large entry, which will show us that we were not operating with the right neighborhood. In contrast, if we see that for all possibilities we add this never occurs, then we can also conclude that the vertices we are missing will only add a small error and this can be neglected. At the end, we end up with an estimate of the Hamiltonian matrix which has error at most $\kappa/2$ in operator norm, so that we can identify non-edges thanks to the promise on the strength of the interactions.

More formally, we have:

\begin{thm}\label{theographlearning}
Let $H$ be a Hamiltonian with graph $G$ of degree $\Delta-1$ and edge set $\mathsf{E}$, the condition $0<\kappa\leq \min_{(i,j)\in \mathsf{E}}\max_{\delta_1,\delta_2\in\{0,1\}}|H_{2i-\delta_1,2j-\delta_2}|$, such that $H=S^{{-\intercal}}DS^{{-1}}$, with $S$ symplectic, $D$ diagonal with entries $(d_1,d_1,..., d_m,d_m)$, and $d_{\max}=\max_{i\in [m]}d_i$, $d_{\min}=\min_{i\in [m]}d_i$. {Let us assume that $\|S\|_{\infty}\leq s$, $d_{\max}\leq \beta_{\max}$, $d_{\min}\geq \beta_{\min}$}. We also denote by $\xi(l)$ the maximum number of vertices at distance at most $l$ from any given vertex. 
{
For $l$, $\zeta$, $\epsilon'$, $\eta$ given by
\begin{align*}
l&= \left\lfloor(2\Delta^2 \beta_{\max} e)\exp{\left(W\left(\frac{\log\frac{C(s,\beta_{\min},\beta_{\max})}{\eps'}}{2\Delta^2 \beta_{\max }e}\right)\right)}\right\rfloor\nonumber\\
\zeta&=
{ \frac{\epsilon' e^{-4\beta_{\max}}}{200\, \Delta^{3l}s^6}\frac{1-e^{-2\beta_{\min}}}{1-e^{-2\beta_{\max}}}}
\nonumber\\
\epsilon'&=\min\left(\kappa\frac{e^{-2\beta_{\max}}s^{-2}}{16\left(1+{e^{4\beta_{\max}}s^4(1-e^{-2\beta_{\max}})}\right)}, \frac{e^{-2\beta_{\max}}s^{-2}}{8}\right)\\
\eta&= \frac{\epsilon'}{e^{2\beta_{\max}}36\, \Delta^{2l}s^4}{\frac{(1-e^{-2\beta_{\min}})}{(1-e^{-2\beta_{\max}})}}\,,
\end{align*}
where we denoted $C(s,\beta_{\min},\beta_{\max}):=25s^{10}e^{6\beta_{\max}}{\frac{(1-e^{-2\beta_{\max}})^2}{(1-e^{-2\beta_{\min}})^2}}$ and $C'(s,\beta_{\max}):=2s^4e^{4\beta_{\max}}$, 
given 
\begin{align}
{N}=\mathcal{O}\left(\frac{2^{14} (s^2\frac{1}{1 -e^{-2\beta_{\min}}}+\max_i|t_i|^2)^2(1+\max_i|t_i|)^2}{\zeta^2}\log\left(\frac{4m+1}{\delta}\right)\right)
\end{align}
}
copies of the Gaussian state with parameters $(H,t)$, \textsf{LearnGraphGaussian} with parameters $\xi(l),\eta, \kappa$ (Algorithm~\ref{alg:learnGG}) outputs an estimate $\hat{\mathsf{E}}$ of the edge set $\mathsf{E}$ such that $\hat{\mathsf{E}}=\mathsf{E}$
 with probability larger than $1-\delta$.
 Furthermore, we can run Algorithm~\ref{alg:learnGG} in time
 \begin{align}\label{equ:time_complexity_algo}
     \mathcal{O}\left(m\binom{m-\xi(l)}{\xi(l)}\binom{m}{\xi(l)}\xi(l)^3\right)\,.
 \end{align}

\end{thm}

\begin{figure}[ht]

\centering
\begin{minipage}{.9\linewidth}
\begin{algorithm}[H]
\caption{\textsf{LearnGraphGaussian}} 
	\label{alg:learnGG}
	\begin{algorithmic}[1]
        \State Input: estimate $\hat{V}$ of the covariance matrix $V$ of $m$-mode Gaussian state $\rho$, (bound on) cardinality of neighborhood $\xi$, threshold parameters $\eta,\kappa>0$.
\For{$i\in [m]$}
                    \For{$\mathcal{N}_i$ candidate neighborhoods of $i$ of cardinality $\xi(l)$}
                                    \State set $\hat{\mathcal{N}}_i=\mathcal{N}_i$
                                    \For{${\overline{\mathcal{N}}}_i \in$ sets of vertices in $[m]/\{{\mathcal{N}}_i\}$ of cardinality $\xi(l)$}
                                    \State Compute \{$\operatorname{LI}_{\mathcal{N}_i\cup \overline{\mathcal{N}}_i}(2\hat{V}-i\Omega)_{{2i-\delta_1},{2j-\delta_2}}$, $j\in \overline{\mathcal{N}}_i$, $\delta_1,\delta_2\in\{0,1\}\}$     
                                     \EndFor 
                                     \If{ $\max_{\overline{\mathcal{N}}_i}\max_{j\in \overline{\mathcal{N}}_i,\delta_1,\delta_2\in\{0,1\}}|\operatorname{LI}_{\mathcal{N}_i\cup \overline{\mathcal{N}}_i}(2\hat{V}-i\Omega)_{{2i-\delta_1},{2j-\delta_2}}|\leq 3\eta$}
                                     \State \textbf{break}
                                     \EndIf
                                     \EndFor
                \EndFor
        \State For each $i\in[m]$, $j\in\hat{{\mathcal{N}}_i}$, eliminate element $j$ from $\hat{\mathcal{N}}_i$ if $i\notin\hat{\mathcal{N}}_j$.
        \State Compute estimate $\hat{H}$ of $H$ from $\operatorname{LI}_{\hat{\mathcal{N}}}(2\hat{V}-i\Omega)$ as in Equation~\ref{eq:locinvtoh}. 
        \State Output the edges $(i,j)$ such that ${\max_{\delta_1\in\{0,1\},\delta_2\in\{0,1\}}|\hat{H}_{2i-\delta_1,2j-\delta_2}|\geq\kappa/2}$.
	\end{algorithmic} 
\end{algorithm}
\caption{Algorithm for learning the graph of a Gaussian state}
\end{minipage}
\end{figure}

\begin{proof}
To learn the graph, we use Algorithm \ref{alg:learnGG}.
The goal of the algorithm is to find a neighborhood $\hat{{\mathcal{N}}}_i$ of each $i\in[m]$ such that the local inversion gives accurate estimates of the matrix elements corresponding to edges including $i$. We use the notation of Lemma~\ref{lem:notgoodneigh}.

At each test step of Algorithm \ref{alg:learnGG}, local inversions with neighborhoods $\widetilde{\mathcal{N}_i}:={\mathcal{N}}_i\cup \overline{\mathcal{N}}_i$ are performed, where $|\widetilde{\mathcal{N}_i}|\leq 2\xi(l)$ and $\eta>0$ is a fixed threshold parameter.  
Let us make the following assumptions:

\begin{itemize}
\item For any $j\in{\widetilde{\mathcal{N}}}_i$, $\delta_1,\delta_2\in\{0,1\}$, we have $|\operatorname{LI}_{\widetilde{\mathcal{N}}}(2V-i\Omega)_{2i-\delta_1,2j-\delta_2}-\operatorname{LI}_{\widetilde{\mathcal{N}}}(2\hat{V}-i\Omega)_{2i-\delta_1,2j-\delta_2}|\leq \eta$.
We can ensure this by increasing the number of copies of $\rho(H,t)$ so that each entry of $\hat V$ is close to the corresponding entry of $V$ with high probability. By the argument of \eqref{eq:inversentry}, with $N=2V-i\Omega$,  $\hat{N}=2\hat{V}-i\Omega$ and $M=(2V-i\Omega)^{-1}$, if $|2\hat{V}_{k,k'}-2V_{k,k'}|\leq{\zeta}$, with $\zeta\le \frac{1}{8\xi(l)}\|(2V-i\Omega)^{-1}\|_\infty^{-1}$ we get 
$|\operatorname{LI}_{\widetilde{\mathcal{N}}}(2V-i\Omega)_{2i-\delta_1,2j-\delta_2}-\operatorname{LI}_{\widetilde{\mathcal{N}}}(2\hat{V}-i\Omega)_{2i-\delta_1,2j-\delta_2}|\leq 8\xi(l)\|(2V-i\Omega)^{-1}\|_\infty\zeta$. Thus, we need to ensure also 
\begin{equation}
\zeta\leq \frac{\|(2V-i\Omega)^{-1}\|_\infty^{-1}\eta}{8\xi(l)}.
\end{equation}
\item  When ${\mathcal{N}}_i$ coincides with the $l$-neighborhood $ \mathcal{N}(l)_i$ (see Definition \ref{nbhdl}), we assume that \\$\max_{\delta\in\{0,1\}}\|B_{i,\delta}-\widetilde{B}_{i,\delta}\|_\infty \|M^{-1}\|_{\infty}\|M\|_{\infty}\leq \eta\,,$ and at the same time $\|E\|_\infty\leq \eta$ for the error $E$ defined in Lemma \ref{th:seriesapprox}, {which we can always ensure by taking $l$ large enough}. In fact, it suffices to impose that 
\begin{align}\label{conditionESdd}
\|E\|_\infty\|M^{-1}\|_{\infty}\|M\|_{\infty}\leq \|E\|_\infty\|S\|_\infty^4\frac{(1-e^{-2d_{\max}})e^{2d_{\max}}}{1-e^{-2d_{\min}}} 
\leq \eta\,,
\end{align}

since $B_{i,\delta}-\widetilde{B}_{i,\delta}$ is a sub-matrix of $E$ and $\|S\|_\infty^4{e^{2d_{\max}}\frac{1-e^{-2d_{\max}}}{1-e^{-2d_{\min}}}}\geq 1$ .

For any ${\mathcal{N}}_i$ such that ${\mathcal{N}}_i(l)\subseteq \mathcal{N}_i$ the same relation will hold.
\end{itemize}

At each search step of the algorithm, there are two cases.

\begin{itemize}
\item ${\mathcal{N}}_i$ already contains $\mathcal{N}(l)_i$, the $l$-neighborhood of $i$. Then, all $\operatorname{LI}_{\mathcal{N}_i\cup \overline{\mathcal{N}}_i}(2\hat{V}-i\Omega)_{2i-\delta_1,2j-\delta_2}$ for $j\in \overline{\mathcal{N}}_i$, $\delta_1,\delta_2\in\{0,1\}$, satisfy
\begin{align*}
&|\operatorname{LI}_{\mathcal{N}_i\cup \overline{\mathcal{N}}_i}(2\hat{V}-i\Omega)_{2i-\delta_1,2j-\delta_2}-(2{V}-i\Omega)^{-1}_{2i-\delta_1,2j-\delta_2}|\\
&\qquad\qquad= |\operatorname{LI}_{\mathcal{N}_i\cup \overline{\mathcal{N}}_i}(2\hat{V}-i\Omega)_{2i-\delta_1,2j-\delta_2}-E_{2i-\delta_1,2j-\delta_2}|\\
&\qquad\qquad\leq |\operatorname{LI}_{\mathcal{N}_i\cup \overline{\mathcal{N}}_i}(2\hat{V}-i\Omega)_{2i-\delta_1,2j-\delta_2}-\operatorname{LI}_{\mathcal{N}_i\cup \overline{\mathcal{N}}_i}(2{V}-i\Omega)_{2i-\delta_1,2j-\delta_2}|\\
&\qquad\qquad\qquad +|\operatorname{LI}_{\mathcal{N}_i\cup \overline{\mathcal{N}}_i}(2{V}-i\Omega)_{2i-\delta_1,2j-\delta_2}-E_{2i-\delta_1,2j-\delta_2}| \,.
\end{align*}

The first term is bounded by $\eta$ by assumption. The second term is bounded via Lemma~\ref{theoapprinvsparse} by $\|B_{i,\delta_1}-\widetilde{B}_{i,\delta_1}\|_\infty \|M^{-1}\|_{\infty}\|M\|_{\infty}$, since the off-diagonal blocks $B_{i,\delta}$ defined with respect to the neighborhood structure $\mathcal{N}$ include those defined for the set $\mathcal{N}_i\cup \overline{\mathcal{N}}_i$, implying 
$$|\operatorname{LI}_{\mathcal{N}_i\cup \overline{\mathcal{N}}_i}(2\hat{V}-i\Omega)_{i+\delta_1,j+\delta_2}|\leq \|B_{i,\delta_1}-\widetilde{B}_{i,\delta_1}\|_\infty \|M^{-1}\|_{\infty}\|M\|_{\infty}+\|E\|_\infty+\eta\leq 3\eta\,.$$
\item If ${\mathcal{N}}_i$ does not contain $\mathcal{N}(l)_i$, there is an adversarial neighborhood $\overline{\mathcal{N}}_i$ such that ${\mathcal{N}}_i\cup \overline{\mathcal{N}}_i$ contains $\mathcal{N}(l)_i$. Suppose there is at least one matrix element $|(2V-i\Omega)^{-1}_{2i-\delta_1,2j-\delta_2}|\geq 6\eta $ with $j\in \overline{\mathcal{N}}_i$,$\delta_1,\delta_2\in\{0,1\}$ for some $\overline{\mathcal{N}}_i$ with the aforementioned property. Then, we have
\begin{align*}
&|\operatorname{LI}_{\mathcal{N}_i\cup \overline{\mathcal{N}}_i}(2\hat{V}-i\Omega)_{2i-\delta_1,2j-\delta_2}-(2{V}-i\Omega)^{-1}_{2i-\delta_1,2j-\delta_2}|\\
&\qquad \leq |\operatorname{LI}_{\mathcal{N}_i\cup \overline{\mathcal{N}}_i}(2\hat{V}-i\Omega)_{2i-\delta_1,2j-\delta_2}-\operatorname{LI}_{\mathcal{N}_i\cup \overline{\mathcal{N}}_i}(2{V}-i\Omega)_{2i-\delta_1,2j-\delta_2}|\\
&\qquad +|\operatorname{LI}_{\mathcal{N}_i\cup \overline{\mathcal{N}}_i}(2{V}-i\Omega)_{2i-\delta_1,2j-\delta_2}-(2{V}-i\Omega)^{-1}_{2i-\delta_1,2j-\delta_2}|,
\end{align*}

and by the same reason as before, we obtain 
\begin{align}
|\operatorname{LI}_{\mathcal{N}_i\cup \overline{\mathcal{N}}_i}(2\hat{V}-i\Omega)_{2i-\delta_1,2j-\delta_2}|\geq |(2{V}-i\Omega)^{-1}_{2i-\delta_1,2j-\delta_2}|-2\eta\geq 4\eta \,.
\end{align}

\end{itemize}

This implies that at the end of the process we obtain a set of neighborhoods $\hat{\mathcal{N}}:=\{{\mathcal{N}}_{i}\}_{i\in[m]}$ such that the matrix elements corresponding to edges which are in actual neighborhoods ${\mathcal{N}}_{i}(l)$ but not in $\hat{\mathcal{N}}_{i}$ are less than $6\eta$ in absolute value. The elimination steps ensure that $\{\hat{\mathcal{N}}_i\}$ form a neighborhood structure, and we only lose matrix elements smaller than $6\eta$.

The error made when estimating $(2V-i\Omega)^{-1}$  by $\operatorname{LI}_{\hat{\mathcal{N}}}(2\hat{V}-i\Omega)$ is given by Lemma~\ref{lem:estpertaprrneigh}, as follows: with $M=(2V-i\Omega)^{-1}$ and $E$ as in Equation~\ref{analyticform}, if $|2\hat{V}_{k,k'}-2V_{k,k'}|\leq{\zeta}\leq 2\eta$ and $\zeta\leq \frac{1}{4\xi(l)}{\|(2V-i\Omega)^{-1}\|_{\infty}^{-1}}$, then,
\begin{align}
\|M-\operatorname{LI}_{\hat{\mathcal{N}}}(2\hat{V}-i\Omega)\|_{\infty}&\leq 2\max_{i\in [m]}\max_{\delta\in\{0,1\}}{|\hat{\mathcal{N}}_i|}\|B_{i,\delta}-\widetilde{B}_{i,\delta}\|_{\infty} \|M^{-1}\|_{\infty}\|M\|_{\infty}+(2\max_{i\in [m]}{|\hat{\mathcal{N}}_i|}+1)\|E\|_{\infty}\nonumber\\
&+\max_{i\in[m]}\max_{\delta_1\in\{0,1\}}\sum_{j\in \mathcal{N}(l)_i, j\notin \hat{\mathcal{N}}_i,\delta_2\in\{0,1\}}|(M-E)_{2i-\delta_1,2j-\delta_2}|+8\|M\|_{\infty}^{2}(\max_{i}|{\hat{\mathcal{N}}}_i|)^2\zeta.
\end{align}

We can divide the complement of $\hat{\mathcal{N}}_i$ in two parts: a) the vertices that are in $\mathcal{N}(l)_i$ b) The vertices that are not in $\mathcal{N}(l)_i$. Accordingly, we can split $B_{i,\delta}-\widetilde{B}_{i,\delta}$ in two parts: the first part are the rows corresponding to a), the second the rows corresponding to b) and $B_{i,\delta}-\widetilde{B}_{i,\delta}$ can be bounded by the sum of the norms of the two sub-columns. Since the sub-column corresponding to b) is a submatrix of $E$, we can bound its norm by $\|E\|_{\infty}\leq
\eta$. The sub-column corresponding to a) is instead made of at most $2\xi(l)$ entries which are bounded in absolute value by $6\eta$, so its norm is bounded by $12\xi(l)\eta$. Thus, we have $\|B_{i,\delta}-\widetilde{B}_{i,\delta}\|_{\infty}\leq \|E\|_{\infty}+12\xi(l)\eta\leq 13\xi(l)\eta$.
Moreover, for ${j\in \mathcal{N}(l)_i, j\notin \hat{\mathcal{N}}_i,\delta_1,\delta_2\in\{0,1\}}$, $|(M-E)_{2i-\delta_1,2j-\delta_2}|\leq |M_{2i-\delta_1,2j-\delta_2}|+\|E\|_{\infty}\leq 7\eta$, and $\max_{i\in[m]}\max_{\delta_1\in\{0,1\}}\sum_{j\in \mathcal{N}(l)_i, j\notin \hat{\mathcal{N}}_i,\delta_1,\delta_2\in\{0,1\}}|(M-E)_{2i-\delta_1,2j-\delta_2}|\leq 14\xi(l)\eta$.

Putting all these bounds together, we obtain after simplification
\begin{align*}
&\|M-\operatorname{LI}_{\hat{\mathcal{N}}}(2\hat{V}-i\Omega)\|_{\infty}\\
&\qquad\qquad \leq 
 26\xi(l)^2 \eta\|S\|_\infty^4\frac{1-e^{-2d_{\max}}}{1-e^{-2d_{\min}}}e^{2d_{\max}}+(16\xi(l)+1)\eta+2\xi(l)^2\zeta\|S\|^4_{\infty}e^{4d_{\max}}(1-e^{-2d_{\max}})^2\\
 &\qquad\qquad \le 50\, \xi(l)^2\eta\|S\|_\infty^4{\frac{1-e^{-2d_{\max}}}{1-e^{-2d_{\min}}}e^{2d_{\max}}}+2\xi(l)^2\zeta\|S\|^4_{\infty}e^{4d_{\max}}\,,
\end{align*}
where in the last line we used that $\xi(l),\|S\|_\infty\ge 1$. Therefore, we can make the above term as small as we want if we can choose $l$ as a function of $\eta$ such that $\lim_{\eta\rightarrow 0}\xi(l(\eta))\eta =0$, such that all the above assumptions also hold. $\zeta$ is then computed as a function of $\xi(l(\eta))$.
Say that we want to ensure $50\, \xi(l)^2\eta\|S\|_\infty^4{\frac{1-e^{-2d_{\max}}}{1-e^{-2d_{\min}}}e^{2d_{\max}}}\leq \epsilon'$. Recall that from Equation \eqref{conditionESdd} we also need to ensure $\|E\|_\infty\|S\|_\infty^4{\frac{1-e^{-2d_{\max}}}{1-e^{-2d_{\min}}}e^{2d_{\max}}} \leq \eta$. 

The former condition can be satisfied if we take 
$$\eta= \frac{\epsilon'}{e^{2d_{\max}}50\, \Delta^{2l}\|S\|_\infty^4}{\frac{1-e^{-2d_{\min}}}{1-e^{-2d_{\max}}}}
\leq \frac{\epsilon'}{e^{2d_{\max}}50\, \xi(l)^{2}\|S\|_\infty^4}{\frac{1-e^{-2d_{\min}}}{1-e^{-2d_{\max}}}}\,,$$
 since $\xi(l)\leq \Delta^l$, while the latter condition is implied by
\begin{align}
\|E\|_\infty\le \frac{\|S\|_{\infty}^2}{{2}}\left(\frac{(2d_{\max})^{l+1}e^{2d_{\max}}}{(l+1)!}\right)&\overset{!}{\le}\frac{\epsilon'}{e^{4d_{\max}}50\, \Delta^{2l}\|S\|_\infty^8}{\left(\frac{1-e^{-2d_{\min}}}{1-e^{-2d_{\max}}}\right)^2}\\&{\leq} \frac{\epsilon'}{e^{4d_{\max}}50\, \xi(l)^2\|S\|_\infty^8}{\left(\frac{1-e^{-2d_{\min}}}{1-e^{-2d_{\max}}}\right)^2}\,,
\end{align}
where we further used the bound on $\|E\|_\infty$ derived in Lemma \ref{th:seriesapprox}. This is ensured by taking (Lemma~\ref{lem:Lambert})
\begin{align*}l&=\left\lfloor(2\Delta^2 d_{\max} e)\exp{\left(W\left(\frac{\log\frac{C(S,d_{\min},d_{\max})}{\eps'}}{2\Delta^2 d_{\max }e}\right)\right)}\right\rfloor,
\end{align*}
 where we denoted $C(S,d_{\min},d_{\max}):=25\|S\|_\infty^{10}e^{6d_{\max}}{\left(\frac{1-e^{-2d_{\max}}}{1-e^{-2d_{\min}}}\right)^2}$. Then, $2\xi(l)^2\zeta\|S\|^4_{\infty}e^{4d_{\max}}\leq \epsilon'$ {implies that $\|M-\operatorname{LI}_{\hat{\mathcal{N}}}(2\hat{V}-i\Omega)\|_\infty\le 2\epsilon'$. Moreover, assuming $\eta\leq 1$, $\zeta\leq 2\eta$, $\zeta\le\frac{1}{8\xi(l)}\|(2V-i\Omega)^{-1}\|_\infty^{-1}\le \frac{1}{4\xi(l)}\|(2V-i\Omega)^{-1}\|_\infty^{-1}$, and ${\zeta\leq \frac{\|(2V-i\Omega)^{-1}\|_\infty^{-1}\eta}{8\xi(l)}}$ are ensured by simply the last inequality, which is implied by taking ${\zeta\leq \frac{e^{-2d_{\max}}\eta}{4\|S\|_\infty^2\xi(l)}}$, as a consequence of \eqref{eq:minsing1} {(this also makes the assumption $\eta\leq 1$ verified)}. All of these inequalities are satisfied by choosing
\begin{align*}
\zeta&={ \frac{\epsilon' e^{-4d_{\max}}}{200\, \Delta^{3l}\|S\|_\infty^6}\frac{1-e^{-2d_{\min}}}{1-e^{-2d_{\max}}}}\,,
\end{align*}
with $C'(S,d_{\max}):=2\|S\|_\infty^4e^{4d_{\max}}$.} Now, we want $\epsilon'$ to be small enough so that $\|\hat{H}-H\|_{\infty}\leq \kappa/2$, where
we compute $\hat{H}$ as in Equation~\ref{eq:locinvtoh} from $\operatorname{LI}_{\hat{\mathcal{N}}}(2\hat{V}-i\Omega)$. As in Remark~\ref{rem:hbound} we can take $\epsilon'=\min\left(\kappa\frac{e^{-2d_{\max}}\|S\|^{-2}_{\infty}}{16\left(1+{e^{4d_{\max}}\|S\|_{\infty}^4(1-e^{-2d_{\max}})}\right)}, \frac{e^{-2d_{\max}}\|S\|_{\infty}^{-2}}{8}\right)$.
Since now we ensured that $\|\hat{H}-H\|_{\infty}\leq \kappa/2$, we can learn the correct edges. By Lemma \ref{lem:cov_matrix}, the number of copies we used to get $\hat{V}$ satisfying $|\hat{V}_{k,k'}-V_{k,k'}|\le \zeta/2$ is

\begin{align}
N\ge  \frac{2^{14} (\|S\|^2_{\infty}\frac{1}{1 -e^{-2d_{\min}}}+\max_i|t_i|)^2(1+\max_i|t_i|)^2}{\zeta^2}\log\left(\frac{4m+1}{\delta}\right).
\end{align}

{The theorem statement is then obtained by noting that by substituting $\|S\|_{\infty}$, $d_{\max}$ and $d_{\min}$ with their bounds the assumptions we checked to prove the theorem are still satisfied.}

This concludes the analysis of the sample complexity. Let us now analyse the computational complexity of the algorithm.
First, we iterate over all possible vertices and over all candidate neighborhoods of size $\xi(l)$. There are $\binom{m}{\xi(l)}$ of these. For each of these, we then iterate over all remaining subsets of size $\xi(l)$, of which there are {$\binom{m-\xi(l)}{\xi(l)}$}, and for each of them we need to invert a matrix of size $2\xi(l)$, which takes time $\mathcal{O}(\xi(l)^3)$. Then checking if the threshold condition is met takes time $\xi(l)^2$. We conclude that iterating over all possible neighborhoods and remaining candidates takes the time claimed in Equation~\eqref{equ:time_complexity_algo}, $ \mathcal{O}\left(m\binom{m-\xi(l)}{\xi(l)}\binom{m}{\xi(l)}\xi(l)^3\right)$. It is then easy to check that the remaining steps only have a subleading order contribution.
\end{proof}
\begin{rem}
For the sample complexity, assuming all relevant parameters of constant order except the threshold $\kappa$, we see that the number of samples required will be of order $\mathcal{O}(\kappa^{-2-\gamma}\log(m\delta^{-1}))$ for any $\gamma>0$ and, crucially, is only logarithmic in the number of modes.
Regarding the computational complexity, as in the worst case we have $\xi(l)\leq \Delta^l$ and given our bound on $l$, the worst-case computational complexity of our algorithm is polynomial in $m$ with the exponent depending {on single parameters, when the other are fixed, as $e^{O(\Delta\log \Delta)}$, $e^{O(d_{\max})}$, $o\left(\frac{1}{d_{\min}^{\gamma}}\right)$, $o\left(\frac{1}{\kappa^{\gamma}}\right)$, $\forall \gamma>0$.} Again, assuming these parameters to be of constant order, we see that the graph learning algorithm runs in polynomial time.
\end{rem}
\begin{rem}
As was the case for Hamiltonian learning, we can reduce the sample complexity of the graph learning protocol to scale like $\tilde O(\kappa^{-2})$ if we are willing to assume that the neighborhoods only grow polynomially in distance. However, this might be an assumption that is more difficult to verify or justify when the graph is unknown.
\end{rem}

\section{Data availability}

This study did not generate or analyze any datasets. 

\section{Code availability}

The code used to generate the figures in this paper is availabe at \href{https://github.com/dsfranca/QuantumGaussianLearning}{this Github repository}.

\section{Acknowledgements}
We thank Antonio Anna Mele for informing us of an issue with the proof of concentration of the covariance estimator in the first version of the paper.
M. F. is supported by the European Research Council (ERC) under Agreement 818761 and by VILLUM FONDEN via the QMATH Centre of Excellence (Grant No. 10059). M.F. was previously supported by a Juan de la Cierva Formaci\'on fellowship (Spanish MCIN project FJC2021-047404-I), with funding from MCIN/AEI/10.13039/501100011033 and European Union NextGenerationEU/PRTR, by European Space Agency, project ESA/ESTEC 2021-01250-ESA, by Spanish MCIN (project PID2022-141283NB-I00) with the support of FEDER funds, by the Spanish MCIN with funding from European Union NextGenerationEU (grant PRTR-C17.I1) and the Generalitat de Catalunya, as well as the Ministry of Economic Affairs and Digital Transformation of the Spanish Government through the QUANTUM ENIA ``Quantum Spain'' project with funds from the European Union through the Recovery, Transformation and Resilience Plan - NextGenerationEU within the framework of the "Digital Spain 2026 Agenda". C.R. acknowledge financial support from the ANR project QTraj (ANR-20-CE40-0024-01) of the French National Research Agency (ANR) as well as France 2030 under the French National Research Agency award number “ANR-22-PNCQ-0002”. D.S.F. acknowledges financial support from the Novo Nordisk
Foundation (Grant No. NNF20OC0059939 Quantum for Life) and by France 2030 under the French National Research Agency award number “ANR-22-PNCQ-0002”. 

\bibliographystyle{alpha}
\bibliography{tomography}

\appendix
\section{Symplectic diagonalization and bounds on singular values}\label{sec:symplectic_diag}

In the following, we denote as $\sigma_{\max}(A)$ the maximum singular value of $A$, (which coincides with the operator norm $\|A\|_{\infty}$) and as $\sigma_{\min}(A)$ the minimum singular value of $A$. The following fact is useful. 

\begin{lem}\label{lem:bound}
For any (symmetric) $2m\times 2m$ Hamiltonian $H>0$, let a symplectic diagonalization be
\begin{equation}\label{eqsymplect}
H=S^{-\intercal}D{S^{-1}}\,,
\end{equation}
with $D=\oplus_{i=1}^m d_i I_{M_2}$, with $d_i>0$, and $S\in\mathrm{Sp}(2m,\mathbb{R})$. Let us also denote as $d_{\min}$ and $d_{\max}$ the smallest and the largest symplectic eigenvalue of $H$, respectively. Then, the corresponding covariance matrix of the Gaussian state with Hamiltonian $H$ is
\begin{equation}\label{eq:SdecV}
V=S f(D) S^{\intercal},
\end{equation}
with $f(x)=\frac{1}{2}\coth(x)=\frac{1}{2}\frac{1+e^{-2x}}{1-e^{-2x}}$, 
and the following bounds hold for $t\in[0,\infty)$:
\begin{align}
\left\|2i\Omega V+\frac{t-1}{t+1}I\right\|_{\infty}&\leq \|S\|_{\infty}^2\frac{2}{1-e^{-2d_{\min}}},\label{eq:normbound1}\\
\sigma_{\min}\left(2i\Omega V+\frac{t-1}{t+1}I\right) &\geq \frac{2}{ \|S\|_{\infty}^2}\frac{e^{-2d_{\max}}}{1-e^{-2d_{\max}}},\label{eq:minsing1}\\
\sigma_{\min}\left(\frac{t}{t+1}\frac{2i\Omega V+I}{2i\Omega V-I}+\frac{1}{t+1}\right)&\geq \frac{e^{-2d_{\max}}}{ \|S\|_{\infty}^2}.\label{eq:minsing2}
\end{align}
\end{lem}
\noindent  
It is worth to stress that  $\|S\|_{\infty}$ does not depend on the temperature and it is equal to $1$ if there is no squeezing.
 
\begin{proof}

Equation~\ref{eq:SdecV} is a simple consequence of Equation~\ref{eq:HtoV}. More precisely,
\begin{align*}
V&=\frac{i}{2}\Omega \operatorname{coth}(iH\Omega)\\
&=\frac{i\Omega}{2}\operatorname{coth}(iS^{-\intercal}DS^{-1} \Omega)\\
&=\frac{i\Omega}{2}\operatorname{coth}(iS^{-\intercal }D\Omega S^\intercal)\\
&=\frac{i\Omega}{2}S^{-\intercal}\operatorname{coth}(iD\Omega )S^\intercal\\
&=S\frac{i\Omega}{2}\operatorname{coth}(iD\Omega)S^\intercal\,.
\end{align*}
By denoting by $U$ the matrix diagonalizing $iD\Omega$, i.e. $U(iD\Omega) U^{-1}=\bigoplus_{j}\operatorname{diag}(d_j,-d_j)$ and since $\operatorname{coth}$ is odd, we get
\begin{align*}
\operatorname{coth}(iD\Omega)=U^{-1}\bigoplus_j\operatorname{coth}(\operatorname{diag}(d_j,-d_j))U=U^{-1}\oplus_j\operatorname{coth}(d_j)\operatorname{diag}(1,-1)U=\bigoplus_j\operatorname{coth}(d_j) i\begin{pmatrix} 0&1\\-1&0\end{pmatrix}\,,
\end{align*}
and hence
\begin{align*}
V=\frac{1}{2}S \bigoplus_i \operatorname{coth}(d_j)i^2\begin{pmatrix}0&1\\-1&0\end{pmatrix}^2 \,S^\intercal=\frac{1}{2}S \operatorname{coth}(D)S^\intercal\equiv S f(D)S^\intercal\,.
\end{align*}
As a consequence, we obtain that for any $t\ge 0$

\begin{align}
2i\Omega V+\frac{t-1}{t+1}I&=i \Omega \left(2V+\frac{t-1}{t+1} i \Omega \right)=i\Omega S \left(2f(D)+\frac{t-1}{t+1} i\Omega\right)S^{\intercal}.
\end{align}
Now, notice that $2f(D)+\frac{t-1}{t+1}i\Omega$ is Hermitian, and

\begin{align}
\left(2f(D)+i\frac{t-1}{t+1}\Omega\right)^2&=\bigoplus_{i=1}^m \left(2f(d_i){-}\frac{t-1}{t+1}\right)^2\ketbra{0_y}+\left(2f(d_i){+}\frac{t-1}{t+1}\right)^2\ketbra{1_y}.
\end{align}

\noindent where $\ket{0_y}$ and $\ket{1_y}$ are the eigenvectors of the Pauli-Y matrix. As a consequence, since $2f(d_i)\geq 1$, $\left|\frac{t-1}{t+1}\right|\leq 1$, and $f$ is monotone decreasing, the largest eigenvalue is smaller than $(2f(d_{\min})+1)^2=\left(\frac{2}{1-e^{-2d_{\min}}}\right)^2$ and the smallest eigenvalue is larger than $(2f(d_{\max})-1)^2=\left(2\frac{e^{-2d_{\max}}}{1-e^{-2d_{\max}}}\right)^2$, that is:
\begin{align}
\left(2\frac{e^{-2d_{\max}}}{1-e^{-2d_{\max}}}\right)^2 I\le \left(2f(D)+\frac{t-1}{t+1}i\Omega\right)^2&\leq \left(\frac{2}{1-e^{-2d_{\min}}}\right)^2 I\,.\label{lowerupperbounds}
\end{align}
Hence the following operator inequalities hold:
\begin{align*}
\left(2i\Omega V+\frac{t-1}{t+1}I\right)^{\dagger}\left(2i\Omega V +\frac{t-1}{t+1}I\right)&=S \left(2f(D)+\frac{t-1}{t+1}i\Omega\right)S^{\intercal}S\left(2f(D)+\frac{t-1}{t+1}i\Omega\right)S^{\intercal}\\
&\leq \|S\|^2_{\infty} S \left(2f(D)+\frac{t-1}{t+1}i\Omega\right)^2S^{\intercal} \\
&\leq \|S\|^2_{\infty} \left(\frac{2}{1-e^{-2d_{\min}}}\right)^2 SS^{\intercal} \\
&\leq \|S\|^4_{\infty} \left(\frac{2}{1-e^{-2d_{\min}}}\right)^2 I\,.
\end{align*}
This proves Equation~\eqref{eq:normbound1}. Morever, since $S$ is symplectic, $\sigma_{\min}(S)=\sigma^{-1}_{\max}(S)=\|S\|^{-1}_{\infty}$, and $S^{\intercal}S\geq \sigma_{\min}^2(S) I=\|S\|_{\infty}^{-2}I$, and similarly $SS^{\intercal}\geq \|S\|_{\infty}^{-2}I$. 
Using this fact together with the lower bound in \eqref{lowerupperbounds}, we directly get that
\begin{align*}
\left(2i\Omega V+\frac{t-1}{t+1}I\right)^{\dagger}\left(2i\Omega V +\frac{t-1}{t+1}I\right)
\geq \|S\|^{-4}_{\infty} \left(\frac{2e^{-2d_{\max}}}{1-e^{-2d_{\max}}}\right)^2 I,
\end{align*}
 proving Equation~\eqref{eq:minsing1}. Finally, we consider
\begin{align*}
g(D):=(2f(D)+i\Omega)(2f(D)-i\Omega)^{-1}&=\bigoplus_{i=1}^m \frac{2f(d_i){-}1}{2f(d_i){+}1}\ketbra{0_y}+\frac{2f(d_i){+}1}{2f(d_i){-}1}\ketbra{1_y}\\
&=\bigoplus_{i=1}^m e^{{-}2d_i}\ketbra{0_y}+e^{2d_i}\ketbra{1_y},
    \end{align*}
from which it follows that $g(D)\geq 0$. Thus we have
\begin{align*}
\left(\frac{2i\Omega V+I}{2i\Omega V-I}\right)&=i\Omega S (2f(D)+i\Omega)S^{\intercal}{S^{\intercal}}^{-1}(2f(D)-i\Omega)^{-1}S^{-1}i \Omega\\
&=i\Omega S (2f(D)+i\Omega)(2f(D)-i\Omega)^{-1}S^{-1}i \Omega\\
&=i\Omega S g(D)S^{-1}i \Omega,
\end{align*}
and hence
\begin{align*}
\frac{t}{t+1}\frac{2i\Omega V+I}{2i\Omega V-I}+\frac{1}{t+1}&=i\Omega S \left(\frac{t g(D)+I}{t+1}\right)S^{-1}i \Omega\,.
\end{align*}
It follows that
\begin{align*}
\left(\frac{t}{t+1}\frac{2i\Omega V+I}{2i\Omega V-I}+\frac{1}{t+1}\right)^{\dagger}&\left(\frac{t}{t+1}\frac{2i\Omega V+I}{2i\Omega V-I}+\frac{1}{t+1}\right)\\&=i\Omega S^{-\intercal} \left(\frac{t g(D)+I}{t+1}\right)S^{\intercal}S\left(\frac{t g(D)+I}{t+1}\right)S^{-1}i\Omega\\
&\geq \|S\|_{\infty}^{-2} i\Omega S^{-\intercal} \left(\frac{t g(D)+I}{t+1}\right)^2S^{-1} i\Omega\\
&\geq \|S\|_{\infty}^{-2} e^{-4d_{\max}} (i\Omega)S^{-\intercal}S^{-1} (i\Omega)\\
&\geq \|S\|^{-4}_{\infty} e^{-4d_{\max}}  I,
\end{align*}
proving Equation~\ref{eq:minsing2}.

\end{proof}

{
\section{Energy bounds}\label{sec:energybounds}
For a Gaussian state $\rho$, the energy is
\begin{equation}\label{eq:energy}E:=\mathbb{E}[E_m]=\frac{\Tr[V[\rho]]}{2}+\frac{\|t[\rho]\|_2^2}{2}.
\end{equation}
Via the relation in Equation~(\ref{eq:SdecV}), i.e. 
\begin{equation}
V=S f(D) S^{\intercal},
\end{equation}
with $f(x)=\frac{1}{2}\coth(x)=\frac{1}{2}\frac{1+e^{-2x}}{1-e^{-2x}}\geq \frac{(1-e^{-2x})^{-1}}{2}\geq \frac{1}{2}$, we also have

\begin{align}
\Tr[V]&\geq\Tr[Sf(D)S^{\intercal}]\geq \Tr[SS^{\intercal}]/2\geq \|S\|_{\infty}^2/2,\\
\Tr[V]&\geq \Tr[Sf(D)S^{\intercal}]=\Tr[SS^{\intercal}f(D)]\geq \|S^{-1}\|_{\infty}^{2}\Tr[f(D)]\geq \|S\|_{\infty}^{-2}(1-e^{-2d_{\min}})^{-1},
\end{align}

where in the first derivation we used $f(x)\geq 1/2$ and in the second that $\|S\|_{\infty}=\|S^{-1}\|_{\infty}$ for symplectic matrices. This also implies 

\begin{equation}
(1-e^{-2d_{\min}})^{-1}\leq 2E^2.
\end{equation}
}

\section{Family of Hamiltonians with condition number scaling with the size of the graph}~\label{condnumbex}

The following example shows why one one should expect a dependence on the condition number for results with sample complexity logarithmic in the number of modes, when additive error guarantees are used. The precision matrices of size $m\times m$
\begin{equation*}
H_m=\begin{pmatrix}
1 & -1 & 0 & \cdots & 0 & 0 \\
-1 & 2 & -1 & \cdots & 0 & 0 \\
0 & -1 & 2 & \cdots & 0 & 0 \\
\vdots & \vdots & \vdots & \ddots & \vdots & \vdots \\
0 & 0 & 0 & \cdots & 2 & -1 \\
0 & 0 & 0 & \cdots & -1 & 2
\end{pmatrix}
\end{equation*}
have inverses

\begin{equation*}
H_m^{-1}=\begin{pmatrix}
m & m-1 & m-2 & \cdots & 2 & 1 \\
m-1 & m-2 & m-3 & \cdots & 2 & 1 \\
m-2 & m-3 & m-4 & \cdots & 2 & 1 \\
\vdots & \vdots & \vdots & \ddots & \vdots & \vdots \\
2 & 2 & 2 & \cdots & 2 & 1 \\
1 & 1 & 1 & \cdots & 1 & 1
\end{pmatrix}\,,
\end{equation*}
meaning that to learn the entries of the covariance matrix with additive error, one expects that $\mathrm{poly}(m)$ samples are needed, as the variance of the random variables is polynomial in $m$. Thus, already in the classical case, estimating the entries of the covariance matrix to constant precision requires polynomial samples, so that with logarithmic number of samples we do not expect to be able to estimate the entries of the Hamiltonian to constant precision, even if its entries are all of constant order. Furth, the condition number of $H_m$, and thus the norm of $H_{m}^{-1}$, is of order $m^2$. 


\end{document}